%% file: xbytm.tex
\def\cmntsoff{}
\documentclass[aps,pra,longbibliography,superscriptaddress,12pt]{revtex4-1}

\usepackage{amsmath,amsfonts,amssymb,amsthm}
\usepackage{mathtools}
\usepackage{setspace}
\usepackage[normalem]{ulem}
\usepackage{stmaryrd}
\usepackage{expdlist}

\usepackage[nolists,heads,nomarkers]{endfloat}

\usepackage{graphicx,xcolor}
\usepackage{hyperref}
\hypersetup{pdfstartview=}
\usepackage{url}

\usepackage{verbatim}
\usepackage{alltt}
\usepackage{moreverb}

\usepackage{xr} 

\input{xbytm_defs}

\begin{document}

\title{Quantum Randomness Generation by Probability Estimation with
  Classical Side Information}

\author{Emanuel Knill}
\affiliation{National Institute of Standards and Technology, Boulder, Colorado 80305, USA}
\affiliation{Center for Theory of Quantum Matter, University of Colorado, Boulder Colorado 80309, USA}
\author{Yanbao Zhang}
\affiliation{NTT Basic Research Laboratories, NTT Corporation, 3-1
  Morinosato-Wakamiya, Atsugi, Kanagawa 243-0198, Japan}
\affiliation{NTT Research Center for Theoretical Quantum Physics, NTT Corporation, 
  3-1 Morinosato-Wakamiya, Atsugi, Kanagawa 243-0198, Japan}
\author{Peter Bierhorst}
\affiliation{National Institute of Standards and Technology, Boulder, Colorado 80305, USA}

\begin{abstract}
  We develop a framework for certifying randomness from Bell-test  trials based on directly estimating the probability of the  measurement outcomes with adaptive test supermartingales. The number  of trials need not be predetermined, and one can stop performing  trials early, as soon as the desired amount of randomness is  extractable. It can be used with arbitrary, partially known and  time-dependent probabilities for the random settings  choices. Furthermore, it is suitable for application to experimental  configurations with low Bell violation per trial, such as current  optical loophole-free Bell tests. It is possible to adapt to  time-varying experimental parameters.  We formulate the framework  for the general situation where the trial probability distributions  are constrained to a known set. Randomness expansion with  logarithmic settings entropy is possible for many relevant  configurations.  We implement probability estimation numerically and  apply it to a representative settings-conditional probability  distribution of the outcomes from an atomic loophole-free Bell test  [Rosenfeld et al., Phys. Rev. Lett. 119:010402 (2017),  arXiv:1611.04604 (2016)] to illustrate trade-offs between the amount  of randomness, error, settings entropy, unknown settings biases, and  number of trials. We then show that probability estimation yields  more randomness from the loophole-free Bell-test data analyzed in  [Bierhorst et al., arXiv:1702.05178 (2017)] and tolerates  adversarial settings probability biases.
\end{abstract}

\maketitle
\tableofcontents

\section{Overview}
\label{sec:intro}

\subsection{Introduction}

Device-independent quantum randomness generation exploits the fact
that there are quantum correlations with measurement outcomes that are
necessarily non-deterministic with respect to prior side information.
This non-determinism is ensured by non-signaling
constraints that can be enforced by causal separation of the relevant
events and certified by tests based on freely chosen measurement
settings. The certified randomness can
then be extracted by means of classical algorithms.
Device-independence means that the physical devices producing and
measuring the quantum correlations can be obtained from untrusted
manufacturers, without affecting the desired randomness
properties. Device-independent quantum randomness generation was
introduced in Colbeck's thesis~\cite{colbeck:2007} and has since been
studied by many researchers. See Refs.~\cite{brunner:qc2014a,acin:qc2016a} for reviews.

The types of randomness generation protocols are diverse and depend on
the context and assumptions. In general, a randomness generation
protocol produces a string of $\sigma$ bits certified to be uniformly
random within a total variation distance $\epsilon$. It requires
devices whose behaviors are known to satisfy restrictions based on
physical principles such as speed-of-light limits on communication.
Device-independence also requires some random input bits, usually used
for measurement settings choices, where different assumptions can be
made on their quality and source.  The certification is conditional on
the device-behavior restrictions as well as assumptions on outside
entities with respect to which the string of generated bits is to be
random. For example, it matters whether these entities are classical
or quantum and what type of access they had to the protocol devices in
the past.  This work considers entities holding classical side
information. That outside entities have only classical information can
be justified if they are not quantum-capable, if we built the devices
and the entities never had access to them, or if we verified the
absence of long-term quantum memory in the devices.  See
Sect.~\ref{subsec:ass_belltests}.  We assume that the source of input
bits is random relative to outside entities at the time they last
interacted with the protocol devices.

For device-independent randomness generation, we use experiments
modelled on Bell tests~\cite{Bell}. These tests are designed to show
that there are quantum correlations that cannot be mixtures of locally
determined probability distributions referred to as ``local
realistic'' (LR) distributions. See Ref.~\cite{genovese:qc2005a} for a
review.  A Bell test consists of a sequence of trials (sometimes
called ``rounds'') where two or more stations make measurements on a
shared physical state with randomly chosen measurement settings.  For
full device-independence, each trial needs to satisfy that the
different stations' time intervals between applying the settings to
the devices and determining the measurement outcomes are space-like
separated, preventing any communication between them. Furthermore,
trials must be committed to in advance, so that it is not possible to
postselect them on a success criterion.  Because these conditions
require large separation and/or fast devices as well as
high-efficiency measurements, only recently has it become possible to
perform successful Bell tests satisfying these criteria. Such Bell
tests may be referred to as ``loophole-free'', and the list of
successful loophole-free Bell tests includes ones based on heralded
atom entanglement~\cite{hensen:2015,rosenfeld:qc2016a} and ones
utilizing entangled photon-pairs with high-efficiency
detectors~\cite{giustina:2015,shalm:2015}.

Experimental certified randomness was first demonstrated in
Ref.~\cite{pironio:2010} (see also
Refs.~\cite{fehr:2013,pironio:2013}) with pairs of ions located in
separate traps. This demonstration claimed the presence of $42$ bits
of entropy with an error of $0.01$ in their string of measurement
outcomes, with respect to classical side information and restricting
correlations to quantum achievable ones. Extraction would have reduced
the number of bits produced and increased the error by an amount that
was not determined. Recently, we and our collaborators demonstrated
end-to-end randomness extraction~\cite{bierhorst:qc2017a}, producing
$256$ bits within $0.001$ of uniform from one of the data sets from
the loophole-free Bell test reported in Ref.~\cite{shalm:2015}. These
bits are certified with respect to classical side information and
non-signaling assumptions, which in principle allows for super-quantum
correlations.  Extracting randomness from today's optical
loophole-free Bell tests required the theoretical advances in
Ref.~\cite{bierhorst:qc2017a} to deal with the fact that each trial
demonstrates very little violation of Bell inequalities. Previous
works were not sufficient for certifying entropy without increasing
the number of trials by orders of magnitude. A specific comparison is
in Ref.~\cite{bierhorst:qc2017a}.

A benefit of the theory developed in Ref.~\cite{bierhorst:qc2017a} is
that it allows for an adaptive protocol that can track changes in the
trial statistics during the protocol. This is helpful in
current experiments, where we find that
measurable drifts in parameters can wipe out a randomness
certificate if not accounted for. The fact that the protocol can adapt
is inherited from its use of the ``probability-based ratio'' protocol
for obtaining $p$-value bounds against local realism (LR) in Bell
tests~\cite{zhang:2011,zhang:2013}.  Here we develop a different class
of randomness generation protocols based on ``probability
estimation.''  Probability estimation involves obtaining
high-confidence-level upper bounds on the actual probability of the
measurement outcomes given the known constraints on the
distributions. We show that randomness generation can be reduced to
probability estimation. Since probability estimation is a statistical
estimation problem, we then take advantage of the theory of test
supermartingales~\cite{shafer:qc2009a} to bypass the framework of Bell
inequalities and directly determine probability estimators expressed
as products of \emph{probability estimation factors} (PEFs). PEFs are
functions of a trial's settings and outcomes and provide a way of
multiplicatively accumulating probability estimates trial-by-trial.
While relationships between PEFs and Bell inequalities exist,
characteristic measures of quality for Bell inequalities, such as
violation signal-to-noise, winning probability or 
statistical strength for rejecting LR, are not good measures of PEF performance.  We
develop tools for obtaining PEFs. In particular, we show that when the
distributions of settings and outcomes are constrained to a convex
polytope, PEFs can be effectively optimized with convex optimization
over the polytope given one parameter, the ``power'' (see
Def.~\ref{def:pef}).  The optimization can explicitly take into
account the number of trials and the error goal.  In the limit where
the power parameter goes to zero, an asymptotic rate is obtained that
can be interpreted as the optimal rate for producing entropy for
random bits.  This generalizes and improves on min-entropy estimators
for Bell configurations described in works such as
Refs.~\cite{nieto:2014,bancal:2014,nieto-silleras:qc2016a}, which are
optimal for single-trial min-entropy estimation.

In a large class of situations including the standard Bell-test
configurations, PEFs directly lead to exponential expansion of input
randomness, as expected from previous
works~\cite{vazirani:qc2012b,miller_c:qc2014a,miller_c:qc2014b}, which
prove exponential expansion with more trials or worse settings entropy
than PEFs but secure against quantum side information. We prove that
asymptotically, the settings entropy can be logarithmic in the output
entropy. This is the best result so far for randomness expansion
without using a cross-feeding protocol, which can accomplish infinite
expansion~\cite{coudron:2014,chung:2014}.  To accomplish exponential
expansion we use highly biased settings distributions.  We point out
that it is not necessary to have independent and identically
distributed (i.i.d.) settings choices.  In particular, if the settings
are obtained by choosing, ahead of time, a random ``test'' trial among
a block of $2^{k}$ trials, we can achieve good expansion while eliminating
the need for decompressing uniformly random input bits into a stream
of highly biased and independent ones.

As a demonstration of the power of probability estimation, we show how
it would perform on a representative example for distributions
achieved in loophole-free Bell experiments with atoms based on
heralded entanglement~\cite{rosenfeld:qc2016a}. We then apply it to
the main data set from Ref.~\cite{shalm:2015} analyzed in
Ref.~\cite{bierhorst:qc2017a}, showing that we could improve the
amount of randomness extracted substantially, while ensuring that the
certificates are valid even if the input randomness is biased, a
problem that was noticed and accounted for in the report on the
loophole-free Bell test for this experiment~\cite{shalm:2015}.
Finally, we reanalyze the data from the first demonstration of
certified experimental randomness in a 
Bell test free of the detection loophole but subject to the locality 
loophole, which was based on ions~\cite{pironio:2010}. We demonstrate 
significantly more randomness than reported in this reference. 
These examples demonstrate that probability estimation
is a practical way for implementing entropy production in randomness
generation.

Our framework is in the spirit of the entropy accumulation framework
of Ref.~\cite{arnon-friedman:2018}, but takes
advantage of the simplifications possible for randomness generation
with respect to classical side information. In particular, the outside
entities have no interaction with the protocol devices after the
protocol starts, and the framework can be cast purely in terms of
random variables without invoking quantum states. This avoids the
complications of a full representation of the protocol in terms of
quantum processes. With these simplifications, our framework applies
to any situation with known constraints on past-conditional
probability distributions and accumulates the probability estimates
trial-by-trial. We interpret these estimates as smooth min-entropy
estimates, but prefer to certify the extracted randomness directly. 

In the entropy accumulation framework, the relevant estimators, called
min-tradeoff functions, must be chosen before the protocol, and the
final certificate is based on the sum of statistics derived from these
functions. Finding suitable min-tradeoff functions is in general
difficult.  In the probability estimation framework, probability
estimators can be adapted and accumulate multiplicatively.  For
relevant situations, PEFs are readily obtained and the tradeoff
between randomness and error can be effectively optimized.

The analog of
min-tradeoff functions in the probability estimation framework are
entropy estimators.  We show that logarithms of PEFs are
proportional to entropy estimators, and essentially all entropy
estimators are related to PEFs in this way.  In this sense, there
is no difficulty in finding entropy estimators.  However, PEFs are
more informative for applications, so except for illuminating
asymptotic behavior, there is little to be
gained by seeking entropy estimators directly.

A feature of entropy accumulation is optimality of asymptotic rates
for min-tradeoff functions. Probability estimation also achieves
optimal asymptotic rates. In both cases, the tradeoff between
error and amount of randomness makes these asymptotic rates less
relevant, which we demonstrate for probability estimation on
the Bell-test examples mentioned above. Furthermore, the framework can be used with any
randomness-generating device with randomized measurements to verify
the behavior of the device subject to trusted physical constraints.
The only requirement is that the constraints can be formulated as
constraints on the probability distributions and are
sufficiently strong to allow for randomness certification.  

The remainder of the paper is structured as follows. 
We summarize the main results in Sect.~\ref{s:somr}.  We lay out our
notation and define the basic concepts required for the probability
estimation framework in Sect.~\ref{sec:prelims}.  This section
includes introductions to less familiar material on classical smooth
min-entropies, test martingales and the construction of test
  martingales from Bell inequalities.  In Sect.~\ref{sec:probest}, we define exact and soft
probability estimation and show how randomness generation can be
reduced to probability estimation. The measurement outcomes can be fed
into appropriate randomness extractors, where the number of
near-uniform random bits is naturally related to the probability
estimate. We give three protocols that compose
probability estimators with randomness extractors. The
first is based on general relationships between probability estimation
and smooth min-entropy and reprises techniques from
Refs.~\cite{konig:2008,pironio:2013,bierhorst:qc2017a}.  The second
relies on banked randomness to avoid the possibility of protocol
failure. The third requires extractors that are well-behaved on
uniform inputs to enable a direct analysis of the composition.
Although we do not demonstrate an end-to-end randomness-generation
protocol including extraction here, our goal is to provide all the
information needed for implementing such a protocol in future work,
with all relevant constants given explicitly.  Sect.~\ref{sec:prots}
shows how to perform probability estimation for a sequence of trials
by means of implicit test supermartingales defined adaptively. The
main tool involves PEFs to successively accumulate probability
estimates.  The main results involve theorems showing how PEFs can be
``chained'' to form probability estimates.  PEFs are readily
constructed for distribution constraints relevant in Bell-test
configurations. We proceed to an exploration of basic PEF properties
in Sect.~\ref{sec:pef_props}, where we find that there is a close
relationship between the rates for PEFs and those for a class of
functions called ``entropy estimators'', which are the analog of
min-tradeoff functions in our framework.  We establish that for error
that is $e^{-o(n)}$, the achievable asymptotic rates are
optimal. Next, in Sect.~\ref{sec:pef_constructions}, we consider a
family of PEFs constructed from Bell functions whose expectations
bound maximum conditional probabilities in one trial. In Sect.~\ref{sec:reduce_entropy} we show that this
family can be used for exponential expansion by means of highly
  biased settings choices. Given a constant error
bound, the settings distribution can be interpreted as a random choice
of a constant (on average) number of test trials with uniformly
  random settings, where the remaining trials have fixed settings.
We note that this is a theoretical proof-of-principle of exponential
expansion. In practice, we prefer to numerically optimize the PEFs with respect to the desired
error and calibrated experimental distribution.  The final section
Sect.~\ref{sec:apps} explores the three examples mentioned above.

\subsection{Summary of Main Results}\label{s:somr}

This manuscript aims to establish the foundations of probability
estimation and contains a large number of mathematical results based
on mathematical concepts introduced here.  In this section we
summarize the main results without precise definitions, in more
familiar terms and with less generality than the probability
estimation framework established later.

The context for our work consists of experiments where a sequence of
trials is performed. In each trial, settings are chosen according to a
random variable (RV) $Z$ and outcomes are obtained according to an RV
$C$. The sequences of outcomes and settings obtained in the experiment
are denoted by $\Sfnt{C}$ and $\Sfnt{Z}$, where for $n$ trials,
$\Sfnt{C}=(C_{i})_{i=1}^{n}$ and $\Sfnt{Z}=(Z_{i})_{i=1}^{n}$, with
$C_{i}$ and $Z_{i}$ the $i$'th trial's outcomes and settings. The main
example of such an experiment is the standard Bell test, where there
are two physically separated stations.  In a trial, the stations
randomly select measurement settings $X$ and $Y$ (respectively) and
the stations' devices produce outcomes $A$ and $B$ (respectively).  In
this case $Z=XY$ and $C=AB$.  The physical separation of the stations
and the physics of the measurement devices constrain the distributions
of $ABXY$ to be non-signaling once loopholes are accounted
for. Classical devices are further constrained by LR, which can be
violated by quantum devices.  This violation is associated with
randomness that can be exploited for randomness
generation~\cite{colbeck:2007}, see the introduction above.  Here we
consider randomness generation where the generated bits are random
relative to any external entity holding classical side information
$E$.

The traditional approach to randomness generation is to first derive a
bound on the smooth min-entropy of the outcomes conditional on
settings and $E$ from the statistics of the observed value $\Sfnt{cz}$
of $\Sfnt{CZ}$ given constraints on the joint distribution of
$\Sfnt{CZ}$ and $E$. The smooth conditional min-entropy
$H_{\min,\mu}^{\epsilon}(\Sfnt{C}|\Sfnt{Z}E)$ for the joint
distribution $\mu$ is given by the negative logarithm of the maximum
probability of $\Sfnt{C}$ given $\Sfnt{Z}$ and $E$, averaged over
$\Sfnt{Z}$ and $E$, up to an error bound $\epsilon$, which is the
smoothness parameter.  It can be formally defined as the maximum
$\lambda\geq 0$ such that there exist $\mu'$ and $\nu$ where $\mu'$ is
within total variation distance $\epsilon$ of $\mu$ and
$\mu'(\Sfnt{c}\Sfnt{z}e)\leq 2^{-\lambda}\nu(\Sfnt{z}e)$ for all
$\Sfnt{cz}e$. (A more involved but equivalent definition based on
maximum probabilities is given in
Sect.~\ref{sec:minentropy_extraction},
Def.~\ref{def:smoothminentropy}.) The smooth conditional min-entropy
characterizes the number of near-uniform random bits that can be
extracted from the outcomes with a randomness extractor applied after
obtaining the min-entropy bound.  The distance from uniform of the
random bits obtained is parametrized by an error bound that, in the
simplest case, is the sum of the error bound of the smooth min-entropy
and an error parameter of the extractor used to obtain the random
bits. Here we reduce the problem of obtaining a smooth conditional
min-entropy bound to that of estimating the conditional probability
$\Prob(\Sfnt{c}|\Sfnt{z}e)$ for the observed values $\Sfnt{cz}$,
independent of the value $e$ of $E$.

Constraints on the joint distribution of $\Sfnt{CZ}$ and $E$ are
determined by a statistical model $\cH$ consisting of the allowed
joint probability distributions, which may enforce non-signaling
conditional on $E$ and other constraints such as that the conditional
distributions are quantum achievable by causally separated devices
sharing an initial state. Given $\cH$, a level-$\epsilon$
(conditional) probability estimator for $\cH$ and $\Sfnt{C}|\Sfnt{Z}$
is a function $U:\Sfnt{cz}\mapsto U(\Sfnt{cz})\in[0,1]$ such that for
all $\mu\in\cH$ and all values $e$ of $E$, the probability that
$\Sfnt{CZ}$ takes a value $\Sfnt{cz}$ for which $U(\Sfnt{cz})\geq
\Prob_{\mu}(\Sfnt{C}=\Sfnt{c}|\Sfnt{Z}=\Sfnt{z},E=e)$ is at least
$1-\epsilon$. The probability estimate $U(\Sfnt{cz})$ differs from a
smooth min-entropy estimate in that the quantity being estimated
depends on the data $\Sfnt{cz}$, while the smooth conditional
min-entropy is a characteristic of the overall distribution $\mu$. Our
first result is that one can obtain a smooth min-entropy estimate from
a probability estimate.
\begin{lemma*}
  \emph{(Lem.~\ref{lm:esupe_fail_esmaxprob})} Consider a
  level-$\epsilon$ probability estimator $U$ for $\cH$ and
  $\Sfnt{C}|\Sfnt{Z}$, $\mu\in\cH$ and $\kappa=\Prob_{\mu}(U\leq p)$.
  Define the probability distribution $\mu'$ by
  $\mu'(\Sfnt{cz}e)=\mu(\Sfnt{cz}e|U\leq p)$ for all $\Sfnt{cz}e$. Then
  $H_{\min,\mu'}^{\epsilon/\kappa}(\Sfnt{C}|\Sfnt{Z}E)\geq
  -\log(p/\kappa)$.
\end{lemma*}
We establish this lemma for the larger class of soft probability
estimators, which provide extensions that may be useful in some
applications. In particular, softness enhances adaptability and
enables use of trial information not determined by $\Sfnt{CZ}$.  For
simplicity, we do not consider softening in this overview. 

Here is a sketch of one way to generate randomness from $\Sfnt{C}$
using the lemma above.  First determine a level-$\epsilon^{2}$
probability estimator $U$, then run an experiment to obtain an
instance $\Sfnt{cz}$ of $\Sfnt{CZ}$. If $U(\Sfnt{cz})> p$, the
protocol failed. If not, apply a classical-proof extractor $\cE$ to
$\Sfnt{c}$ with input min-entropy $-\log_{2}(p/\epsilon)$ to produce
random bits. The number of random bits produced can be close to the
input min-entropy.  The definition and properties of extractors are
summarized in Sect.~\ref{sec:minentropy_extraction}.  The parameters
chosen ensure that if the probability of success is at least
$\epsilon$, then conditional on success, the random bits produced are
uniform within $\TV$ distance $\epsilon+\epsilon_{x}$, where
$\epsilon_{x}$ is the extractor error.  In
Sect.~\ref{subsec:protocols}, we provide details for three protocols
for randomness generation from probability estimators that take
advantage of features of probability estimation to improve on the
sketch just given.

We previously developed a powerful method for constructing Bell
functions that optimize the statistical strength for testing LR by
multiplying ``probability based
ratios''~\cite{zhang:2011,zhang:2013}. The method can be seen as an
application of the theory of test
supermartingales~\cite{shafer:qc2009a}.  The definitions and basic
properties of test supermartingales are given in
Sect.~\ref{subsec:supermart}.  Here we show that this theory can be
applied to the problem of constructing probability estimators.  The
basic tool is to construct \emph{probability estimation factors} (PEFs) that
are computed for each trial of an experiment. Let $\cC$ be the trial
model.  A PEF with power $\beta>0$ for $\cC$ and $C|Z$ is a
non-negative function $F:cz\mapsto F(cz)$ such that for all
$\mu\in\cC$ we have $\Exp_{\mu}(F(CZ)\mu(C|Z)^{\beta})\leq 1$. The
fundamental theorem of PEFs is that they can be ``chained'' over
trials to construct probability estimators.  When chaining trials, the
experimental model $\cH$ is constructed from individual trial models
$\cC$ by requiring that each trial's probability distribution
conditional on the past is in $\cC$. The trial models and PEFs may
depend on past settings and outcomes.  Conditioning on settings
requires an additional conditional independence property as specified
in Sect.~\ref{subsect:standardmodels}.  A simplified version of the
fundamental theorem for the case where all trial models and PEFs are
the same is the following:
\begin{theorem*}
  \emph{(Thm.~\ref{thm:uest_constr})} Let $F$ be a PEF with power
  $\beta$ for $\cC$ and $C|Z$ and define
  $T(\Sfnt{CZ})=\prod_{i=1}^{n}F(C_{i}Z_{i})$. Then $1/(\epsilon
  T(\Sfnt{CZ}))^{1/\beta}$ is a level-$\epsilon$ probability estimator
  for $\cH$ and $\Sfnt{C}|\Sfnt{Z}$.
\end{theorem*}
The proof is enabled by martingale theory and requires constructing a
test supermartingale.

From the fundamental theorem for PEFs, we can see that up to
adjustments for probability of success and the error bound, the
min-entropy per trial witnessed by a PEF $F$ with power $\beta$ at
distribution $\mu$ is expected to be the ``log-prob rate''
$\Exp_{\mu}(\log(F(cz)))/\beta$. In a sequence of trials with devices
behaving according to specifications, we can expect the trial
distributions to be approximately i.i.d.  with known distribution
$\mu$. In this case the RVs $\log(F(C_{i}Z_{i}))/\beta$ are also
approximately i.i.d., and their sum is typically close to
$n\Exp_{\mu}(\log(F(cz)))/\beta$. Since the sum determines the
conditional min-entropy witnessed by the probability estimator
obtained from $F$, a goal of PEF construction is to maximize the
log-prob rate. We show that PEF construction and optimization reduces
to the problem of maximizing a concave function over a convex
domain. If the trial model is a convex polytope, the convex domain is
defined by finitely many extreme points and PEF optimization has an
effective implementation (Thm.~\ref{thm:prob_constr}).  For the
standard Bell-test configuration, convex polytopes that include the
model and that have a manageable number of extreme points exist. We
implement and apply PEF optimization in Sect.~\ref{sec:apps}.

Given a trial model and a distribution $\mu$ in the model, the maximum
number of near-uniform random bits that can be produced per trial in
an asymptotically long sequence of trials is given by the minimum
conditional entropy of the outcomes given the settings and $E$. The
minimum is over all distributions $\nu$ of $CZE$ such that
$\nu(CZ|e)\in\cC$ for all $e$ and $\nu(CZ)=\mu(CZ)$. This is a
consequence of the asymptotic equipartition
property~\cite{tomamichel:qc2009a}.  We prove that PEFs can achieve
the maximum rate in the asymptotic limit.
\begin{theorem*}
  \emph{(Thm.~\ref{thm:gainrate_optimality})} For any trial model
  $\cC$ and distribution $\mu\in\cC$ with minimum conditional entropy
  $g$, the supremum of the log-prob rates of PEFs for $\cC$ is $g$.
\end{theorem*}
To prove optimality we define entropy estimators for model $\cC$ as
real-valued functions $K(CZ)$ such that $\Exp_{\mu}(K(CZ))$ is a lower
bound on the conditional entropy for all $\mu\in\cC$. We then show
that there are entropy estimators for which the entropy estimate
$\Exp_{\mu}(K(CZ))$ approaches the minimum conditional entropy, and
that for every entropy estimator, there are PEFs whose log-prob rates
approach the entropy estimate.

It is desirable to minimize the settings entropy used for randomness
generation. For this we consider \emph{maximum probability estimators}
for a model $\cC$, which are defined as functions $G(CZ)$ such that
$\Exp_{\mu}(G(CZ)) \geq \max_{cz}\mu(c|z)$ for all $\mu\in\cC$.
Non-trivial maximum probability estimators exist. For example, every
Bell inequality for the standard two-settings, two-outcomes Bell-test
configuration for two parties has associated maximum probability
estimators $G(CZ)$, and a distribution that violates the Bell
inequality satisfies $\Exp_{\mu}(G(CZ))<1$.  For every maximum
probability estimator $G(CZ)$, we construct a family of PEFs for which
the log-prob rates at $\mu\in\cC$ approach
$-\log(\Exp_{\mu}(G(CZ)))$. We analyze this family of PEFs to
determine how log-prob rates depend on number of trials and power, and
find that with this family it is possible to achieve exponential
expansion of settings entropy.  For this, we consider models
$\cC_{C|Z}$ of distributions of $C$ conditional on $Z$.  An
unconditonal model $\cC$ is obtained by specifying a probability
distribution for the settings.
\begin{theorem*}
  \emph{(Thm.~\ref{thm:expexp})} Let $G(CZ)$ be a maximum probability
  estimator for $\cC$ where $\cC$ is determined by the conditional model
  $\cC_{C|Z}$ with the uniform settings distribution. Assume given
  $\mu\in\cC$ such that $\Exp_{\mu}(G(CZ))<1$.  Then for a constant
  error bound, there exists a family of PEFs and settings probability
  distributions determined by the number of trials $n$ such that the
  settings entropy is $O(\log(n))$ and the smooth conditional
  min-entropy of the outcomes is $\Omega(n)$.
\end{theorem*}
The constants in the construction for exponential expansion are not of
excessive size, but we consider the construction a proof of principle,
not a practical proposal. Given that for relevant configurations, we
can optimize PEFs directly, for finite experiments and the best
expansion, it is preferable to use directly optimized PEFs.  We
determine expansion opportunities for the trial distribution observed
in an atomic Bell test in Sect.~\ref{sec:apps}.

\section{Basic Concepts}
\label{sec:prelims}

\subsection{Notation}

Much of this work concerns stochastic sequences of random variables
(RVs). RVs are functions on an underlying probability space.  The
range of an RV is called its value space.  Here, all RVs have finite
value spaces. We truncate sequences of RVs so that we only consider
finitely many RVs at a time. With this we may assume that the
underlying probability space is finite too. We use upper-case letters
such as $A,B,\ldots,X,Y,\ldots$ to denote RVs.  The value space of an
RV such as $X$ is denoted by $\Rng(X)$. The cardinality of the
  value space of $X$ is $|\Rng(X)|$.  Values of RVs are denoted by
the corresponding lower-case letters. Thus $x$ is a value of $X$,
often thought of as the particular value realized in an experiment.
In the same spirit, we use $\Omega$ to denote the \emph{universal} RV
defined as the identity function on the set of the underlying
probability space.  Values of $\Omega$ are denoted by $\omega$.  When
using symbols for values of RVs, they are implicitly assumed to be
members of the range of the corresponding RV.  In many cases, the
value space is a set of letters or a set of strings of a given length.
We use juxtaposition to denote concatenation of letters and strings.
For a string $s$, $|s|$ denotes its length. Unless stated otherwise, a
string-valued RV $S$ produces fixed length strings. $|S|$ denotes the
length of the strings, $S_{i}$ is the $i$'th letter of the string, and
$S_{\le i}$ is the length $i$ prefix of the string. By default,
strings are binary, which implies, for example,
$|S|=\log_{2}(|\Rng(S)|)$.  Sequence RVs (stochastic sequences) are
denoted by capital bold-face letters, with the corresponding
lower-case bold-face letters for their values. For example, we write
$\Sfnt{A}=(A_{i})_{i=1}^{N}$ and $\Sfnt{A}_{\le m}=(A_{i})_{i=1}^{m}$.
Our conventions for indices are that we generically use $N$ to denote
a large upper bound on sequence lengths, $n$ to denote the available
length and $i,j,k,l,m$ as running indices.  By convention,
$\Sfnt{A}_{\le 0}$ is the empty sequence of RVs. Its value is
constant, independent of $\Omega$.  When multiple stochastic sequences
are in play, we refer to the collection of $i$'th RVs in the sequences
as the data from the $i$'th \emph{trial}. We typically imagine the trials
as happening in time and being performed by an experimenter.  We refer
to the data from the trials preceding the upcoming one as the
\emph{past}. The past can also include initial conditions and any
additional information that may have been obtained. These are normally
implicit when referring to or conditioning on the past.

Probabilities are denoted by $\Prob(\ldots)$.  If there are multiple
probability distributions involved, we disambiguate with a subscript
such as in $\Prob_{\nu}(\ldots)$ or simply $\nu(\ldots)$, where $\nu$
is a probability distribution. We generally reserve the symbol $\mu$
for the global, implicit probability distribution, and may write
$\mu(\ldots)$ instead of $\Prob(\ldots)$.  Expectations are similarly
denoted by $\Exp(\ldots)$ or $\Exp_{\mu}(\ldots)$.  If $\phi$ is a
logical expression involving RVs, then $\{\phi\}$ denotes the event
where $\phi$ is true for the values realized by the RVs.  For example,
$\{f(X)>0\}$ is the event $\{\omega:f(X(\omega))>0\}$ written in full
set notation.  The brackets $\{\ldots\}$ are omitted for events inside
$\Prob(\ldots)$ or $\Exp(\ldots)$. As is conventional, commas
separating logical expressions are interpreted as conjunction.  When
the capital/lower-case convention can be unambiguously interpreted, we
abbreviate ``$X=x$'' by ``$x$''. For example, with this convention,
$\Prob(x,y)=\Prob(X=x,Y=y)$.  Furthermore, we omit commas in the
abbreviated notation, so $\Prob(xy)=\Prob(x,y)$.  RVs or functions of
RVs appearing outside an event but inside $\Prob(\ldots)$ or after the
conditioner in $\Exp(\ldots|\ldots)$ result in an expression that is
itself an RV. We can define these without complications because of our
assumption that the event space is finite.  Here are two
examples. $\Probv(f(X)|Y)$ is the RV whose value at $\omega$ is
$\Prob(f(X)=f(X(\omega))|Y=Y(\omega))$.  This is a function of the RVs
$X$ and $Y$ and can be described as the RV whose value is
$\Prob(f(X)=f(x)|Y=y)$ whenever the values of $X$ and $Y$ are $x$ and
$y$, respectively. Similarly $\Exp(X|Y)$ is the RV defined as a
function of $Y$, with value $\Exp(X|Y=y)$ whenever $Y$ has value $y$.
Note that $X$ plays a different role before the conditioners in
$\Exp(\ldots)$ than it does in $\Prob(\ldots)$,
as $\Exp(X|Y)$ is not a function of $X$, but only of $Y$.  We comment
that conditional probabilities with conditioners having probability
zero are not well-defined, but in most cases can be defined
arbitrarily. Typically, they occur in a context where they are
multiplied by the probability of the conditioner and thereby
contribute zero regardless. An important context involves
expectations, where we use the convention that when expanding an
expectation over a finite set of values as a sum, zero-probability
values are omitted. We do so without explicitly adding the constraints
to the summation variables.  We generally use conditional
probabilities without explicitly checking for probability-zero
conditioners, but it is necessary to monitor for well-definedness of
the expressions obtained.

To denote general probability distributions, usually on the joint
value spaces of RVs, we use symbols such as $\mu,\nu,\rho$, with
modifiers as necessary. As mentioned, we reserve the unmodified $\mu$
for the distinguished global distribution under consideration, if
there is one. Other symbols typically refer to probability
distributions defined on the joint range of some subset of the
available RVs.  The set of probability distributions on $\Rng(X)$ is
denoted by $\cS_{X}$.  We usually just say ``distribution'' instead of
``probability distribution''.  The terms ``distributions on
$\Rng(X)$'' and ``distributions of $X$'' are synonymous.  The support
of a distribution $\nu$ on $X$ is denoted by
$\Supp(\nu)=\{x|\nu(x)>0\}$.  When multiple RVs are involved we denote
marginal and conditional distributions by expressions such as
$\mu[X|Y=y]$ for the distribution of $X$ on its value space
conditional on $\{Y=y\}$.  The probability of
$x$ for this distribution can
be denoted by $\mu[X|Y=y](x)$, which is well-defined for
$\Prob(Y=y)>0$ and therefore well-defined with probability one
because of our finiteness assumptions.  Note the use of square
brackets to distinguish the distribution specification from the
argument determining the probability at a point.  If $\nu$ is a joint
distribution of RVs, then we extend the conventions for arguments of
$\Prob(\ldots)$ to arguments of $\nu$, as long as all the arguments
are determined by the RVs for which $\nu$ is defined. For example, if
$\nu$ is a joint distribution of $X$, $Y$, and $Z$, then $\nu(x|y)$
has the expected meaning, as does the RV $\nu(X|Y)$ in contexts
requiring no other RVs.  We denote the uniform distribution on
$\Rng(X)$ by $\Unif_{X}$, omitting the subscript if the value space is
clear from context.  If $R$ and $S$ are independent RVs with marginal
distributions $\nu=\mu[R]$ and $\nu'=\mu[S]$ on their ranges, then
their joint distribution is denoted by $\mu[R,S]=\nu\otimes\nu'$.

In our work, probability distributions are constrained by a
statistical \emph{model}, which is defined as a set of distributions and
denoted by letters such as $\cH$ or $\cC$. The models for trials to be
considered here are usually convex and closed.  For a model $\cC$, we
write $\xtrm{\cC}$ for the set of extreme points of $\cC$ and
$\cCvx(\cC)$ for the convex closure of $\cC$ defined as the smallest
closed convex set containing $\cC$.

The total variation (TV) distance between $\nu$ and $\nu'$ is defined
as
\begin{equation}
  \TV(\nu,\nu') = \sum_{x}(\nu(x)-\nu'(x))\knuth{\nu(x)\geq\nu'(x)}
  = \frac{1}{2}\sum_{x}|\nu(x)-\nu'(x)|,
  \label{eq:def_tv}
\end{equation}
where $\knuth{\phi}$ for a logical expression $\phi$ denotes the
$\{0,1\}$-valued function evaluating to $1$ iff $\phi$ is true.
True to its name, the TV distance satisfies the triangle
inequality. Here are three other useful properties.  First, if 
$\nu$ and $\nu'$ are joint distributions of $X$ and $Y$ and the 
marginals satisfy $\nu[Y]=\nu'[Y]$, then the TV distance between 
$\nu$ and $\nu'$ is the average of the TV distances of the 
$Y$-conditional distributions:
\begin{align}
  \TV(\nu,\nu') &= \sum_{y}\sum_{x}(\nu(x,y)-\nu'(x,y))\knuth{\nu(x,y)\ge\nu'(x,y)}\notag\\
  &= \sum_{y}\sum_{x}(\nu(x|y)\nu(y)-\nu'(x|y)\nu'(y))\knuth{\nu(x|y)\nu(y)\ge\nu'(x|y)\nu'(y)}\notag\\
  &= \sum_{y}\sum_{x}(\nu(x|y)\nu(y)-\nu'(x|y)\nu(y))\knuth{\nu(x|y)\nu(y)\ge\nu'(x|y)\nu(y)}\notag\\
  &= \sum_{y}\nu(y)\sum_{x}(\nu(x|y)-\nu'(x|y))\knuth{\nu(x|y)\ge\nu'(x|y)}\notag\\
  &= \sum_{y}\nu(y) \TV(\nu[X|y],\nu'[X|y]).
  \label{eq:tv_samemarg}
\end{align}
Second, if for all $y$, $\nu[X|y]=\nu'[X|y]$, then the TV distance
between $\nu$ and $\nu'$ is given by the TV distance between the marginals on $Y$:
\begin{align}
  \TV(\nu,\nu') &= \sum_{y}\sum_{x}(\nu(x,y)-\nu'(x,y))\knuth{\nu(x,y)\ge\nu'(x,y)}\notag\\
  &= \sum_{y}\sum_{x}(\nu(x|y)\nu(y)-\nu'(x|y)\nu'(y))\knuth{\nu(x|y)\nu(y)\ge\nu'(x|y)\nu'(y)}\notag\\
  &= \sum_{y}\sum_{x}(\nu(x|y)\nu(y)-\nu(x|y)\nu'(y))\knuth{\nu(x|y)\nu(y)\ge\nu(x|y)\nu'(y)}\notag\\
  &= \sum_{y}\sum_{x}\nu(x|y)(\nu(y)-\nu'(y))\knuth{\nu(y)\ge\nu'(y)}\notag\\
  &= \sum_{y}(\nu(y)-\nu'(y))\knuth{\nu(y)\ge\nu'(y)}\notag\\
  &= \TV(\nu[Y],\nu'[Y]).
  \label{eq:tv_samecond}
\end{align}
Third, the TV distance satisfies the data-processing inequality. That
is, for any stochastic process $\cE$ on $\Rng(X)$ and distributions
$\nu$ and $\nu'$ of $X$, $\TV(\cE(\nu),\cE(\nu'))\leq \TV(\nu,\nu')$.
We use this property only for functions $\cE$, but for 
general forms of this result, see Ref.~\cite{pardo:1997}.  

When constructing distributions close to a given one in TV distance,
it is often convenient to work with subprobability distributions.  A
subprobability distribution of $X$ is a sub-normalized non-negative
measure on $\Rng(X)$, which in our case is simply a non-negative
function $\tilde\nu$ on $\Rng(X)$ with \emph{weight}
$w(\tilde\nu)=\sum_{x}\tilde\nu(x)\leq 1$. For expressions not
involving conditionals, we use the same conventions for subprobability
distributions as for probability distributions.  When comparing
subprobability distributions, $\tilde\nu\leq\tilde\nu'$ means that
for all $x$, $\tilde\nu(x)\leq\tilde\nu'(x)$, and we say that
$\tilde\nu'$ \emph{dominates} $\tilde\nu$.

\begin{lemma}\label{lm:tvdist_weight}
  Let $\tilde\nu$ be a subprobability distribution of $X$ of weight
  $w=1-\epsilon$.  Let $\nu$ and $\nu'$ be distributions of $X$
  satisfying $\tilde\nu\leq\nu$ and $\tilde\nu\leq \nu'$. Then
  $\TV(\nu,\nu')\leq \epsilon$.
\end{lemma}

\begin{proof}
  Calculate 
  \begin{align}
    \TV(\nu,\nu') &=
    \sum_{x}(\nu(x)-\nu'(x))\knuth{\nu(x)\geq \nu'(x)}\notag\\
    &\leq \sum_{x}(\nu(x)-\tilde\nu(x))\knuth{\nu(x)\geq \tilde\nu(x)}\notag\\
    &=\sum_{x}(\nu(x)-\tilde\nu(x))\notag\\
    &=1-w=\epsilon.
  \end{align}
\end{proof}

\begin{lemma}\label{lm:tvsub}
  Assume that $p\geq1/|\Rng(X)|$.
  Let $\nu$ be a distribution of $X$ and $\tilde\nu\leq\nu$ a subprobability
  distribution of $X$ with weight $w=1-\epsilon$
  and $\tilde\nu\leq p$. Then there exists a distribution $\nu'$ of $X$
  with $\nu'\geq \tilde\nu$, $\nu'\leq p$, and $\TV(\nu,\nu')\leq \epsilon$.
\end{lemma}

\begin{proof}
  Because $p\geq 1/|\Rng(X)|$ and $\sum_{x}(p-\tilde\nu(x))=
  |\Rng(X)|p-w\geq 1-w$, there exists a distribution
  $\nu'\geq\tilde\nu$ with $\nu'\leq p$. Since $\nu'$ and $\nu$ are
  distributions dominating $\tilde\nu$ and by
  Lem.~\ref{lm:tvdist_weight}, $\TV(\nu,\nu')\leq \epsilon$.
\end{proof}

\subsection{Bell-Test Configurations}
\label{subsec:setup_belltests}

The standard example of a Bell-test configuration involves a pair of
devices located at two stations, $\Pfnt{A}$ and $\Pfnt{B}$.  The
stations engage in a sequence of trials.  In a trial, each station
chooses one of two settings and obtains a measurement outcome, either
$0$ or $1$, from their device. A Bell test with this configuration is
called a $(2,2,2)$ Bell test, where the numbers indicate the number of
stations, the number of settings available to each station and the
number of outcomes at each setting. The test produces two stochastic
sequences from each station. The sequence of settings choices is
denoted by $\Sfnt{X}$ for $\Pfnt{A}$ and $\Sfnt{Y}$ for $\Pfnt{B}$,
and the sequence of measurement outcomes is denoted by $\Sfnt{A}$ and
$\Sfnt{B}$, respectively.  In using this notation, we allow for
arbitrary numbers of settings and measurement outcomes. For a
$(2,k,l)$ Bell-test configuration, $|\Rng(X_{i})|=|\Rng(Y_{i})|=k$
and $|\Rng(A_{i})|=|\Rng(B_{i})| = l$. The main role of separating
the configuration by station is to justify the assumptions, which can
be generalized to more stations if desired. The assumptions constrain
the models consistent with the configuration. In this work, once the
constraints are determined, the separation into stations plays little
role, but we continue to identify settings and outcomes RVs.  We use
$\Sfnt{Z}$ and $\Sfnt{C}$ to denote the corresponding stochastic
sequences. For the case of two stations, $\Sfnt{Z} =
(X_{i}Y_{i})_{i=1}^{N}$ and $\Sfnt{C}=(A_{i}B_{i})_{i=1}^{N}$.

We also refer to RVs $\Sfnt{D}$ and $\Sfnt{R}$. For $\Sfnt{D}$, we
assume that $D_{i}$ is determined by (that is a function of)
$C_{i}$. We use $\Sfnt{D}$ to study protocols where the randomness is
extracted from $\Sfnt{D}$ instead of $\Sfnt{C}$.  For example, this
includes protocols where only $\Pfnt{A}$'s measurement outcomes are
used, in which case we set $\Sfnt{D}=\Sfnt{A}$.  We assume that
$C_{i}$ is determined by $R_{i}$, and we use $\Sfnt{R}$ to contain
additional information accumulated during the experiment that can be
used to adapt the protocol but is kept private.  The main purpose of
the RV $\Sfnt{Z}$ is to contain data that may become public. This
applies to the settings if the source of the random settings choices
is public or otherwise not sufficiently trusted.  In situations where
there is only one trial under consideration or we do not need the
structure of the stochastic sequences, we use the non-boldface
variants of the RVs, that is $D$, $C$, $R$, and $Z$.

\subsection{Assumptions for Bell Tests}
\label{subsec:ass_belltests}

We are interested in limiting the information held by an external
entity.  This work is concerned with external entities holding
classical side information, so the external entity $\Pfnt{E}$'s state
is characterized by an RV $E$. Physical entities and systems are
denoted by capital sans-serif letters, a convention we already used to
refer to stations. All our results are proven for finite event spaces.
While this is reasonable with respect to the RVs representing
information explicitly used by the protocols, it is an unverifiable
restriction on $E$.  For countable $E$ one can use the observation
that conditioning on large-probability finite subsets of $E$ gives
distributions that are arbitrarily close in $\TV$ distance to the
unconditioned distribution.

Our theory is formulated in the untrusted device
model~\cite{colbeck:2007} where $\Pfnt{E}$ may have had the devices in
possession before the trials of the protocol of interest start. Once
the protocol starts, $\Pfnt{E}$ cannot receive information from the
devices.  Note that this does not preclude that $\Pfnt{E}$ sends
information to the devices via one-way channels, see below.  This
simplification is possible for local randomness-generation protocols
where the participating parties need no remote classical
communication. In many other applications such as quantum key
distribution, the protocols involve both quantum and classical
communication, in which case $\Pfnt{E}$ can gain additional
information at any time.

To ensure that $\Pfnt{E}$ cannot receive information from the devices
after the start of the protocol requires physical security to isolate
the joint location of the devices and the stations.  Formalizing this
context requires a physical framework for which subsystems and
interactions can be defined, with classical subsystems playing a
special role. In the case of applications to Bell-test configurations,
we motivate the constraints on the models with physical arguments, but
we do not prove the constraints with respect to a formal physical
framework.  Let $\Pfnt{D}$ be the pair of devices with which
$\Pfnt{A}$ and $\Pfnt{B}$ implement the protocol. The absence of
information passing from $\Pfnt{ABD}$ to $\Pfnt{E}$ is ensured if
there is no physical interactions between the two sets of systems
after the protocol starts. While $\Sfnt{Z}$ is private, we can time
shift any one-way communication from $\Pfnt{E}$ to $\Pfnt{ABD}$, or
non-interacting dynamics of $\Pfnt{E}$ to before the beginning of the
protocol.  A communication from $\Pfnt{E}$ to $\Pfnt{ABD}$ can be
described by adding the communication as a system $\Pfnt{C}$ to
$\Pfnt{E}$, which interacts with $\Pfnt{E}$ then becomes isolated from
$\Pfnt{E}$ and later becomes part of the devices. This makes time
shifting possible, after which we can assume that any physical
processes relevant to the protocol act only on $\Pfnt{ABD}$. Of
course, we insist that after time-shifting, the state of $\Pfnt{E}$ is
classical.  This justifies the use of a single RV $E$ to determine the
state of $\Pfnt{E}$. In this work, we make few assumptions on the
physics of $\Pfnt{D}$ and formulate constraints on distributions
purely in terms of the RVs produced by the protocol, conditional on
$E$.

The formal restriction that $\Pfnt{E}$ holds only classical side
information allows for quantum $\Pfnt{E}$ provided that there is an
additional quantum system $\Pfnt{H}$ independent of $\Pfnt{E}$ such that the
joint state at the start of the protocol of the systems $\Pfnt{ABD}$,
$\Pfnt{H}$ and $\Pfnt{E}$ forms a quantum Markov
chain $\Pfnt{ABD}\leftrightarrow
\Pfnt{H}\leftrightarrow\Pfnt{E}$~\cite{hayden:qc2004b}. This is a
quantum generalization of stating that $\Pfnt{E}$ and $\Pfnt{ABD}$ are
independent conditionally on $\Pfnt{H}$. The system $\Pfnt{H}$ needs
to have no interaction with \Pfnt{E} and \Pfnt{ABD}
after the start of the
protocol. For example, $\Pfnt{H}$ can be part of a generic decohering
environment.  The Markov chain property implies that after an
extension of $\Pfnt{E}$, without loss of generality, the information
held by $\Pfnt{E}$ can be assumed to be classical. Operationally, this
situation can be enforced if we trust or ensure that the devices have
no long-term quantum memory.

Two assumptions constraining the possible distributions are required
for randomness generation with Bell tests.  The settings constraints
restrict the settings distributions, and the non-signaling constraints
enforce the absence of communication between the stations during a
trial. There are different ways in which the settings distributions
can be constrained.  The strongest constraint we use is that $Z_{i+1}$
is uniform, independent of $\Sfnt{Z}_{\le i}$, $\Sfnt{C}_{\le i}$ and
$E$. In general, we can consider weaker constraints on $Z_{i+1}$. For
example, $Z_{i+1}$ may have any distribution that is independent of
$E$ but determined in a known way by the past.  Most generally,
$Z_{i+1}$ has a distribution belonging to a specified set conditional
on $E$ and the past.  If we later allow $\Sfnt{Z}$ to become public, or
if $\Sfnt{Z}$ is from a public source of randomness, then $Z_{i+1}$
must be conditionally independent of $\Sfnt{R}_{\le i}$ given
$\Sfnt{Z}_{\le i}$ and $E$.  This constraint can be avoided if we
trust that the data that normally contributes to $\Sfnt{Z}$ remains
private, in which case we can exploit that our framework is flexible
in how the experimental data relates to the random variables: If
normally $\Sfnt{Z}$ contains settings choices, we can instead
let $\Sfnt{Z}$ be empty, add the
settings choices into the RV $\Sfnt{C}$ and extract randomness from
all or some of $\Sfnt{C}$.  With this assignment of RVs, conditional
independence is not required.  However, for non-trivial randomness
generation, it is then desirable to extract strictly more random bits
than were required for the settings choices. Otherwise the randomness
obtained may just be due to the input randomness with no other
contribution.

To state the non-signaling constraints requires
explicitly referencing
each station's RVs.  The non-signaling constraints are assumed to
be satisfied for trial $i$ conditional on the past and $E$. With
conditioning on these RVs implicit and for two stations, the
non-signaling constraints consist of the identities
\begin{equation}
  \Probv(A_{i}|X_{i}Y_{i}) = \Probv(A_{i}|X_{i}),\;\;
  \Probv(B_{i}|X_{i}Y_{i}) = \Probv(B_{i}|Y_{i}),
  \label{eq:non-signaling_constraints}
\end{equation}
which assert remote settings independence of measurement outcomes.
The non-signaling constraints can be strengthened using the
semidefinite-programming hierarchy for quantum
distributions~\cite{navascues:2007} when the devices are assumed to be
quantum.

In any given implementation of randomness generation, the settings and
non-signaling constraints must be justified by physical arguments.
The settings constraints require a source of random bits that are
sufficiently independent of external entities at the last point in
time where they had interactions with the devices and the stations.
For protocols such as those discussed here, where the output's
randomness is assured conditionally on the settings, the settings
randomness can become known after the last interaction, so it is
possible to use public sources of randomness. However, if the same
public string is used by many randomness generators, a scattershot
attack may be possible.  To justify the non-signaling constraints,
this randomness must not leak to the stations' devices before it is
used to apply the device settings.  Thus, we need to trust that it is
possible to generate bits that are, before the time of use,
unpredictable by the stations' devices.  For the non-signaling
constraints, we also rely on physical isolation of the stations during
a trial, preferably enforced by relativistic causality.  It usually
goes without saying that we trust the timing, recording and computing
devices that monitor and analyze the data according to the
protocol. However, given the complexity of modern electronics and the
prevalence of reliability and security issues, caution is advised.

\subsection{Min-Entropy and Randomness Extraction}
\label{sec:minentropy_extraction}

To be able to extract some amount of randomness from a bit-string RV
$D$, it suffices to establish a min-entropy bound for the distribution
of $D$.  The min-entropy of $D$ is $H_{\min}(D)=-\log(p_{\max}(D))$,
where $p_{\max}(D) = \max_{d}\Prob(d)$. By default, logarithms are
base $e$, so entropy is expressed in \emph{nits} (or $e$-bits), which
simplifies calculus with entropic quantities. We convert to bits to
  determine string lengths as needed. We switch to
logarithms base $2$ in Sect.~\ref{sec:apps} for quantitative
comparisons.  Since we usually work with bounds on min-entropies
rather than exact values, we say that $D$ has min-entropy $-\log(p)$
or $D$ has max-prob $p$ if $p_{\max}\leq p$ and add the adjective
``exact'' to refer to $-\log(p_{\max})$ or $p_{\max}$. If $D$ has
min-entropy $\sigma\log(2)$, then it is possible to extract close to
$\sigma$ near-uniform bits from $D$ given some uniform seed bits.  The
actual number of near-uniform bits that can be safely extracted
depends on $|D|$, how many uniform seed bits are used, the maximum
acceptable TV distance from uniform of the near-uniform output bits,
and the \emph{extractor} used.  An extractor $\cE$ has as input a bit
string $D$, a string of uniform seed bits $S$ independent of all other
relevant RVs, and the following parameters: A lower bound
$\sige\log(2)$ on the min-entropy of $D$, the desired number $\sigma$
of uniform bits, and the maximum acceptable TV distance from uniform
$\epsx $ of the output.  We write $\cE(D,S;\sige,\sigma,\epsx)$ for
the length $\sigma$ bit string produced by the extractor.  This bit
string is close to uniform provided the input satisfies
\emph{extractor constraints} that depend on the specific extractor
used.  Formally, a \emph{strong extractor} has the property that if
the parameters $n=|D|,l=|S|,\sige ,\sigma,\epsx $ satisfy its
extractor constraints, $S$ is an independent and uniform bit string,
and $D$ has min-entropy at least $\sige\log(2)$, then
\begin{equation}
  \TV(\mu[\cE(D,S;\sige,\sigma,\epsx)S],\Unif)\leq\epsx.
  \label{eq:strongext}
\end{equation}
Extractors used in this work are strong by default.  The
statement that $(n,\sige ,\sigma,\epsx )$ satisfies the extractor
constraints means that there exists $l$ so that the full list of
parameters $(n,l,\sige ,\sigma,\epsx )$ satisfies the extractor
constraints.  We always include $\sige\leq n$, $1\leq \sigma\leq
\sige$ and $0<\epsilon_{x}\leq 1$ among the extractor
constraints. Extractor constraints are partially determined by
information theoretic bounds, but constraints for practical extractors
are typically stronger. See Ref.~\cite{vadhan:2012}.

In previous work~\cite{bierhorst:qc2017a}, we have used an
implementation of Trevisan's strong extractor based on the framework
of Mauerer, Portmann and Scholz~\cite{mauerer:2012} that we called the
TMPS extractor $\cE_{\TMPS}$. The extractor constraints for
$\cE_{\TMPS}$ are not optimized. Simplified constraints are
$0<\epsx \leq 1$, $2\leq\sigma\leq\sige \leq n$,
$\sigma\in\nats$, and
\begin{align}
  \sigma+4\log_{2}(\sigma) &\le \sige -  4\log_2(1/\epsx )-6, \notag\\
  l &\ge 36 \log_{2}(\sigma) (\log_2(4n\sigma^2/\epsx ^2))^2.
  \label{eq:TMPS_bounds}
\end{align}
See Ref.~\cite{bierhorst:qc2017a} for the smaller expression for $l$
used by the implementation.  The TMPS extractor is designed to be
secure against quantum side information.  In general, strong
extractors can be made classical-proof, that is, secure against
classical side information, with a modification of the extractor
constraints. Some strong extractors such as Trevisan's can also be
made quantum-proof, that is, secure against quantum side information,
with a further modification of the constraints.  See
Ref.~\cite{mauerer:2012} for explanations and references to the
relevant results, specifically Lem.~A.3, Thm.~B.4, Lem.~B.5 and
Lem.~B.8.  All our protocols assume strong extractors and are
secure against classical side information. But they do
not require that the strong extractor be classical- or quantum-proof:
In one case, the relevant constraint modification is implicit in our
proof, and in the other cases, we take advantage of special properties
of our techniques.  The TMPS extractor constraints given above include
the modifications required for being quantum-proof; we did not roll
back these modifications.

The TMPS extractor also satisfies that conditional on the
seed bits, the output is a linear hash of $D$. For bits, a linear
hash of $D$ is a parity of $D$, computed as $\sum_{i}h_{i}D_{i}
\mod(2)$, where the $h_{i}\in\{0,1\}$ form the parity vector.  We
call extractors with this property \emph{linear extractors}.

\Pc{Here is the exact expression for and derivation of the bound on
  $l$ in Eq.~\ref{eq:TMPS_bounds}.

  For the TMPS extractor, the number of seed bits $l$ is constrained
  according to 
  \begin{equation*}
    l \ge w^2\max \left(2, 1+
      \left\lceil[\log_2(\sigma-e)-\log_2(w-e)]/[\log_2e-\log_2(e-1)]\right\rceil\right) 
  \end{equation*}
  where $w$ is the smallest prime number satisfying $w\geq
  2\lceil\log_2(4n\sigma^2/\epsx^2)\rceil$.  First drop the
  negative contribution $-e$ inside the first logarithm and the summand
  $-\log_{2}(w-e)$ from the expression and apply
  $\log_{2}(e)-\log_{2}(e-1)\approx 0.66173>1/2$ to show that $l \ge
  w^2\max(2,1+ \lceil 2\log_2(\sigma)\rceil)$
  implies the exact constraint. 
  Next, we increase the right-hand side according to
  \begin{equation*}
    \max(2,1+\lceil 2\log_{2}(\sigma)\rceil)
    \leq \max(2,2+2\log_{2}(\sigma))
    =2+2\log_{2}(\sigma)\leq 4\log_{2}(\sigma),
  \end{equation*}
  taking advantage of the constraint $\sigma\geq 2$.  Now we have
  $l\ge 4 w^{2}\log_{2}(\sigma)$.  The lower bound on $w$ is at least
  $10$.  According to the improvement on Bertrand's postulate
  in Ref.~\cite{nagura:qc1952a}, for any integer $x\ge 9$, there exists a 
  prime $p$ with $x<p<4x/3$. With $x=2\lceil\log_{2}(4n\sigma^{2}/\epsx^{2})\rceil-1\geq 9$, we have $w<4x/3$. Since
  $2\lceil y\rceil\leq 2(y+1)$, the upper bound on $w$ is
  \begin{equation*}
    w\le \lfloor (4/3)(\lceil 2(\log_{2}(4n\sigma^{2}/\epsx^{2}\rceil-1)\rfloor
    = \lfloor (8/3)\log_{2}(4n\sigma^{2}/\epsx^{2})+4/3\rfloor
    \le 3\log_{2}(4n\sigma^{2}/\epsx^{2}),
  \end{equation*}
  since $(1/3)\log_{2}(4n\sigma^{2})\ge 4/3$.
  The simplified constraint in
  Eq.~\ref{eq:TMPS_bounds} therefore suffices.  }

Given that the extracted bits are allowed a non-zero distance from
uniform, it is not necessary to satisfy a strict min-entropy bound on
the extractor input $D$. It suffices for the distribution of $D$ to
have some acceptably small TV distance $\epse $ from a distribution
with min-entropy $\sige\log(2) $.  We say that $D$ has $\epse $-smooth
min-entropy $\sige\log(2) $ if its distribution is within TV distance
$\epse$ of one with min-entropy $\sige\log(2) $.  For the situation
considered so far in this section, the error bound (or
\emph{smoothness}) $\epse $ and the extractor error $\epsx $ can be
added when applying the extractor to $D$. In this work, we generally
work directly with the maximum probability.  We say that $D$ has
$\epse $-smooth max-prob $p$ if it has $\epse$-smooth min-entropy
$-\log(p)$.

\Pc{The following lemma is not used in
  the paper.

  \begin{lemma*}
    Let $p\geq 1/|\Rng(D)|$.  Suppose that there is an event $\{\phi\}$
    such that $\Prob(\phi)\geq 1-\epse $ and $\Probv(D,\phi) \leq
    p$ as a function on event space. Then $D$ has
    $\epse $-smooth max-prob $p$. Conversely, if $D$ has
    $\epse $-smooth max-prob $p$, then there is an event $\{\phi\}$,
    possibly in an extended event space, such that $\Prob(\phi)\geq
    1-\epse $ and $\Probv(D, \phi) \leq p$.
  \end{lemma*}

  \begin{proof} Define the subprobability distribution $\tilde\nu$ of
    $D$ by $\tilde\nu(d)=\Prob(d,\phi)$. Then $\tilde\nu$ has weight
    \begin{equation}
      w=\sum_{d}\tilde\nu(d)=\sum_{d}\Prob(d,\phi)=\Prob(\phi)\geq 1-\epse 
      \notag
    \end{equation}
    and is dominated by $\mu[D]$.  We have $\tilde\nu\leq p$, so by
    Lem.~\ref{lm:tvsub}, there is a distribution $\nu'\geq\tilde\nu$
    with $\nu'\leq p$.  By Lem.~\ref{lm:tvdist_weight},
    $\TV(\nu',\mu[D])\leq 1-w=\epse$.  It follows that $\mu[D]$ has $\epse
    $-smooth max-prob $p$.

    For the converse, let $\nu$ be a max-prob $p$ distribution on
    $\Rng(D)$ within $\TV$ distance $\epse $ of $\mu[D]$. Define
    $\tilde\nu$ by $\tilde\nu(d)=\min(\nu(d),\mu(d))$. Then $\tilde\nu$
    has weight $w\geq 1-\epse $.  Extend the event space so that we have
    an RV $F$ uniformly distributed on $[0,1]$ and independent of all
    original events.  We can now define $\{\phi\}= \{F\leq
    \tilde\nu(D)/\mu(D)\}$ where we let $\tilde\nu(d)/\mu(d)=0$ if
    $\mu(d)=0$.  By independence of $F$, $\Prob(d,\phi) =
    \Prob(d)\Prob(\phi|c)=\Prob(d)\tilde\nu(d)/\mu(d) =
    \tilde\nu(d) \leq p$ (because by definition $\Prob(d)=\mu(d)$), and
    $\Prob(\phi)=\sum_{d}\Prob(d,\phi) = \sum_{d}\tilde\nu(d)=w\geq
    1-\epse $.  While $F$ as introduced is not finite-valued, only its 
    cumulative distribution at values achieved by $\tilde\nu(D)/\mu(D)$ 
    matter, so we can replace it with a finite-valued RV to preserve 
    our finiteness assumptions, if necessary.
  \end{proof}
}

We generally want to generate bits that are near-uniform conditional
on $E$ and often other variables such as $Z$. For our analyses, $E$ is
not particularly an issue because our results hold uniformly for all
values of $E$, that is, conditionally on $\{E=e\}$ for each
$e$. However this is not the case for $Z$.

\begin{definition}\label{def:smoothminentropy}
  The distribution $\mu$ of $DZE$ has \emph{$\epsilon$-smooth
    $ZE$-conditional max-prob $p$} if the following two conditions
  hold: 1) For each $ze$ there exists a subprobability distribution
  $\tilde\mu_{ze}$ of $D$ such that $\tilde\mu_{ze}\leq\mu[D|ze]$ and
  $\tilde\mu_{ze}\leq p$; 2) the total weight of these subprobability
  distributions satisfies $\sum_{ze}w(\tilde\mu_{ze})\mu(ze)\geq
  1-\epsilon$.  The minimum $p$ for which $\mu$ has $\epsilon$-smooth
  $ZE$-conditional max-prob $p$ is denoted by $P^{{\rm
      u},\epsilon}_{\max,\mu}(D|ZE)$. (The superscript ${\rm u}$
  alludes to the uniformity of the bound with respect to $ZE$.)

  The distribution $\mu$ of $DZE$ has \emph{$\epsilon$-smooth average
    $ZE$-conditional max-prob $p$} if there exists a
  distribution $\nu$ of $DZE$ with $\TV(\nu,\mu)\leq\epsilon$ and
  $\sum_{ze}\max_{d}(\nu(d|ze))\nu(ze) \leq p$. The minimum $p$ for
  which $\mu$ has $\epsilon$-smooth average $ZE$-conditional
  max-prob $p$ is denoted by $P^{\epsilon}_{\max,\mu}(D|ZE)$.  The quantity
  $H^{\epsilon}_{\min,\mu}(D|ZE)=-\log(P^{\epsilon}_{\max,\mu}(D|ZE))$ is the
  \emph{$\epsilon$-smooth $ZE$-conditional min-entropy}.
\end{definition}

We refer to the smoothness parameters as \emph{error bounds}.  Observe
that the definitions are monotonic in the error bound.  For example,
if $P^{\epsilon}_{\max,\mu}\leq p$ and $\epsilon'\geq \epsilon$, then
$P^{\epsilon'}_{\max,\mu}\leq p$.  The quantity
$\sum_{ze}\max_{d}(\nu(d|ze))\nu(ze)$ in the definition of
$P^{\epsilon}_{\max,\mu}$ can be recognized as the (average) maximum
guessing probability of $D$ given $Z$ and $E$ (with respect to $\nu$),
whose negative logarithm is the guessing entropy defined, for
example, in Ref.~\cite{konig:2008}.  The relationships established
below reprise results from the references. We use them to prove
soundness of the first two protocols for composing probability
estimation with randomness extractors
(Thms.~\ref{thm:prot_chainlemmas} and~\ref{thm:estext_nofail}) but
bypass them for soundness of the third (Thm.~\ref{thm:estext_fail}).

We focus on probabilities rather than entropies because in this
work, we achieve good performance for finite data by direct
estimates of probabilities of actual events, not entropies of
distributions. While entropy estimates may be considered the
ultimate goal for existing applications, they are not fundamental in
our approach. The focus on probabilities helps us to avoid
introducing logarithms unnecessarily.

A summary of the relationships between conditional min-entropies and
randomness extraction is given in Ref.~\cite{koenig:qc2009a} for the
quantum case and can be specialized to the classical case.  When so
specialized, the definition of smooth conditional min-entropy in, for
example, Ref.~\cite{koenig:qc2009a} differs from the one above in that
Ref.~\cite{koenig:qc2009a} uses one of the fidelity-related
distances. One such distance reduces to the Hellinger distance $h$ for
probability distributions for which $h^{2}\leq \TV \leq \sqrt{2}h$.

\Pc{This is to fill in the
  details on the connection between our definitions of smooth max-prob
  and the quantum conditional min-entropies.

  It takes a bit of work to relate the typical definition of quantum
  conditional min-entropy to the classical definition given here. For
  establishing the relationship, we set the smoothness parameter to
  $0$.  For a quantum state $\rho$ and a positive semidefinite operator $\sigma$ on
  systems $\Pfnt{AB}$, define $d_{\max}(\rho|\sigma)=\min\{p: \rho\leq
  p\sigma\}$.  The conditional min-entropy of $\rho$ given $\Pfnt{B}$
  is the negative log of the minimum value of
  $d_{\max}(\rho|I\otimes\sigma)$ over states $\sigma$ of
  $\Pfnt{B}$. Let $\tilde d_{\max}(\rho|\Pfnt{B})$ be this minimum
  value.  We show that for classical-classical states $\rho$ on
  $\Pfnt{AB}$, we have $\tilde
  d_{\max}(\rho|\Pfnt{B})=P^{0}_{\max,\mu}(A|B)$, where $\mu$ is the
  probability distribution defined by the diagonal of $\rho$. By
  definition, $\rho$ is classical-classical if it is diagonal in the
  logical (or classical) basis of $\Pfnt{AB}$, so that
  $\rho=\sum_{ab}\rho_{ab}\hat{a}\otimes\hat{b}$ with $\hat u =
  \ketbra{u}{u}$. We then define $\mu$ by $\mu(ab)=\rho_{ab}$.
  
  We first show that $\tilde d_{\max}(\rho|\Pfnt{B})\geq
  P^{0}_{\max,\mu}(A|B)$.  For all $p$ and $\sigma$ such that
  $\rho\leq pI\otimes\sigma$, we have $\rho\leq
  pI\otimes\tilde\sigma$, where $\tilde\sigma=\sum_{j}\nu(b)\hat b$ is
  $\sigma$ decohered in the classical basis. That is,
  $\nu(b)=\sigma_{bb}$ with index notation for $\sigma$.  The
  inequality is equivalent to the statement that for all $b$,
  $\max_{a}(\mu(ab))\leq p\nu(b)$. This implies
  $\sum_{b}\max_{a}(\mu(a|b))\mu(b)=
  \sum_{b}\max_{a}(\mu(a|b)\mu(b))= \sum_{b}\max_{a}\mu(ab)\leq
  \sum_{b}p\nu(b)=p$. That is, $P^{0}_{\max,\mu}(A|B)\leq p$.  The
  definition of $\tilde d_{\max}(\rho|\Pfnt{B})$ and arbitrariness of
  $p$ and $\sigma$ subject to $\rho\leq pI\otimes\sigma$ implies the
  desired inequality.

  To show that $\tilde d_{\max}(\rho|\Pfnt{B})\leq
  P^{0}_{\max,\mu}(A|B)$, suppose that
  $\sum_{b}\max_{a}(\mu(a|b))\mu(b)\leq p$, that is,
  $P^{0}_{\max,\mu}(A|B)\leq p$. Define
  $\tilde\nu(b)=\max_{a}(\mu(a|b))\mu(b)/p$.  Then $\tilde\nu(b)$ is a
  subprobability distribution and $\max_{a}(\mu(ab))/p =
  \tilde\nu(b)$. Let $\nu$ be a distribution dominating $\tilde\nu$.
  Then the density matrix $\sigma=\sum_{b}\nu(b)\hat b$ contributes to
  the definition of $\tilde d_{\max}(\rho|\Pfnt{B})$, which
  establishes the desired inequality.
}

Uniform or average bounds on $Z$-conditional max-probs with respect to
$E=e$ can be lifted to the $ZE$-conditional max-probs, as formalized
by the next lemma.

\begin{lemma}\label{lm:pmax_noe}
  If for all $e$, $P^{{\rm u},\epsilon}_{\max,\mu[DZ|e]}(D|Z)\leq p$, then
  $P^{{\rm u},\epsilon}_{\max,\mu[DZE]}(D|ZE)\leq p$. Suppose that for
  all $e$, $P^{\epsilon_{e}}_{\max,\mu[DZ|e]}(D|Z)\leq p_{e}$, and
  let $\bar\epsilon=\sum_{e}\epsilon_{e}\mu(e)$ and $\bar
  p=\sum_{e}p_{e}\mu(e)$. Then
  $P^{\bar\epsilon}_{\max,\mu[DZE]}(D|ZE)\leq \bar p$.
\end{lemma}

\begin{proof}
  For the first claim, for each $e$, let $\tilde\mu_{ze}$ be
  subprobability distributions witnessing that $P^{{\rm
      u},\epsilon}_{\max,\mu[DZ|e]}(D|Z)\leq p$ according to the
  definition.  Then $\tilde\mu_{ze}\leq \mu[D|ze]$,
  $\tilde\mu_{ze}\leq p$ and $\sum_{z}w(\tilde\mu_{ze})\mu(z|e)\geq
  1-\epsilon$, from which we get
  \begin{equation}
    \sum_{ze}w(\tilde\mu_{ze})\mu(ze) 
    =\sum_{e}\mu(e)\sum_{z}w(\tilde\mu_{ze})\mu(z|e)
    \geq \sum_{e}\mu(e)(1-\epsilon)=(1-\epsilon).
  \end{equation}

  For the second claim, for each $e$, let $\nu_{e}$ witness
  $P^{\epsilon_{e}}_{\max,\mu[DZ|e]}(D|Z)\leq p_{e}$.
  Then $\TV(\nu_{e},\mu[DZ|e])\leq\epsilon_{e}$ and 
  $\sum_{z}\max_{d}(\nu_{e}(d|z))\nu_{e}(z) \leq p_{e}$.
  Define $\nu$ by $\nu(dze)=\nu_{e}(dz)\mu(e)$.
  Then $\nu[E]=\mu[E]$, so we can apply Eq.~\ref{eq:tv_samemarg} for
  \begin{equation}
    \TV(\nu,\mu)=\sum_{e}\TV(\nu_{e},\mu[DZ|e])\mu(e)
    \leq \sum_{e}\epsilon_{e}\mu(e)=\bar\epsilon.
  \end{equation}
  Furthermore,
  \begin{align}
    \sum_{ze}\max_{d}(\nu(d|ze))\nu(ze)
    &= \sum_{e}\mu(e)\sum_{z}\max_{d}(\nu_{e}(d|z))\nu(z|e)\notag\\
    &\leq \sum_{e}\mu(e) p_{e} = \bar p,
  \end{align}
  as required for the conclusion.
\end{proof}

The main relationships between conditional and average conditional
max-probs are determined by the next lemma.

\begin{lemma}\label{lm:pmax_avpmax}
  Let $p\geq 1/|\Rng(D)|$, $\epsilon>0$ and $\mu$ a distribution of
  $DZE$. If $P^{{\rm u},\epsilon}_{\max,\mu}(D|ZE)\leq p$, then
  $P^{\epsilon}_{\max,\mu}(D|ZE)\leq p$. If
  $P^{\epsilon}_{\max,\mu}(D|ZE)\leq p\delta$, then there exists a
  distribution $\nu'$ of $DZE$ with $\TV(\nu',\mu)\leq\epsilon+\delta$ such
  that $P^{{\rm u},0}_{\max,\nu'}(D|ZE)\leq p$.  If
  $P^{\epsilon}_{\max,\mu}(D|ZE)\leq p\delta$, then
  $P^{{\rm u},2\epsilon+\delta}_{\max,\mu}(D|ZE)\leq p$.
\end{lemma}

\begin{proof}
  Suppose that $P^{{\rm u},\epsilon}_{\max,\mu}(D|ZE)\leq p$. Let
  $\tilde\mu_{ze}$ be subprobability distributions witnessing this
  inequality as required by the definition.  Since $p\geq 1/|\Rng(D)|$
  and by Lem.~\ref{lm:tvsub}, there exist distributions $\nu_{ze}$
  with $\nu_{ze}\geq\tilde\mu_{ze}$, $\nu_{ze}\leq p$ and
  $\TV(\nu_{ze},\mu[D|ze])\leq 1-w(\tilde\mu_{ze})$.
  Define $\nu$ by $\nu(dze)=\nu_{ze}(d)\mu(ze)$. Then $\nu[ZE]=\mu[ZE]$, and
  according to Eq.~\ref{eq:tv_samemarg}
  \begin{align}
    \TV(\nu,\mu)&=\sum_{ze}\TV(\nu_{ze},\mu[D|ze])\mu(ze)\notag\\
    &\leq \sum_{ze}(1-w(\tilde\mu_{ze}))\mu(ze) \le 1-(1-\epsilon)=\epsilon.
   \label{eq:tv_lemma6}
  \end{align}
  Since $\sum_{ze}\max_{d}(\nu_{ze}(d))\nu(ze)\leq \sum_{ze}p\mu(ze)=p$,
  the first part of the lemma follows.

  For the second part of the lemma, suppose that
  $P^{\epsilon}_{\max,\mu}(D|ZE)\leq p\delta$ and let $\nu$ be a
  distribution witnessing this inequality. By definition,
  $\TV(\nu, \mu)\leq \epsilon$ and $\sum_{ze}\max_{d}(\nu(d|ze))\nu(ze)\leq p\delta$.  Let
  $m(ze)=\max_{d}(\nu(d|ze))$.  By Markov's inequality,
  $\nu(m(ZE)>p)\leq\delta$. Let $\nu'(dze)=\nu(dze)$ if $m(ze)\leq p$
  and $\nu'(dze)=\nu(ze)/|\Rng(D)|$ otherwise. Then $\TV(\nu',\nu)\leq
  \delta$ and for all $ze$, $\max_{d}(\nu'(d|ze))\leq p$.  Note that
  $\nu'[ZE]=\nu[ZE]$.  By the triangle inequality, $\TV(\nu',\mu)\leq
  \epsilon+\delta$.  By letting $\tilde\nu'_{ze}=\nu'[D|ze]$ in the
  definition for $ZE$-conditional max-prob, we see that $P^{{\rm
      u},0}_{\max,\nu'}(D|ZE)\leq p$. 

  For the last part of the lemma, define
  $\nu''(dze)=\nu'(d|ze)\mu(ze)$.  Then
  \begin{equation}
    \TV(\nu',\nu'')=\TV(\nu'[ZE],\nu''[ZE])=\TV(\nu[ZE],\mu[ZE])\leq
    \epsilon,
  \end{equation}
  where we used the identities $\nu'[D|ze]=\nu''[D|ze]$ and
  Eq.~\ref{eq:tv_samecond}, and applied the data-processing
  inequality. 
  Let
  $\tilde\mu_{ze}(d)=\min(\nu''(d|ze),\mu(d|ze))=
  \min(\nu'(d|ze),\mu(d|ze))$. Then the $\tilde\mu_{ze}$ are
  subprobability distributions, $\tilde\mu_{ze}\leq p$,
  $\tilde\mu_{ze}\leq \mu[D|ze]$ and
  \begin{align}
    \sum_{ze}w(\tilde\mu_{ze})\mu(ze) &=
    \sum_{zed}\min(\nu''(d|ze),\mu(d|ze))\mu(ze) \notag\\
    &= \sum_{zed}\min(\nu''(dze),\mu(dze))\notag\\
    &= \sum_{zed}\mu(dze)-\sum_{zed}(\mu(dze)-\nu''(dze))\knuth{\mu(dze)>\nu''(dze)}\notag\\
    &= 1-\TV(\nu'',\mu)\notag\\
    &\geq 1-\TV(\nu'',\nu')-\TV(\nu',\nu)- \TV(\nu,\mu)\notag\\
    &\geq 1-2\epsilon-\delta,
  \end{align}
  so $\tilde\mu_{ze}$ is as required by the definition.
\end{proof}

We remark for emphasis that as noted in the proof of the
lemma, the construction of $\nu'$ satisfies that
$\nu'[ZE]=\nu[ZE]$, where $\nu$ is the witness of the assumption in
the second part. 

The smoothness parameter for smooth conditional max-prob composes
well with strong extractors.  See also Ref.~\cite{konig:2008}, Prop. 1.

\begin{lemma}\label{lm:composition}
  Suppose that the distribution $\mu$ of $DZE$ satisfies
  $P^{{\rm u},\epse}_{\max,\mu}(D|ZE)\leq 2^{-\sige}$, and $S$ is a uniform
  and independent seed string with respect to $\mu$. If
  $(n=|D|,l=|S|,\sige , \sigma, \epsx )$ satisfies the extractor
  constraints, then $\TV(\mu[\cE(D,S;\sige ,\sigma,\epsx
  )SZE],\Unif\otimes\mu[ZE])\leq\epsx +\epse $.
\end{lemma}

\begin{proof} 
  Let $\tilde\mu_{ze}$ be subprobability distributions witnessing
  $P^{{\rm u},\epse}_{\max,\mu}(D|ZE)\leq 2^{-\sige}$ as required by
  the definition, and define $\nu$ as we did at the beginning of the
  proof of the first part of Lem.~\ref{lm:pmax_avpmax}, with
  $p=2^{-\sige}$ and $\epsilon=\epse$.  Then $\nu[D|ZE]\leq p$,
    $\nu[ZE]=\mu[ZE]$, and by Eq.~\ref{eq:tv_lemma6},
    $\TV(\nu,\mu)\leq\epse$.  According to the extractor
    specification, for all $ze$, we have
    $\TV(\nu[\cE(D,S)S|ze],\Unif)\leq\epsilon_{x}$. Since
    $\nu[ZE]=\mu[ZE]$, we can apply Eq.~\ref{eq:tv_samemarg} to get
  \begin{align}
    \TV(\nu[\cE(D,S)SZE],\Unif\otimes\mu[ZE]) 
    &= \sum_{ze}\TV(\nu[\cE(D,S)S|ze],\Unif)\mu(ze)\notag\\
    &\leq \sum_{ze}\epsilon_{x}\mu(ze)=\epsilon_{x}.
  \end{align}  
  By the data-processing inequality,
  $\TV(\nu[\cE(D,S)SZE],\mu[\cE(D,S)SZE])\leq \TV(\nu,\mu)\leq \epse$.
  The result now follows by the triangle inequality.
\end{proof}

By means of the third part of
Lem.~\ref{lm:pmax_avpmax} we can also compose smooth average
conditional max-prob with a strong extractor.

\begin{lemma}\label{lm:composition2}
  Suppose that the distribution $\mu$ of $DZE$ satisfies
  $P^{\epse}_{\max,\mu}(D|ZE)\leq 2^{-\sige}\delta$, and $S$ is a uniform
  and independent seed string with respect to $\mu$. If
  $(n=|D|,l=|S|,\sige , \sigma, \epsx )$ satisfies the extractor
  constraints, then $\TV(\mu[\cE(D,S;\sige ,\sigma,\epsx
  )SZE],\Unif\otimes\mu[ZE])\leq \epsx +2\epse+\delta$.
\end{lemma}

\begin{proof}
  From the last part of Lem.~\ref{lm:pmax_avpmax}, 
  we have $P^{{\rm u},2\epse+\delta}_{\max,\mu}(D|ZE)\leq 2^{-\sige}$.
  The result now follows from Lem.~\ref{lm:composition}.
\end{proof}

\subsection{Randomness Generation Protocols}

A randomness generation protocol $\cP$ is parameterized by $\sigma$,
the requested number of uniform bits and $\epsilon$, the protocol
error bound, defined as the $\TV$ distance from uniform of these
bits. Minimally, the output consists of a string of $\sigma$ bits. In
addition, the output can contain an additional string containing other
random bits that may have been used internally for seeding the
extractor or choosing settings, in case these bits can be used
elsewhere.  Further, the protocol may output a flag indicating success
or failure.  If the probability of success can be less than $1$,
the protocol may require as input a minimum acceptable probability
of success $\kappa$. We write $\cP =
(\cP_{X},\cP_{S},\cP_{P})$, where $\cP_{X}$ is the principal
$\sigma$-bit output, $\cP_{S}$ is the possibly empty but potentially
reusable string of seed and other random bits, and $\cP_{P}$ is the
success flag, which is $0$ (``fail'') or $1$ (``pass''). Here, we have
suppressed the arguments $\sigma$ and $\epsilon$ (or $\sigma$,
$\epsilon$ and $\kappa$) of $\cP,\cP_{X},\cP_{S},$ and $\cP_{P}$,
since they are clear from context.  The outputs of $\cP$ are treated
as RVs, jointly distributed with $E$ and any other RVs relevant to the
situation.  In this context, we constrain the joint distribution of
the RVs according to a model $\cH$.  When $\cP_{S}$ is non-empty, we
assume that the marginal distribution $\mu[\cP_{S}]$ is known to
the user.

The property that a protocol output satisfies the request is
called \emph{soundness}. We formulate soundness so that the $\TV$ distance
is conditional on success and $\epsilon$ absorbs the probability of
success criterion in a sense to be explained.

\begin{definition}
  The randomness generation protocol $\cP$ is
  \emph{$(\sigma,\epsilon)$-sound with respect to $E'$ and model $\cH$} if
  for all $\mu\in\cH$, $|\cP_{X}|=\sigma$ and there is a distribution
  $\nu_{E'}$ of $E'$ such that
  \begin{equation}
    \TV(\mu[\cP_{X}\cP_{S}E'|\cP_{P}=1],\Unif_{\cP_{X}}
    \otimes\mu[\cP_{S}]\otimes\nu_{E'})
    \Prob(\cP_{P}=1)
    \leq\epsilon.
    \label{eq:rngsoundness}
  \end{equation}
\end{definition}

We normally use $E'=ZE$, but $E'=E$ is an option if the settings
choices are private. The definition
ensures that $\cP$ is nearly indistinguishable, namely within $\TV$
distance $\epsilon$, from an \emph{ideal} protocol with the same
probability of success, namely such a protocol for which $\cP_{X}$
is perfectly uniform and independent of other variables conditional on success.

An alternative definition of soundness is to make the minimum
acceptable probability of success $\kappa$ explicit and require that
either $\Prob(\cP_{P}=1)<\kappa$ or the conditional $\TV$ distance
in Eq.~\ref{eq:rngsoundness} is bounded by $\epsilon$.  We refer to
a protocol satisfying this property as
$(\sigma,\epsilon,\kappa)$-sound.  The two definitions are closely
related:

\begin{lemma}\label{lm:smallkappa_or}
  Fix $E'$ and $\cH$.  Let $\kappa>0$ and $\epsilon>0$. Suppose
  that $\cP$ is a $(\sigma,\epsilon\kappa)$-sound randomness
  generation protocol, then $\cP$ is
  $(\sigma,\epsilon,\kappa)$-sound. If $\cP$ is
  $(\sigma,\epsilon,\kappa)$-sound, then $\cP$ is
  $(\sigma,\max(\epsilon,\kappa))$-sound.
\end{lemma}

\begin{proof}
  Suppose that $\cP$ is $(\sigma,\epsilon\kappa)$-sound.
  Consider any $\mu\in\cH$ and let $\nu_{E'}$ be the 
  distribution of $E'$ in Eq.~\ref{eq:rngsoundness}.
  If $\Prob(\cP_{P}=1)\ge\kappa$, then the desired $\TV$ distance is at most
  $\epsilon\kappa/\kappa=\epsilon$. Since $\mu\in\cH$ is arbitrary,
  $\cP$ is $(\sigma,\epsilon,\kappa)$-sound.

  For the other direction, again consider any $\mu\in\cH$. If
  $\Prob(\cP_{P}=1)<\kappa$, then the left-hand side
  of Eq.~\ref{eq:rngsoundness} is at most $\kappa$, because the $\TV$
  distance is bounded by $1$. If $\Prob(\cP_{P}=1)\geq\kappa$,
  let $\nu_{E'}$ be such that the $\TV$ distance in Eq.~\ref{eq:rngsoundness}
  is bounded by $\epsilon$. Then the full left-hand side
  of this equation is bounded by $\epsilon\Prob(\cP_{P}=1)\leq\epsilon$.
  The conclusion follows.  
\end{proof}

We prefer to use $(\sigma,\epsilon)$-soundness, since this is
often more informative than $(\sigma,\epsilon,\kappa)$-soundness,
allowing for flexibility in how the protocol is used without looking
inside the soundness proof. However, in both cases, soundness proofs
may contain additional information about the relationship between
the probability of success and the error bound that is useful when
implementing the protocol in a larger context. For example, the
soundness established in Thm.~\ref{thm:prot_chainlemmas} establishes
a relationship between relevant error bounds that does not match
either definition, but readily implies forms of either one.

\Pc{The version of soundness in
  Ref.~\cite{arnon-friedman:2018} is closely related to
  $(\sigma,\epsilon)$-soundness.}

In studies of specific randomness generation protocols, an important
consideration is \emph{completeness}, which requires that the actual
probability of success of the protocol is non-trivially large,
preferably exponentially close to $1$ with respect to relevant
resource parameters. Completeness is readily satisfied by our
  protocols for typical Bell-test configurations.

The actual probability of success should be distinguished from any
relevant minimum acceptable probability of success in view of the
discussion above.  Experimental implementations so far demonstrate
that success probabilities are acceptable, but not arbitrarily close
to $1$. As discussed in the introduction, theory shows that there are
randomness generation protocols using quantum systems that can achieve
high success probabilities.  Here we emphasize protocols for which the
success probability is $1$ by allowing for injection of banked
randomness when insufficient randomness is available from the allotted
resources.  Assuming completeness, we can also take advantage of the
ability to stop the protocol only when enough randomness is generated,
but since this requires care when extracting near-uniform bits from
potentially long strings, we do not develop this approach further
here.

\subsection{Test Factors and Test Supermartingales}
\label{subsec:supermart}

\begin{definition}
  A \emph{test supermartingale~\cite{shafer:qc2009a} with respect to a
    stochastic sequence $\Sfnt{R}$ and model $\cH$} is a stochastic
  sequence $\Sfnt{T}=(T_{i})_{i=0}^{N}$ with the properties that
  $T_{0}=1$, for all $i$ $T_{i}\ge 0$, $T_{i}$ is determined by
  $\Sfnt{R}_{\leq i}$ and the governing distribution, and for all
  distributions in $\cH$, $\Exp(T_{i+1}|\Sfnt{R}_{\le i})\leq T_{i}$.
  The ratios $F_{i}=T_{i}/T_{i-1}$ with $F_{i}=1$ if $T_{i-1}=0$ are
  called the \emph{test factors of $\Sfnt{T}$}.
\end{definition}

Here $\Sfnt{R}$ captures the relevant information that accumulates in
a sequence of trials. It does not need to be accessible to the
experimenter.  The $\sigma$-algebras induced by $\Sfnt{R}_{\leq i}$
define the nested sequence of $\sigma$-algebras used in more general
formulations.  Between trials $i$ and $i+1$, the sequence
$\Sfnt{R}_{\le i}$ is called the past. In the definition, we allow for
$T_{i}$ to depend on the governing distribution $\mu$.  With this, for
a given $\mu$, $T_{i}$ is a function of $\Sfnt{R}_{\leq i}$.  Below,
when stating that RVs are determined, we implicitly include the
possibility of dependence on $\mu$ without mention.  The
$\mu$-dependence can arise through expressions such as
$\Exp_{\mu}(G|\Sfnt{R}_{\le i})$ for some $G$, which is determined by
$\Sfnt{R}_{\leq i}$ given $\mu$.  One way to formalize this is to
consider $\mu$-parameterized families of RVs. We do not make this
explicit and simply allow for our RVs to be implicitly parameterized
by $\mu$. We note that the governing distribution in a given
experiment or situation is fixed but usually unknown with most of its
features inaccessible. As a result, many RVs used in mathematical
arguments cannot be observed even in principle. Nevertheless, they
play important roles in establishing relationships between observed
and inferred quantities.

Defining $F_{i}=1$ when $T_{i-1}=0$ makes sense because given
$\{T_{i-1}=0\}$, we have $\{T_{i}=0\}$ with probability $1$.  The
sequence $\Sfnt{F}$ satisfies the conditions that for all $i$, 1)
$F_{i}\geq 0$, 2) $F_{i}$ is determined by $\Sfnt{R}_{\le i}$, and 3)
for all distributions in $\cH$, $\Exp(F_{i+1}|\Sfnt{R}_{\le i})\leq 1$.  
We can define test supermartingales in terms of such sequences: Let $\Sfnt{F}$ be a
stochastic sequence satisfying the three conditions.  Then the
stochastic sequence with members $T_0=1$ and $T_{i}=\prod_{1\leq j\leq
  i}F_{j}$ for $i\geq1$ is a test supermartingale. It suffices to
check that $\Exp(T_{i+1}|\Sfnt{R}_{\leq i}) \leq T_{i}$. This follows
from
\begin{equation}
  \Exp(T_{i+1}|\Sfnt{R}_{\leq i})=
  \Exp(F_{i+1}T_{i}|\Sfnt{R}_{\leq i})
  = \Exp(F_{i+1}|\Sfnt{R}_{\leq i})T_{i}\leq T_{i},
\end{equation}
where we pulled out the determined quantity $T_{i}$ from the
conditional expectation.  In this work, we construct test
supermartingales from sequences $\Sfnt{F}$ with the above
properties. We refer to any such sequence as a sequence of test
factors, without necessarily making the associated test
supermartingale explicit. We extend the terminology by calling an RV
$F$ a test factor with respect to $\cH$ if $F\geq 0$ and
$\Exp(F)\leq 1$ for all distributions in $\cH$.

For an overview of test supermartingales and their properties, see
Ref.~\cite{shafer:qc2009a}. The notion of test supermartingales
and proofs of their basic properties are due to
Ville~\cite{ville:qc1939a} in the same work that introduced the
notion of martingales.  The name ``test supermartingale''
appears to have been introduced in Ref.~\cite{shafer:qc2009a}. Test
supermartingales play an important theoretical role in proving many
results in martingale theory, including that of proving tail bounds
for large classes of martingales. They have been studied and applied
to Bell tests~\cite{zhang:2011,zhang:2013,christensen:qc2015a}.

The definition implies that for a test supermartingale $\Sfnt{T}$, for
all $n$, $\Exp(T_{n})\leq 1$. This follows inductively from
$\Exp(T_{i+1})=\Exp(\Exp(T_{i+1}|\Sfnt{R}_{\le i}))\le \Exp(T_i)$ and
$T_0=1$. An application of Markov's inequality shows that for all
$\epsilon>0$,
\begin{equation}
  \Prob(T_{n} \geq 1/\epsilon)\leq \epsilon.
  \label{eq:testmart_markov}
\end{equation}
Thus, a large final value $t=T_{n}$ of the test supermartingale is
evidence against $\cH$ in a hypothesis test with $\cH$ as the
(composite) null hypothesis.  Specifically, the RV $1/T$ is a
$p$-value bound against $\cH$, where in general, the RV $U$ is a
$p$-value bound against $\cH$ if for all distributions in $\cH$,
$\Prob(U\leq\epsilon)\leq\epsilon$.  This result can be
strengthened as follows. Define $T^{*} = \max_{1\leq i\leq n}T_{i}$.

\begin{theorem}\label{thm:testmart_max}
  For a test supermartingale $\Sfnt{T}$ with respect to $\Sfnt{R}$ and $\cH$,
  the RV $1/T^{*}$ is a $p$-value bound against $\cH$.
\end{theorem}

This theorem is due to
Ville~\cite{ville:qc1939a} Ch. V.2.2 (pgs. 97--98).

\begin{proof}
  The theorem follows from Doob's maximal inequalities. The particular
  inequality we need is normally stated for a non-negative
  submartingale $\Sfnt{T}'$ in the form
  \begin{equation}
    \Prob(\max_{1\leq i\leq n}T'_{i}\geq \lambda)\leq \Exp(T'_{n})/\lambda,
  \end{equation}
  where $\lambda=1/\epsilon$ for our purposes. Statements and proofs
  of this inequality are readily found online. A textbook treatment is
  in Ref.~\cite{dellacherie:qc2011a} Ch. V, Cor. 22.  To apply the
  above maximal inequality to the test supermartingale $\Sfnt{T}$, let
  $\bar T_{i+1}=\Exp(T_{i+1}|\Sfnt{R}_{\le i})$. Note that $\bar
  T_{i+1}$ is determined by $\Sfnt{R}_{\le i}$ and by the definition of
  supermartingales, $\bar T_{i+1}\leq T_{i}$.  Define $T'_{0}=1$ and
  $T'_{i} = \prod_{1\leq j\leq i}T_{j}/\bar T_{j}$.  By pulling out
  determined quantities from the conditional expectation, we get
  \begin{align}
    \Exp(T'_{i+1}|\Sfnt{R}_{\le i})&= T'_{i}\Exp(T_{i+1}/\bar T_{i+1}|\Sfnt{R}_{\le i})\notag\\
    &=    T'_{i}\Exp(T_{i+1}|\Sfnt{R}_{\le i})/\bar T_{i+1} = T'_{i}. 
  \end{align}
  Hence $\Sfnt{T}'$ is a test martingale. In particular, it is a
  non-negative submartingale with $\Exp(T'_{n})=1$.  We claim that
  $T'_{i}\geq T_{i}$.  This holds for $i=0$. For a proof by induction,
  compute $T'_{i+1}=T'_{i}T_{i+1}/\bar T_{i+1}\geq T_{i} T_{i+1}/\bar
  T_{i+1}\ge T_iT_{i+1}/T_i= T_{i+1}$. This gives
  $\max_{1\leq i\leq n}T'_{i}\geq T^{*}$, so the event
  $\{T^{*}\geq\lambda\}$ implies the event in the maximal inequality,
  and monotonicity of probabilities implies the theorem.
\end{proof}

One can produce a test supermartingale adaptively by determining the
test factors $F_{i+1}$ to be used at the next trial.  If the $i$'th
trial's data is $R_{i}$, including any incidental information
obtained, $F_{i+1}$ is expressed as a function of $\Sfnt{R}_{\le i}$
and data from the $i+1$'th trial (a \emph{past-parameterized} function of
$R_{i+1}$), and constructed to satisfy $F_{i+1}\geq 0$ and
$\Exp(F_{i+1}|\Sfnt{R}_{\le i})\leq 1$ for any distribution in 
the model $\cH$.
Note that in between trials, we can effectively stop the experiment by
assigning all future $F_{i+1}=1$ conditional on the past. This is
equivalent to constructing the stopped process relative to a stopping
rule. This argument also shows that the stopped process is still a
test supermartingale.  Here we consider only bounded length test
supermartingales (with very large bound $N$ if necessary), so are not
concerned with questions of convergence as $N\rightarrow\infty$.  But
we note that since $\Exp(|F_{i}|)=\Exp(F_{i})\leq 1$ for all $i$,
Doob's martingale convergence theorem applies, and
$\lim_{i\rightarrow\infty}F_{i}$ exists almost surely and is
integrable. Furthermore, this limit also has expectation at most $1$.
See, for example, Ref.~\cite{dellacherie:qc2011a} Ch. V, Thm. 28. 

More generally, we use test supermartingales for estimating lower
bounds on products of positive stochastic sequences $\Sfnt{G}$. Such
lower bounds are associated with unbounded-above confidence
intervals. We need the following definition:

\begin{definition} 
  Let $U,V,X$ be RVs and $1\ge\epsilon\ge 0$.  $I=[U,V]$ is a
  (conservative) \emph{confidence interval for $X$ at level $\epsilon$
    with respect to $\cH$} if for all distributions in $\cH$ 
    we have
    $\Prob(U\leq X\leq V)\geq 1-\epsilon$. The quantity
    $\Prob(U\leq X\leq V)$ is called the \emph{coverage probability}.
\end{definition}

As noted above, the RVs $U$, $V$ and $X$ may be $\mu$-dependent.  For
textbook examples of confidence intervals, $X$ is a parameter
determined by $\mu$, and $U$ and $V$ are calculated
from an estimation error. We need the full generality of the
definition above.  The quantity $\epsilon$ in the definition is a
significance level, which corresponds to a confidence level of
$(1-\epsilon)$.

\begin{lemma}
  Let $\Sfnt{F}$ and $\Sfnt{G}$ be two stochastic sequences with
  $F_{i}\in[0,\infty)$, $G_{i}\in (0,\infty]$, and $F_{i}$ and
  $G_{i}$ determined by $\Sfnt{R}_{\leq i}$.  Define $T_{0}=1$,
  $T_{i}=\prod_{1\leq j \leq i}F_{i}$ and $U_{0}=1$,
  $U_{i}=\prod_{1\leq j\leq i}G_{i}$, and suppose that for all
  $\mu\in\cH$, $\Exp(F_{i+1}/G_{i+1}|\Sfnt{R}_{\le i}) \leq 1$.
  Then $[T_{n}\epsilon,\infty)$ is a confidence interval for $U_{n}$
  at level $\epsilon$ with respect to $\cH$.  If $U_{i}$ is
  monotone non-decreasing, then $[T^{*}\epsilon,\infty)$ is a
  confidence interval for $U_{n}$ at level $\epsilon$ with respect
  to $\cH$.  
\end{lemma}

\begin{proof}
  The assumptions imply that the $F_{i+1}/G_{i+1}$ form a sequence of
  test factors with respect to $\cH$ and generate the test
  supermartingale $\Sfnt{T}/\Sfnt{U}$, where division in this
  expression is term-by-term.  Therefore,
  \begin{equation}
    \Prob(T_{n}\epsilon \geq U_{n}) = \Prob(T_{n}/U_{n}\geq 1/\epsilon) \leq \epsilon,
    \label{eq:itestmart_markov}
  \end{equation}
  So $[T_{n}\epsilon,\infty)$ is a confidence
  interval for $U_{n}$ at level $\epsilon$.  

  If $U_{i}$ is monotone non-decreasing, observe that the event
  $\{\max_{1\leq i\leq n}T_{i}\epsilon \geq U_{n}\}$ is the same as
  $\{(\max_{1\leq i\leq n}T_{i})/U_{n}\geq 1/\epsilon\}$. This event in turn
  implies $\{\max_{1\leq i\leq n}T_{i}/U_{i}\geq 1/\epsilon\}$, because 
  $U_{n}$ being non-decreasing implies that for all $i$,
  $T_{i}/U_{n}\leq T_{i}/U_{i}$.  Applying Thm.~\ref{thm:testmart_max}
  we obtain
  \begin{equation}
    \Prob(T^{*}\epsilon\geq U_{n})=\Prob(T^{*}/U_{n}\geq 1/\epsilon)
    \leq \Prob(\max_{1\leq i\leq n}T_{i}/U_{i}\geq 1/\epsilon)\leq \epsilon.
    \label{eq:itestmart_max}
  \end{equation}
  It follows that $[T^{*}\epsilon,\infty)$ is a confidence interval
  for $U_{n}$ at level $\epsilon$. 
\end{proof}

\subsection{From Bell Functions to Test Factors}
\label{subsec:bell_to_test}

Consider a specific trial subject to settings and non-signaling
constraints, where the settings distribution is fixed and known. To
simplify notation, we omit trial indices and conditioning on the
past. Let $\nu=\mu[Z]$ be the settings distribution. A \emph{Bell function}
$B$ maps settings and outcomes to real values and satisfies
$\Exp(B(CZ))\leq 0$ for all \emph{local realistic} (LR) distributions
with the given settings distribution.    The set of LR distributions given that the
settings distribution is $\nu$ consists of mixtures of non-signaling
distributions with outcomes determined by the settings.  These are
distributions such that $C=AB$ is a function of $Z=XY$ for which $A$
does not depend on $Y$ and $B$ does not depend on $X$.  See
Sect.~\ref{sec:apps}, Eq.~\ref{eq:lrf} for more detail.  The
expression $\Exp(B(CZ))\leq 0$ is closely related to traditional Bell
inequalities, which may be expressed in the \emph{conditional form}
$\sum_{z}\Exp(B'(Cz)|z)\leq 0$.  Provided $\nu(z)= 0$ implies
$B'(cz)=0$ for all $c$, any Bell inequality in this form can be
converted to a Bell function by defining
$B(cz)=B'(cz)/\nu(z)$. Conversely, a Bell function $B$ for settings
distribution $\nu$ determines a Bell inequality by defining
$B'(cz)=B(cz)\nu(z)$.  Bell inequalities apply to LR distributions
independent of what the settings probabilities are, and these settings
probabilities are then considered to be the free choice of the
experimenter.  We do not use this perspective here.

Let $-l$ be a lower bound for the Bell function $B$ with $l> 0$.  Then
$F=(B+l)/l$ is a test factor with respect to LR distributions. Such
test factors can provide an effective method for rejecting local
realism (LR) at high significance levels by use of
Eq.~\ref{eq:testmart_markov}.  As an example, consider the $(2,2,2)$
Bell-test configuration, with the uniform settings distribution,
$\nu=1/4$.  The ranges of the settings $X$, $Y$ and the outcomes $A$,
$B$ are $\{0,1\}$.  The following is a Bell function:
\begin{equation}
  B(xyab) = (\knuth{xy=11}-\knuth{xy=01}-\knuth{xy=10}-\knuth{xy=00})|a-b|.
  \label{eq:chshvariant}
\end{equation}
The inequality $\Exp(B(CZ))\leq 0$ is equivalent to the well-known
Clauser-Horne-Shimony-Holt (CHSH) inequality~\cite{clauser:qc1969a}.
We give the function in this form for simplicity and
because one way to verify that $B$ is a Bell function is to note that
$d(a,b)=|a-b|$ satisfies the triangle inequality, so this Bell
function belongs to the class of distance-based Bell functions. See
Refs.~\cite{schumacher:qc1991a,knill:qc2014a}.  Since for all
arguments, exactly one of the expressions inside $\knuth{\ldots}$ is true,
the minimum value of $B$ is $-1$.  Thus $B+1$ is a test factor. More
generally, $1+\lambda B$ is a test factor for any $0\leq\lambda\leq
1$~\cite{zhang:2013}. By optimizing these test factors, asymptotically 
optimal statistical strength (expected $-\log(p\text{-value})$ per trial) for
rejecting LR can be achieved~\cite{zhang:2011,zhang:2013}. 

While probability estimation uses probability estimation factors
(PEFs) to assign numerical values to experimental outcomes, it does
not require Bell functions or even Bell-test configurations. The
connection of Bell functions to test factors serves as an instructive
example and can yield PEFs useful for probability estimation as
witnessed implicitly by Ref.~\cite{bierhorst:qc2017a}.  See also
Sect.~\ref{sec:maxe_pefs}.  In general, optimal PEFs cannot be
constructed from Bell functions. Therefore, in this work we prefer to
directly optimize PEFs without referencing Bell functions, see
Sects.~\ref{sec:prots} and~\ref{sec:pef_constructions}.

\section{Probability Estimation and Randomness Extraction}
\label{sec:probest}

\subsection{Probability Estimation}

Consider RVs $C$, $Z$ and $E$ whose joint distribution is constrained
to be in the model $\cH$. In an experiment, we observe $Z$ and $C$,
but not $E$. For this section, it is not necessary to structure
the RVs as stochastic sequences, so we use $D$, $C$, $Z$ and $R$ in
place of $\Sfnt{D}$, $\Sfnt{C}$, $\Sfnt{Z}$ and $\Sfnt{R}$, but
follow the conventions introduced in
Sect.~\ref{subsec:setup_belltests}.  In particular, we assume that $D$
is determined by $C$, and $C$ is determined by $R$.  There is no
loss of generality in this assumption, it simply allows us to omit $C$
(or $D$) as arguments to functions and in conditioners when $R$ is
already present.  The distinction between $C$ and $R$ only appears
when we consider time-ordered trials, so for the remainder of this
section, only $R$ is used.

We focus on probability estimates with coverage probabilities that do
not depend on $E$, formalized as follows.
\begin{definition} Let $1\geq \epsilon \geq 0$.
  The function $F:\Rng(RZ)\rightarrow [0,\infty)$ is a
  \emph{level-$\epsilon$ $E$-uniform probability estimator for $\cH$}
  ($\epsilon$-$\UPE$ or with specifics,
  $\epsilon$-$\UPE(D{:}R|Z;E,\cH)$) if for all $e$ and
  distributions $\mu$ in $\cH$, we have
  $\mu(F(RZ)\geq\mu(D|Ze)|e)\geq 1-\epsilon$.  We omit specifics such
  as $\cH$ if they are clear from context.  The function
  $\cU:\Rng(RZ)\times[0,1]\rightarrow [0,\infty)$ is a $\UPE$ if it is
  monotone non-decreasing in the second argument (the confidence
  level) and for each $\epsilon$, $\cU(RZ,1-\epsilon)$ is an
  $\epsilon$-$\UPE$.
  \label{def_UPE}
\end{definition}
The level of a probability estimator relates to the smoothness
parameter for smooth min-entropy via the relationships 
established below in Sect.~\ref{sec:petominent}.  We also use
the term \emph{error bound} to refer to the level of a probability
estimator.  The first condition on $\cU$ ensures that the
confidence upper bounds that it provides are non-decreasing with
confidence level, so that the corresponding confidence intervals are
consistently nested. The second is the required minimum coverage
probability for confidence regions at a given confidence level. Our
inclusion of the random variable $E$ here and in the next
definition is in a sense redundant: Uniformity of the estimator
means that we could instead have considered the model $\cH'$ of
distributions on $Z$ and $R$ consisting of distributions $\mu[RZ|e]$
over all $\mu\in\cH$ and all $e$.  We refer to $E$ explicitly because
of the role played by external entities in this work. 

For the most general results, we need a softening of the above
definition.  The softening adds smoothing and averaging similar to
what is needed to define smooth average conditional min-entropy.

\begin{definition} Let $1\geq \epsilon \geq 0$.
  The function $F:\Rng(RZ)\rightarrow[0,\infty)$ is a
  \emph{level-$\epsilon$ $E$-uniform soft probability estimator for
    $\cH$} ($\epsilon$-soft $\UPE$, or with specifics, $\epsilon$-soft
  $\UPE(D{:}R|Z;E,\cH)$) if for all $e$ and distributions $\mu$
  in $\cH$ the following hold: 1) There exist subprobability
  distributions $\tilde\mu_{z}\le \mu[DR|ze]$ of $DR$ of total weight
  $\sum_{z}w(\tilde\mu_{z})\mu(z|e)\geq 1-\epsilon$; 2) there exists a non-negative
  function $q(R|DZ)$ of $DRZ$ such that 
  $\sum_{r}q(r|dz)=1$ for all $dz$, and
  \begin{equation}
    \tilde\mu_{z}(dr) \leq F(rz) q(r|dz),\label{eq:supes}
  \end{equation}
  for all $drz$. The function
  $\cU:\Rng(RZ)\times[0,1]\rightarrow [0,\infty)$ is a soft $\UPE$ if
  it is monotone non-decreasing in the second argument, and for each
  $\epsilon$, $\cU(RZ,1-\epsilon)$  is an $\epsilon$-soft $\UPE$.
  \label{def_soft_UPE}
\end{definition}
Note that in the definition of soft probability estimators, we have
chosen to leave the dependence of $\tilde\mu_{z}$ and $q$ on $e$
implicit.

The notion of a soft $\UPE$ is weaker than that of a $\UPE$, see the
next lemma. Note that in this definition we explicitly consider the
joint distribution of $DR$, even though we assume that $D$ is
determined by $R$. Considering the joint distribution helps simplify
the notation for probabilities.  However, we take this assumption into
account by not explicitly including $D$ as an argument to $F$.  In
this definition, for any given $e$, $F$ estimates the probability of $d$
conditional on $z$ but requires $r$. One can interpret
$F(rz)q(r|dz)$ as an implicit and distribution-dependent estimate
of the probability of $dr$ given $z$, where $q(r|dz)$ accounts for the
probability of $r$ conditional on $dz$. 

\begin{lemma}\label{lm:eupe_esupe}
  If $F$ is an $\epsilon$-$\UPE$, then it is an $\epsilon$-soft $\UPE$.
\end{lemma}

\begin{proof}
  For any given $e$,
  define $\tilde\mu_{z}(dr)=\mu(dr|ze)\knuth{F(rz)\ge \mu(d|ze)}$.
  Then by the definition of $\epsilon$-$\UPE$s,
  $\sum_{z}w(\tilde\mu_{z})\mu(z|e)\geq 1-\epsilon$.  Let
  $q(r|dz)=\mu(r|dze)$. It suffices to check Eq.~\ref{eq:supes} for elements of
  $\{drz:F(rz)\ge\mu(d|ze)\}$:
  \begin{align}
    \tilde\mu_{z}(dr) &= \mu(dr|ze)\notag\\
    &=\mu(r|dze)\mu(d|ze)\notag\\
    &\leq q(r|dz)F(rz).
  \end{align}
\end{proof}

\subsection{From Probability Estimation to Min-Entropy}
\label{sec:petominent}

The next lemma shows that $\epsilon$-soft $\UPE$s 
with constant upper bounds certify smooth min-entropy.

\begin{lemma}\label{lm:esupe_const_esmaxprob}
  Suppose that $F$ is an $\epsilon$-soft $\UPE(D{:}R|Z;E,\cH)$ such
  that $F\leq p$ for some constant $p$.  Then for each
  $e$ and $\mu\in\cH$, $P^{{\rm u},\epsilon}_{\max,\mu[DZ|e]}(D|Z)\leq
  p$. Consequently, $P^{{\rm u},\epsilon}_{\max,\mu[DZE]}(D|ZE)\leq
  p$.
\end{lemma}

\begin{proof}
  Fix $e$ and let $\tilde\mu_{z}(DR)$ and $q(R|DZ)$ be as in the
  definition of soft $\UPE$s.  According to our conventions,
  $\tilde\mu_{z}[D](d)=\tilde\mu_{z}(d)=\sum_{r}\tilde\mu_{z}(dr)$.
  Therefore $\sum_{z}w(\tilde\mu_{z}[D])\mu(z|e)\geq 1-\epsilon$.  Since
  $\sum_{r}q(r|dze)=1$, for all $dz$ we have
  \begin{equation}
    \tilde\mu_{z}(d) \leq \sum_{r}F(rz)q(r|dz) \leq \sum_{r}pq(r|dz)=p.
  \end{equation}
  Further, $\tilde\mu_{z}(d)\leq\mu(d|ze)$ for all $dz$, as can be
    seen by summing the defining inequality for $\tilde\mu_{z}$ over
  $r$. It follows that the $\tilde\mu_{z}[D]$ witness that $P^{{\rm
      u},\epsilon}_{\max,\mu[DZ|e]}(D|Z)\leq p$.  For the last
  statement, apply Lem.~\ref{lm:pmax_noe}.
\end{proof}

A weaker relationship holds for general soft $\UPE$s.

\begin{lemma}\label{lm:esupe_fail_esmaxprob}
  Suppose that $F$ is an $\epsilon$-soft $\UPE(D{:}R|Z;E,\cH)$,
  $\mu\in\cH$, $p\geq 1/|\Rng(D)|$, and let $\kappa=\Prob(F\leq
  p)$. Then $P^{\epsilon/\kappa}_{\max,\mu[DZE|F\leq p]}(D|ZE) \leq
  p/\kappa$.
\end{lemma}

\begin{proof}
  Let $\kappa_{e}=\Prob(F\leq p|e)$.
  Below we show that for all values $e$ of $E$,
  $P^{\epsilon/\kappa_{e}}_{\max,\mu[DZ|e,F\leq p]}(D|Z)\leq p/\kappa_{e}$.
  Once this is shown, we can use
  \begin{equation}
    \sum_{e}\frac{1}{\kappa_{e}}\mu(e|F\leq p) 
    = \sum_{e}\frac{1}{\Prob(F\leq p|e)}\mu(e|F\leq p)
    = \sum_{e}\frac{\mu(e)}{\mu(F\leq p)} = 1/\kappa,
  \end{equation}
  and Lem.~\ref{lm:pmax_noe} to complete the proof.  For the
  remainder of the proof, $e$ is fixed, so we simplify the notation
  by universally conditioning on $\{E=e\}$ and omitting the explicit
  condition. Further, we omit $e$ from suffixes.  Thus
  $\kappa=\kappa_{e}$ from here on.  

  Let $\chi(RZ)=\knuth{F(RZ)\leq p}$.
  Then $\kappa=\Exp(\chi)$. Let $\kappa_{z}=\mu(F\leq p|z)$.  We
  have $\sum_{z}\kappa_{z}\mu(z)=\kappa$ and
  $\kappa_{z}=\mu(z|F\leq p)\kappa/\mu(z)$.

  Let $\tilde\mu_{z}$ and $q$ witness that $F$ is an $\epsilon$-soft
  $\UPE(D{:}R|Z;E,\cH)$.  Define $\tilde\nu(drz) =
  \tilde\mu_{z}(dr)\chi(rz)\mu(z)/\kappa$.  The weight of $\tilde\nu$
  satisfies
  \begin{align}
    w(\tilde\nu) &= \sum_{drz}\tilde\mu_{z}(dr)\chi(rz)\mu(z)/\kappa\notag\\
    &\leq \sum_{drz}\mu(drz)\chi(rz)/\kappa \notag\\
    & = \Prob(F\leq p)/\kappa = 1,\\
    w(\tilde\nu) &= 
    \sum_{drz}\mu(drz)\chi(rz)/\kappa 
    -\sum_{drz}(\mu(drz)-\tilde\mu_{z}(dr)\mu(z))\chi(rz)/\kappa \notag\\
    &=1-\sum_{drz}(\mu(drz)-\tilde\mu_{z}(dr)\mu(z))\chi(rz)/\kappa \notag\\
    &\geq 1-\sum_{drz}(\mu(drz)-\tilde\mu_{z}(dr)\mu(z))/\kappa \notag\\
    &\geq 1- (1-(1-\epsilon))/\kappa = 1-\epsilon/\kappa.
  \end{align}
  Thus $\tilde\nu$ is a subprobability distribution of weight at least
  $1-\epsilon/\kappa$.  We use it to construct the distribution $\nu$
  witnessing the conclusion of the lemma. For each $drz$ we bound
  \begin{equation}
    \tilde\nu(drz)/\mu(z|F\leq p) 
    = \tilde\mu_{z}(dr)\chi(rz)/\kappa_{z}
    \leq F(rz)q(r|dz)\chi(rz)/\kappa_{z}
    \leq p q(r|dz)/\kappa_{z},\label{eq:lm:esupe_fail_esmaxprob1}
  \end{equation}
  from which we get $\tilde\nu(dz)/\mu(z|F\leq p)\leq p/\kappa_{z}$ by
  summing over $r$.  Define $\tilde\nu[D|z](d) =
  \tilde\nu(dz)/\mu(z|F\leq p)$ and let $w_{z}=w(\tilde\nu[D|z])$,
  where this extends the conditional probability notation to the
  subprobability distribution $\tilde{\nu}$ with the understanding
  that the conditionals are with respect to $\mu$ given $\{F\leq
  p\}$. Applying the first step of
  Eq.~\ref{eq:lm:esupe_fail_esmaxprob1} and continuing from there, we
  have
  \begin{align}
    \tilde\nu[D|z](d) &= \sum_{r}\tilde\nu(drz)/\mu(z|F\leq p)\notag\\
    &=\sum_{r}\tilde\mu_{z}(dr)\chi(rz)/\kappa_{z}\notag\\
    &\leq \sum_{r}\mu(dr|z)\chi(rz)/\kappa_{z}\notag\\
    &= \sum_{r}\mu(dr|z)\knuth{F(rz)\leq p}/\kappa_{z}\notag\\
    &=\mu(d,F\leq p|z)/\mu(F\leq p|z) = \mu(d|z,F\leq p),  
  \end{align}
  which also implies that $w_{z}\leq 1$.  Since $\tilde\nu[D|z](d)\leq
    p/\kappa_z$ and $p/\kappa_{z}\geq p\geq 1/|\Rng(D)|$, we can apply
    Lem.~\ref{lm:tvsub} to obtain
  distributions $\nu_{z}$ of $D$ such that $\nu_{z}\geq
  \tilde\nu[D|z]$, $\nu_{z}\leq p/\kappa_{z}$, and
  $\TV(\nu_{z},\mu[D|z,F\leq p])\leq 1-w_{z}$.  Now we can define
  $\nu(dz)=\nu_{z}(d)\mu(z|F\leq p)$.  From Eq.~\ref{eq:tv_samemarg},
  \begin{align}
    \TV(\nu,\mu[DZ|F\leq p]) &= \sum_{z}\TV(\nu_{z},\mu[D|z,F\leq p])\mu(z|F\leq p)\notag\\
    &\leq \sum_{z}(1-w_{z})\mu(z|F\leq p)\notag\\
    &= 1-\sum_{z}w(\tilde\nu[D|z])\mu(z|F\leq p)\notag\\
    &= 1-\sum_{z}\sum_{dr}(\tilde\nu(drz)/\mu(z|F\leq p))\mu(z|F\leq p)\notag\\
    &= 1-w(\tilde\nu)\leq \epsilon/\kappa.\label{eq:lm:esupe_fail_esmaxprob_1}
  \end{align}
  For the average maximum probability of $\nu$, we get
  \begin{align}
    \sum_{z}\max_{d}\nu(d|z)\nu(z) &=
    \sum_{z}\max_{d}\nu_{z}(d)\mu(z|F\leq p)\notag\\
    &\leq p\sum_{z}\mu(z|F\leq p)/\kappa_{z}\notag\\
    &= p\sum_{z}\mu(z)/\kappa = p/\kappa,
  \end{align}
  which together with the argument at the beginning of the proof
  establishes the lemma.
\end{proof}

We remark that the witness $\nu$ constructed in the proof of the lemma
satisfies $\nu[ZE]=\mu[ZE|F\leq p]$.  We refer to $\nu$ in the next
section, where to emphasize that $\nu$ is constructed conditionally on
$\{F\leq p\}$ we write $\nu[DZE|F\leq p]$ instead of $\nu[DZE]$.

\subsection{Protocols}
\label{subsec:protocols}

Our goal is to construct probability estimators from test
supermartingales and use them in randomness generation by composition
with an extractor. The required supermartingales are introduced in
Sect.~\ref{sec:prots}. Here we give three ways in which probability
estimators can be composed with an extractor for randomness generation
protocols.  For the first, given a probability estimator, we can
estimate smooth min-entropy and compose with an extractor by chaining
Lem.~\ref{lm:esupe_fail_esmaxprob} with
Lem.~\ref{lm:composition2}. The second draws on banked randomness to
avoid failure. The third requires a linear extractor.  The first
protocol is given in the following theorem:

\begin{theorem}
  \label{thm:prot_chainlemmas}
  Let $F$ be an $\epse$-soft $\UPE(D{:}R|Z;E,\cH)$ with $\Prob(F\leq
  p\delta)=\kappa$, $D$ a bitstring of length $n$, and $p\delta\geq
  1/2^{n}$.  Write $\sige=-\log_{2}(p)$ and suppose that $\cE$ is a
  strong extractor such that $(n,l,\sige,\sigma,\epsx)$ satisfies the
  extractor constraints. Let $S$ be a length $l$ uniform and
  independent seed string. Abbreviate
  $\cE=\cE(D,S;\sige,\sigma,\epsx)$.  Then
  \begin{equation}
    \TV(\mu[\cE SZE|F\leq p\delta],\Unif_{\cE \cS}\otimes\mu[ZE|F\leq p\delta])\leq\epsx+(2\epse+\delta)/\kappa.
    \label{eq:thm:prot_chainlemmas}
  \end{equation}
  In particular, $\cP_{X}=\cE,\cP_{S}=S,\cP_{P}=\knuth{F\leq p\delta}$
  defines a $(\sigma,\epsx+2\epse+\delta)$-sound randomness generation
  protocol.
\end{theorem}

\begin{proof}
  By Lem.~\ref{lm:esupe_fail_esmaxprob},
  $P^{\epse/\kappa}_{\max,\mu[DZE|F\leq p\delta]}(D|ZE)\leq
  p\delta/\kappa$. We obtain Eq.~\ref{eq:thm:prot_chainlemmas} by
  applying Lem.~\ref{lm:composition2} with $\epse$ there
  replaced by $\epse/\kappa$ here and $\delta$ there with
  $\delta/\kappa$ here. The soundness statement
  follows by multiplying both sides by $\kappa$ in consideration of
  $\epsx\kappa\leq\epsx$.
\end{proof}

We can use the remarks after Lem.~\ref{lm:pmax_avpmax} and
Lem.~\ref{lm:esupe_fail_esmaxprob} to improve the factor of $2$ on
$\epse$ in the statement of the theorem. For this, let $\nu$ witness
that $P^{\epse/\kappa}_{\max,\mu[DZE|F\leq p]}(D|ZE)\leq
p\delta/\kappa$ according to Lem.~\ref{lm:esupe_fail_esmaxprob},
where from the remark after this lemma, $\nu[ZE|F\leq
p\delta]=\mu[ZE|F \leq p\delta]$. Then apply the second part of
Lem.~\ref{lm:pmax_avpmax} and the remark thereafter to obtain a
distribution $\nu'$ of $DZE$ with $\TV(\nu',\nu[DZE|F\leq
p\delta])\leq$ $\delta/\kappa$ and
$\nu'[ZE]=\nu[ZE|F\leq p\delta]=\mu[ZE|F\leq p\delta]$, where
$P^{{\rm u},0}_{\max,\nu'}(D|ZE)\leq p$. From
Lem.~\ref{lm:composition}, we have $\TV(\nu'[\cE(D,S;\sige,
\sigma,\epsx )SZE],\Unif_{\cE \cS}\otimes\nu'[ZE])\leq\epsx$.  The observation
now follows from $\nu'[ZE]=\mu[ZE|F\leq p\delta]$ and the triangle
and data-processing inequalities.  Namely, we get
\begin{equation}
  \TV(\mu[\cE SZE|F\leq p\delta],\Unif_{\cE \cS}\otimes\mu[ZE|F\leq p\delta])\leq\epsx+(\epse+\delta)/\kappa.
  \label{eq:thm:prot_chainlemmas_improved}
\end{equation}

The protocol performance and soundness proof of
Thm.~\ref{thm:prot_chainlemmas} is closely related to the performance
and proof of the Protocol Soundness Theorem in
Ref.~\cite{bierhorst:qc2017a}, which in turn are based on results of
Refs.~\cite{konig:2008,pironio:2013}. The parameters in these theorems
implicitly compensate for the fact that we do not require the
extractor to be classical-proof, see the comments on extractor
constraints in Sect.~\ref{sec:minentropy_extraction}.

The next protocol has no chance of failure but needs access to banked
randomness. The banked randomness cannot be randomness from previous
instances of the protocol involving the same devices. Let $\cU$ be a soft $\UPE(D{:}R|Z;E)$ for
$\cH$. Let $\sigma\in\nats_{+}$ and $\epsilon>0$ be the requested
number of bits and error bound.  Assume we have access to a source of
uniform bits $S$ that is independent of all other relevant RVs. For
banked randomness, we also have access to a second source of such
random bits $S_{b}$.  Let $\cE$ be a strong extractor.  We define a
randomness generation protocol $\cP(\sigma,\epsilon;\cU)$ by the
following steps:
\begin{description}[\compact]
\item[ $\cP(\sigma,\epsilon;\cU)$:] 
\item[]
  \begin{description}[\compact]
  \item[0.] Choose $n$, $\sige $, $\epse >0$,
    $\epsx >0$ and $l$ so that
    $(n+\sige ,l,\sige ,\sigma,\epsx )$ satisfies the extractor
    constraints and $\epse +\epsx =\epsilon$.

  \item[1.] Perform trials to obtain values $r$ and $z$ of $R$ and $Z$
    expressed as bit strings, with $|D|=n$. Let $d$ be the value taken
    by $D$.  (Recall that by our conventions, $D$ is determined by
    $R$.)

  \item[2.] Determine $p=\cU(rz,1-\epse )$ and let
    $k=\max(0,\lceil\sige -\log_{2}(1/p)\rceil)$. 

  \item[3.] Obtain values $(s)_{\le l}$ and $(s_{b})_{\le k}$ from $S_{\le l}$ and
    $(S_{b})_{\le k}$.

  \item[4.] Let $d'=d(s_{b})_{\le k}0^{\sige-k}$, where $0^{\sige-k}$
    is a string of zeros of length $\sige-k$.

  \item[5.] If the distribution of $Z$ is known and independent of
    $E$, let $z'=z$, otherwise let $z'$ be the empty string.   

  \item[6.] Output $\cP_{X}=\cE(d',(s)_{\le l},\sige ,\sigma,\epsx )$,
    $\cP_{S}=z'(s)_{\le l}$ and $\cP_{P}=1$.
  \end{description}
\end{description}

Note that we have left the parameter choices made in the first step
free, subject to the given constraints. They can be made to minimize
the expected amount of banked randomness needed given information on
device performance and resource bounds. It is important that the
choices be made before obtaining $r$, $z$ and the seed bits. That is,
they are conditionally independent of these RVs given the pre-protocol
past, where $\cH$ is satisfied conditionally on this past. The length
of $D$ must be fixed before performing trials, but this does not
preclude stopping trials early if an adaptive probability estimator is
used. In this case $D$ can be zero-filled to length $n$. The
requirement for fixed length $n$ is imposed by the extractor
properties: We do not have an extractor that can take advantage of
variable- and unbounded-length inputs. 

\begin{theorem}\label{thm:estext_nofail}
  The protocol $\cP(\sigma,\epsilon;\cU)$ is a
  $(\sigma,\epsilon)$-sound and complete randomness generation
  protocol with respect to $\cH$.
\end{theorem}

\begin{proof}
  The protocol has zero probability of failure, so completeness is
  immediate.  But note that for the protocol to be non-trivial, the
  expected number of banked random bits used should be small
  compared to $\sigma$.

  Let $K$,  $D'$ and $Z'$ be the RVs corresponding to the values
  $k$, $d'$ and $z'$ constructed in the protocol. Given the
  parameter choices that are made first, $K$ is determined by $R$ and
  $Z$.  Define $S_{K}=(S_{b})_{\leq K}0^{\sige-K}$ and
  $R'=RS_{K}$. Taking advantage of uniformity in $e$, we fix $e$ and
  implicitly condition everything on $\{E=e\}$.  Let $\tilde\mu_{z}$
  and $q$ be the functions witnessing that $p=\cU(rz,1-\epse )$ is
  an $\epse $-soft $\UPE$.  
  Now let $\tilde\nu_{z}$ be the family of subprobability distributions on $D'R'$ defined by
  $\tilde\nu_{z}(d'r')=\tilde\mu_{z}(dr)\mu(s_{k}|drz) =
  \tilde\mu_{z}(dr)2^{-k}$, where we used that $K$ is determined by
  $RZ$.  Then $\tilde\nu_{z}(d'r')\leq p2^{-k} q(r|dz)\leq
  2^{-\sige }q(r|dz)$.  Furthermore,
  \begin{align}
    \sum_{z}w(\tilde\nu_{z})\mu(z)
    &= \sum_{z}\sum_{d'r'}\tilde\nu_{z}(d'r')\mu(z)\notag\\
    &= \sum_{z}\sum_{d'r'}\tilde\mu_{z}(dr)\mu(s_{k}|drz)\mu(z)\notag\\
    &=\sum_{z}\sum_{dr}\tilde\mu_{z}(dr)\mu(z)\sum_{k,s_{k}}\mu(s_{k}|drz)\notag\\
    &= \sum_{z}\sum_{dr}\tilde\mu_{z}(dr)\mu(z)\notag\\
    &=\sum_z \omega(\tilde\mu_z)\mu(z)\notag\\
    &\geq 1-\epse .
  \end{align}
  It follows that the function $d'r'z\mapsto 2^{-\sige }$ is a
  $\epse$-soft $\UPE(D'{:}R'|Z;E, \cH)$ with
  a constant upper bound, witnessed by
  $\tilde\nu_{z}$ and $q(r'|d'z)=q(r|dz)$.  From
  Lem.~\ref{lm:esupe_const_esmaxprob}, $D'$ has $\epsilon$-smooth
  $Z$-conditional max-prob $2^{-\sige}$ given
  $\{E=e\}$ for each $e$.  The theorem now follows from the
  composition lemma Lem.~\ref{lm:composition} 
  after taking note of Lem.~\ref{lm:pmax_noe}.
\end{proof}

The third protocol is the same as the second except that if the
estimated probability is too large in step 2, the protocol
fails. In this case, banked random bits are not used.  This gives a
protocol $\cQ(\sigma,\epsilon;\cU)$ whose steps are:
\begin{description}[\compact]
\item[$\cQ(\sigma,\epsilon;\cU)$:]
\item[]
  \begin{description}[\compact]
  \item[0.] Choose $n$, $\sige $, $\epse >0$, $\epsx >0$ and $l$ so
    that $\left(n,l,-\log_{2}(2^{-\sige }+2^{-n}),\sigma,\epsx \right)$
    satisfies the extractor constraints, $\epse +\epsx =\epsilon$ and
    $2^{-n}\leq 2^{-\sige }$.

  \item[1.] Perform trials to obtain values $r$ and $z$ of $R$ and $Z$ expressed
    as bit strings,  with $|D|=n$. 

  \item[2.] Determine $F(rz) =\cU(rz,1-\epse )$. 

  \item[3.] If $F(rz)> 2^{-\sige }$, return
    with $\cQ_{X}$ and $\cQ_{S}$ the empty strings and $\cQ_{P}=0$
    to indicate failure.

  \item[4.] Obtain value $(s)_{\le l}$  from $S_{\le l}$.

  \item[5.] Output $\cQ_{X}=\cE(d,(s)_{\le l}, -\log_{2}(2^{-\sige
      }+2^{-n}),\sigma,\epsx )$, $\cQ_{S}=(s)_{\le l}$ and
    $\cQ_{P}=1$.
  \end{description}
\end{description}

\begin{theorem}\label{thm:estext_fail}
  The protocol $\cQ(\sigma,\epsilon;\cU)$ defined above when used with
  a linear strong extractor is a $(\sigma,\epsilon)$-sound
  randomness generation protocol with respect to $\cH$.
\end{theorem}

\begin{proof}
  Because the distribution of $Z$ conditional on $\{\cQ_{P}=1\}$ can
  differ significantly from its initial distribution in a way that is
  not known ahead of time, $Z$'s randomness cannot be
  reused. Consequently, in this proof, $Z$ and $E$ always occur
  together, so we let $Z$ stand for the joint RV $ZE$ throughout.
  
  The proof first replaces the distribution $\mu[DZS\cQ_{P}]$ by
  $\nu[DZS\cQ_{P}]$ so that $\nu[D|Z]\leq 2^{-\sige}+2^{-n}$ and
  conditionally on $\{\cQ_{P}=1\}$, the
  distribution's change is small in a sense to be defined. Since the
  conclusion only depends on $\{\cQ_{P}=1\}$, changes when
  $\{\cQ_{P}=0\}$ can be arbitrary.  We therefore define
  $\TV_{\mathrm{pass}}$ so that for global distributions $\nu$ and
  $\nu'$ satisfying $\nu[\cQ_{P}]=\nu'[\cQ_{P}]$ and for all RVs $U$,
  \begin{equation}
    \TV_{\mathrm{pass}}(\nu[U\cQ_{P}],\nu'[U\cQ_{P}]) =
    \TV(\nu[U|\cQ_{P}=1], \nu'[U|\cQ_{P}=1])\nu(\cQ_{P}=1).  
  \end{equation}
  Note that $\TV_{\mathrm{pass}}$ is the same as the total variation
  distance between $\nu[\cM(U\cQ_{P})\cQ_{P}]$ and
  $\nu'[\cM(U\cQ_{P})\cQ_{P})$ for any process $\cM$ satisfying
  $\cM(U1)=U$ and $\cM(U0)=v$ for a fixed value $v$. Such processes
  forget $U$ when $\cQ_{P}=0$. Note that $(\sigma,\epsilon)$-soundness
  is defined in terms of $\TV_{\mathrm{pass}}$.

  Write $\kappa=\Prob(\cQ_{P}=1)$ and $p=2^{-\sige }$. 
  We omit the
  subscript ${\leq}l$ from $S_{\leq l}$.  To construct $\nu$ we refine
  the proof of Lem.~\ref{lm:esupe_fail_esmaxprob}.  Write
  $F(rz)=\cU(rz,1-\epse )$, and let $\tilde\mu_{z}$ and $q$ witness
  that $F$ is an $\epse$-soft $\UPE(D{:}R|Z;E, \cH)$.  The event
  $\{\cQ_{P}=1\}$ is the same as $\{F\leq p\}$.  The distribution
  $\nu$ to be constructed satisfies $\nu[ZS\cQ_{P}]=\mu[ZS\cQ_{P}]$.
  We maintain that $S$ is uniform and independent of the other RVs for
  both $\nu$ and $\mu$. Thus it suffices to construct $\nu[DZ\cQ_{P}]$
  and define $\nu = \nu[DZ\cQ_{P}]\otimes\Unif_{S}$. We start by
  defining $\nu[DZ|\cQ_{P}=1]$ here to be $\nu[DZ]$ as constructed in
  the proof of Lem.~\ref{lm:esupe_fail_esmaxprob}.  Consistent with
  the notation in that proof, we also use the notation
  $\nu_{z}=\nu[D|z,\cQ_{P}=1]$.  Note that
  $\nu[Z|\cQ_{P}=1](z)=\mu(z|\cQ_{P}=1)$.  We then set
  \begin{align}
    \nu(dz,\cQ_{P}=b)&= \nu[DZ|\cQ_{P}=1](dz)\kappa\knuth{b=1}
    + \Unif_{D}(d)\mu(z,\cQ_{P}=0)\knuth{b=0}\notag\\
    &= \nu_{z}(d)\mu(z,\cQ_{P}=1)\knuth{b=1}
    + \Unif_{D}(d)\mu(z,\cQ_{P}=0)\knuth{b=0}.
  \end{align}
  This ensures that $\nu[ZS\cQ_{P}]=\mu[ZS\cQ_{P}]$ after adding in the uniform and
  independent RV $S$. The distribution $\nu$ satisfies the following additional properties by
  construction:
  \begin{align}
    \nu(d|z) &= \nu(d,\cQ_{P}=1|z)+\nu(d,\cQ_{P}=0|z)\notag\\
    &= \nu_{z}(d)\mu(\cQ_{P}=1|z) + \mu(\cQ_{P}=0|z)/|\Rng(D)|\notag\\
    &\leq p+2^{-n}= 2^{-\sigma_h}+2^{-n}, \notag\\
  \end{align}
  and
  \begin{align}
    \TV_{\mathrm{pass}}(\nu[DZS\cQ_{P}],\mu[DZS\cQ_{P}])
    &=\TV(\nu[DZS|\cQ_{P}=1],\mu[DZS|\cQ_{P}=1])\kappa\notag\\
    &\leq \epse,\label{eq:tvpass_mu_nu}
  \end{align}
  in view of Eq.~\ref{eq:lm:esupe_fail_esmaxprob_1} and the preceding paragraph.
  
  We can now apply the extractor to get
  \begin{equation}
    \TV(\nu[\cE(D,S)SZ], \Unif_{\cE S}\otimes\nu[Z])\leq \epsx .
    \label{eq:thm:estext_fail2a}
  \end{equation}
  Here we transferred the extractor guarantee conditionally on $Z=z$
  for each $z$, noting that the marginal distribution on $Z$ is the
  same for both arguments.  From here on we abbreviate $\cE=\cE(D,S)$.

  The main task for completing the proof of the theorem is to bound
  $\TV(\nu[\cE SZ|\cQ_{p}=1],\Unif_{\cE S}\otimes\nu[Z|\cQ_{P}=1])$,
  which requires transferring the distance of
  Eq.~\ref{eq:thm:estext_fail2a} to the conditional distributions.
  For this, we need to show that for all $z$
  \begin{equation}
    \mu(\cQ_{P}=1|z)\TV(\nu[\cE S|z, \cQ_{P}=1],\Unif_{\cE S}) 
    \le    \TV(\nu[\cE S|z],\Unif_{\cE S}).\label{eq:thm:estext_fail2b}
  \end{equation}
  Before proving this inequality we use it to prove the theorem as follows:
  \begin{align}
    \TV_{\mathrm{pass}}(\nu[\cE SZ\cQ_{P}],\Unif_{\cE S}\otimes \mu[Z\cQ_{P}]) \hspace*{-1in}&\notag\\
    &=
    \mu(\cQ_{P}=1) \TV(\nu[\cE SZ | \cQ_{P}=1], \Unif_{\cE
      S}\otimes \mu[Z|\cQ_{P}=1]) \notag\\
    &=\mu(\cQ_{P}=1) \sum_{z} \mu(z|\cQ_{P}=1)
    \TV(\nu[\cE S|z, \cQ_{P}=1],\Unif_{\cE S}) \notag \\
    &= \sum_{z} \mu(z,\cQ_{P}=1)\TV(\nu[\cE S|z,\cQ_{P}=1],\Unif_{\cE S}) \notag\\
    &= \sum_{z} \mu(z) \mu(\cQ_{P}=1|z)\TV(\nu[\cE S|z, \cQ_{P}=1],\Unif_{\cE S}) \notag\\
    &\leq \sum_{z}\mu(z) \TV(\nu[\cE S|z],\Unif_{\cE S}) \notag\\
    &= \TV(\nu(\cE SZ),\Unif_{\cE S}\otimes\nu(Z))\notag\\
    &\leq \epsx,
  \end{align}
  where we applied Eq.~\ref{eq:tv_samemarg} (twice, first with
  $\mu[Z|\cQ_{P}=1] = \nu[Z|\cQ_{P}=1]$, then with $\mu[Z]=\nu[Z]$),
  Eq.~\ref{eq:thm:estext_fail2b} (to get the third-to-last line), and
  Eq.~\ref{eq:thm:estext_fail2a}.  The theorem now follows via the
  triangle inequality for $\TV_{\textrm{pass}}$, adding the
  $\TV_{\textrm{pass}}$-distance from $\nu[\cE SZ\cQ_{P}]$ to
  $\mu[\cE SZ\cQ_{P}]$.  According to Eq.~\ref{eq:tvpass_mu_nu} and
  the data-processing inequality, this distance is bounded by
  $\epse$.

  It remains to establish Eq.~\ref{eq:thm:estext_fail2b}.  We start
  by taking advantage of the linearity of the extractor.  We have
  \begin{align}
    \nu(\cE=\cale,s|z)  &=
    \nu(\cE=\cale,s|z,\cQ_{P}=1)\nu(\cQ_{P}=1|z)\notag\\
    &\hphantom{=\;\;}
    + \nu(\cE=\cale,s|z,\cQ_{P}=0)\nu(\cQ_{P}=0|z)\notag\\
    &=
    \nu(\cE=\cale,s|z,\cQ_{P}=1)\mu(\cQ_{P}=1|z)\notag\\
    &\hphantom{=\;\;}
    + \nu(\cE=\cale,s|z,\cQ_{P}=0)\mu(\cQ_{P}=0|z).
    \label{eq:thm:estext_fail3}
  \end{align}
  Define $R_{s}=\Rng(\cE(D,s))$ and $r_{s}=|R_{s}|\leq
  |\Rng(\cE(D,S)|=2^{\sigma}$.  Since $\cE$ is a linear extractor,
  $\Rng(D)$ can be treated as a vector space over a finite field. For
  all $\cale\in R_{s}$, $N_{s\cale}=\{d:\cale=\cE(d,s)\}$ is a
  translation of the null space of $d\mapsto \cE(d,s)$. Hence
  $n_{s}=|N_{s\cale}|$ is independent of $\cale$, and
  $r_{s}n_{s}=|\Rng(D)|$. Since $\nu[D|z,\cQ_{P}=0]$ is uniform by
  design, for $\cale\in R_{s}$
  \begin{equation}
    \nu(\cE=\cale,s|z,\cQ_{P}=0) = \frac{|N_{s\cale}|}{|\Rng(D)|} =
    \frac{1}{r_{s}}\geq 2^{-\sigma}.
  \end{equation}
  For $\cale\not\in R_{s}$,
  $\nu(\cE=\cale,s|z,\cQ_{P}=b)=0$.  In the last two equations,
  we can move $s$ into the conditioner in each expression and use
  independence and uniformity of $S$ with respect to $Z$ and
  $\cQ_{P}$ to get
  \begin{align}
    \nu(\cE=\cale|zs)
    &=      \nu(\cE=\cale|zs,\cQ_{P}=1)\knuth{\cale\in R_{s}}\mu(\cQ_{P}=1|z)\notag\\
    &\hphantom{=\;\;}
    + \frac{1}{r_{s}}\knuth{\cale\in R_{s}}\mu(\cQ_{p}=0|z).
  \end{align}
  Abbreviate $\nu_{zs}=\nu[\cE|zs]$, and
  $\nu_{zsb}=\nu[\cE|zs,\cQ_{P}=b]$, which are conditional
  distributions of $\cE$. Then $\nu_{zs0}(\cale)=\knuth{\cale\in
    R_{s}}/r_{s}$.  Write
  $\kappa_{z}=\mu(\cQ_{P}=1|z)=\mu(\cQ_{P}=1|zs)$ for the conditional
  passing probabilities. With these definitions, we can write
  $\nu_{zs} = \kappa_{z}\nu_{zs1}+(1-\kappa_{z})\nu_{zs0}$ and
  calculate
  \begin{align}
    \TV(\nu_{zs},\Unif_{\cE}) &= \sum_{\cale} (\nu_{zs}(\cale)-2^{-\sigma})\knuth{\nu_{zs}(\cale)>2^{-\sigma}} 
    \notag\\
    &= 
    \sum_{\cale\in R_s}(\kappa_{z} \nu_{zs1}(\cale)+ (1-\kappa_{z})/r_s - 2^{-\sigma})\knuth{\kappa_{z} \nu_{zs1}(\cale)+ (1-\kappa_{z})/r_s >2^{-\sigma}}\notag\\
    &= \sum_{\cale\in R_s}(\kappa_{z} (\nu_{zs1}(\cale)-2^{-\sigma})+ (1-\kappa_{z})(1/r_s - 2^{-\sigma}))\notag\\
    &\hphantom{=\;\;}\times\knuth{\kappa_{z} (\nu_{zs1}(\cale)-2^{-\sigma})+ (1-\kappa_{z})(1/r_s-2^{-\sigma})>0} \notag\\
    &\ge \sum_{\cale\in R_s}(\kappa_{z} (\nu_{zs1}(x)-2^{-\sigma}))\knuth{\kappa_{z} (\nu_{zs1}(x)-2^{-\sigma})>0} \notag\\
    &= \kappa_{z} \sum_{\cale\in R_s}(\nu_{zs1}(\cale)-2^{-\sigma})\knuth{(\nu_{zs1}(\cale)-2^{-\sigma})>0} \notag\\
    &= \kappa_{z} \TV(\nu_{zs1},\Unif_{\cE}).
  \end{align}
  Here, we used the fact that $\sum_\cale f(\cale)\knuth{f(\cale)>0}$
  is monotone increasing with $f$, noting that $1/r_{s}\geq
  2^{-\sigma}$. 

  Since $\nu(s|z)=\nu(s|z,\cQ_{P}=1)=\Unif_{S}(s)$, we can apply
  Eq.~\ref{eq:tv_samemarg} to obtain
  \begin{align}
    \TV(\nu[\cE S|z],\Unif_{\cE S}) 
    &= \sum_s \TV(\nu[\cE|zs],\Unif_{\cE}) \nu(s|z) \notag\\
    &= \sum_s \TV(\nu_{zs},\Unif_{\cE})/|\Rng(S)|,
  \end{align}
  and
  \begin{align}
    \TV(\nu[\cE S|z, \cQ_{P}=1],\Unif_{\cE S})\hspace*{-1.5in}\notag\\ &= 
    \sum_s \TV(\nu[\cE|zs, \cQ_{p}=1],\Unif_{\cE}) \nu(s|z, \cQ_{P}=1) 
    \notag\\
    &= \sum_s \TV(\nu_{sz1},\Unif_{\cE})/|\Rng(S)|.
  \end{align}
  Applying the previous two displayed equations and the inequality
  before that, we conclude that for all $z$
  \begin{align}
    \mu(\cQ_{P}=1|z)\TV(\nu[\cE S|z, \cQ_{P}=1],\Unif_{\cE S}) \hspace*{-1in}&\notag\\
    &= \kappa_{z}\TV(\nu[\cE S|z, \cQ_{P}=1],\Unif_{\cE S}) \notag\\
    &= \kappa_{z} \sum_s \TV(\nu_{sz1},\Unif_{\cE})/|\Rng(S)| \notag\\
    &\le  \sum_s \TV(\nu_{sz},\Unif_{\cE})/|\Rng(S)| \notag\\
    &= \TV(\nu[\cE S|z],\Unif_{\cE S}).
  \end{align}
\end{proof}

\section{Test Supermartingales for Uniform Probability Estimation}
\label{sec:prots}

\subsection{Standard Models for Sequences of Trials}
\label{subsect:standardmodels}

Each of Thms.~\ref{thm:prot_chainlemmas},~\ref{thm:estext_nofail},
and~\ref{thm:estext_fail} reduces the problem of randomness generation
protocols to that of probability estimation.  We now consider the
situation where $\Sfnt{CZ}$ is a stochastic sequence of $n$ trials,
with the distributions of the $i$'th trial RVs $C_{i}Z_{i}$ in a model
$\cC_{i}$ conditional on the past and on $E$.  For the remainder of
the paper, the conditioning on $E$ applies universally, and relevant
statements are uniform in the values of $E$. We therefore no longer
mention $E$ explicitly. For the most general treatment, we define
$R_{i}=C_{i}R_{i}'$ where $R_{i}'$ is additional information
obtainable in a trial and $R_{0}=R_{0}'$ is information available
initially. Test supermartingales and test factors are with respect to
$\Sfnt{R}_{\leq i}\Sfnt{Z}_{\leq i}$.  According to our convention,
$D_{i}$ is a function of $C_{i}$.  Other situations can be cast into
this form by changing the definition of $C_{i}$, for example by
extending $C_{i}$ with other information, and modifying models
accordingly.  

Formally, we are considering models $\cH(\cC)$ of distributions of
$\Sfnt{RZ}$ defined by a family of conditionally defined 
models $\cC_{i+1|\Sfnt{r}_{\le i}\Sfnt{z}_{\le i}}$ of
$C_{i+1}Z_{i+1}$. $\cH(\cC)$ consists of the distributions $\mu$ with
the following two properties. First, for all $i$ and $\Sfnt{r}_{\le
  i}\Sfnt{z}_{\le i}$,
\begin{equation}
  \mu[C_{i+1}Z_{i+1}|\Sfnt{r}_{\le i}\Sfnt{z}_{\le i}]\in
  \cC_{i+1|\Sfnt{r}_{\le i}\Sfnt{z}_{\le i}}.
\end{equation}
Second, $\mu$ satisfies that $Z_{i+1}$ is independent of
$\Sfnt{R}_{\le i}$ conditionally on $\Sfnt{Z}_{\le i}$ (and $E$).
Models $\cH(\cC)$ satisfying these conditions are called \emph{standard}.
The second condition is needed in order to be able to estimate
$\Sfnt{Z}$-conditional probabilities of $\Sfnt{D}$ and corresponds
to the Markov-chain condition in the entropy accumulation
framework~\cite{arnon-friedman:2018}.  

In many cases, the sets
$\cC_{i+1|\Sfnt{r}_{\le i}\Sfnt{z}_{\le i}}$ do not depend on
$\Sfnt{r}$, but we take advantage of dependence on $\Sfnt{z}_{\le
  i}$.  In our applications, the sets capture the settings and
non-signaling constraints.  For simplicity, we write
$\cC_{i+1}=\cC_{i+1|\Sfnt{r}_{\le i}\Sfnt{z}_{\le i}}$, leaving the
conditional parameters implicit.

Normally, models for trials are convex and closed.  If
$\Sfnt{Z}$ is non-trivial, the second condition on standard models
prevents their being convex closed for $n>1$, but we note that our
results generally extend to the convex closure of the model used.

\subsection{UPE Constructions}
\label{subsec:upe_constructions}

Let $F_{i+1}$ be non-negative functions of $C_{i+1}Z_{i+1}$, parameterized by
$\Sfnt{R}_{\leq i}$ and $\Sfnt{Z}_{\leq i}$.  We call such functions
``past-parameterized''.  Let $T_{0}=1$ and $T_{i}=\prod_{1\leq j\leq
  i}F_{j}$ for $i\geq 1$.  
For $\beta>0$, define $\cU_{\beta}$ and $\cU_{\beta}^{*}$ by
\begin{align}
  \cU_{\beta}(\Sfnt{RZ},1-\epsilon) &=
  (T_{n}\epsilon)^{-1/\beta},\label{eq:upe_final}\\
  \cU^{*}_{\beta}(\Sfnt{RZ},1-\epsilon) &= \left(\max_{i\leq n}(T_{i}\epsilon)\right)^{-1/\beta}.
  \label{eq:upe_max}
\end{align}
We choose the $F_{i}$ so that they are \emph{probability estimation
  factors} according to the following definition.
\begin{definition}
  \label{def:pef}
  Let $\beta>0$, and let $\cC$ be any model, not necessarily convex.
  A \emph{probability estimation factor (PEF) with power $\beta$ for
    $\cC$ and $D|Z$} is a non-negative RV $F=F(CZ)$ such that for all
  $\nu\in\cC$, $\Exp_{\nu}(F\nu(D|Z)^{\beta})\leq 1$.
\end{definition}
As usual, if the parameters are clear from context, we may omit
them. In particular, unless differently named RVs are involved, PEFs
are always for $D|Z$.  

In the case where $\Sfnt{R}=\Sfnt{C}=\Sfnt{D}$, we show that
$\cU_{\beta}$ and $\cU^{*}_{\beta}$ are $\UPE$s. In the general case,
$\cU_{\beta}$ is a soft $\UPE$. We start with the special case.

\begin{theorem}\label{thm:uest_constr}
  Fix $\beta>0$ and assume $\Sfnt{R}=\Sfnt{C}=\Sfnt{D}$.  Let
  $\Sfnt{F}$ be a sequence of past-parameterized PEFs, with $F_i$ a PEF
  with power $\beta$ for $\cC_i$. Then $\cU_{\beta}(\Sfnt{RZ},1-\epsilon)$ and
  $\cU_{\beta}^{*}(\Sfnt{RZ},1-\epsilon)$ as defined in Eqs.~\eqref{eq:upe_final}
  and~\eqref{eq:upe_max} are
  $\epsilon$-$\UPE(\Sfnt{C}{:}\Sfnt{C}|\Sfnt{Z};\cH(\cC))$.
\end{theorem}

Note that $\beta$ cannot be adapted during the trials.  On the other
hand, before the $i$'th trial, we can design the PEFs $F_{i}$ for
the particular constraints relevant to the $i$'th trial.

\begin{proof}
  We first observe that
  \begin{equation}
    \prod_{j=0}^{i-1}\Probv(C_{j+1}|Z_{j+1}\Sfnt{Z}_{\le j}\Sfnt{C}_{\le j}) =
    \Probv(\Sfnt{C}_{\leq i}|\Sfnt{Z}_{\leq i}).
  \end{equation} 
  This follows by induction with the identity
  \begin{align}
    \Probv(\Sfnt{C}_{\le j+1}|\Sfnt{Z}_{\le j+1}) &=
    \Probv(C_{j+1}|Z_{j+1}\Sfnt{Z}_{\le j}\Sfnt{C}_{\le j})
    \Probv(\Sfnt{C}_{\le j}|Z_{j+1}\Sfnt{Z}_{\le j})\notag\\
    &= \Probv(C_{j+1}|Z_{j+1}\Sfnt{Z}_{\le j}\Sfnt{C}_{\le j})
    \Probv(\Sfnt{C}_{\le j}|\Sfnt{Z}_{\le j})
  \end{align}
  by conditional independence of $Z_{j+1}$.

  We claim that $F_{i+1}\Probv(C_{i+1}|Z_{i+1}\Sfnt{Z}_{\leq
    i}\Sfnt{C}_{\leq i})^{\beta}$ is a test factor determined by
  $\Sfnt{C}_{\leq i+1}\Sfnt{Z}_{\leq i+1}$.  To prove this claim, for all
  $\Sfnt{c}_{\leq i}\Sfnt{z}_{\leq i}$,
  $\nu=\mu[C_{i+1}Z_{i+1}|\Sfnt{c}_{\leq i}\Sfnt{z}_{\leq
    i}]\in\cC_{i+1}$.  With
  $F_{i+1}=F_{i+1}(C_{i+1}Z_{i+1};\Sfnt{c}_{\leq i}\Sfnt{z}_{\leq
    i})$, we obtain the bound
  \begin{align}
    \Exp\left(F_{i+1}\Probv(C_{i+1}|Z_{i+1}\Sfnt{z}_{\leq
        i}\Sfnt{c}_{\leq i})^{\beta}|\Sfnt{c}_{\leq i}\Sfnt{z}_{\leq i}\right)
    &=\Exp_{\nu}\left(F_{i+1}\nu(C_{i+1}|Z_{i+1})^{\beta}\right)\notag\\
    &\leq 1,
  \end{align}
  where we invoked the assumption that $F_{i+1}$ is a PEF with
  power $\beta$ for $\cC_{i+1}$.  By arbitrariness of
  $\Sfnt{c}_{\leq i}\Sfnt{z}_{\leq i}$, and because the factors are
  determined by $\Sfnt{C}_{\leq i+1}\Sfnt{Z}_{\leq i+1}$, the claim
  follows.  The product of these test factors is
  \begin{align}
    \prod_{j=0}^{i-1}F_{j+1}\Probv(C_{j+1}|Z_{j+1}\Sfnt{Z}_{\leq
      j}\Sfnt{C}_{\leq j})^{\beta} &= T_{i}\prod_{j=0}^{i-1}\Probv(C_{j+1}|Z_{j+1}\Sfnt{Z}_{\leq
      j}\Sfnt{C}_{\leq j})^{\beta}\notag\\
    &= T_{i}\Probv(\Sfnt{C}_{\leq i}|\Sfnt{Z}_{\leq i})^{\beta},
  \end{align}
  with $T_{i}=\prod_{j=1}^{i}F_{j}$.  Thus, the sequence
  $\left(T_{i}\Probv(\Sfnt{C}_{\leq i}|\Sfnt{Z}_{\leq
      i})^{\beta}\right)_{i}$ is a test supermartingale, where the
  inverse of the second factor is monotone non-decreasing.  We
  remark that as a consequence, $T_{n}$ is a PEF with power $\beta$
  for $\cH(\cC)$, that is, for $\Sfnt{R}=\Sfnt{C}=\Sfnt{D}$, chaining
  PEFs yields PEFs for standard models.
  
  From Eqs.~\ref{eq:itestmart_markov}
  and~\ref{eq:itestmart_max} with $U_{i}=\Probv(\Sfnt{C}_{\leq
    i}|\Sfnt{Z}_{\leq i})^{-\beta}$ and manipulating the inequalities
  inside $\Prob(.)$, we get
  \begin{eqnarray}
    \Prob(\Probv(\Sfnt{C}_{\leq n}|\Sfnt{Z}_{\leq n}) \geq (T_{n}\epsilon)^{-1/\beta}) 
    & \leq \epsilon,\label{eq:thm:uest_constr:pefupe}\\
    \Prob(\Probv(\Sfnt{C}_{\leq n}|\Sfnt{Z}_{\leq n}) \geq (\max_{i}T_{i}\epsilon)^{-1/\beta}) 
    & \leq \epsilon.
  \end{eqnarray}
  To conclude that $\cU_\beta$ and $\cU_\beta^*$ are UPEs, it now
  suffices to refer to their defining identities
  Eqs.~\eqref{eq:upe_final} and~\eqref{eq:upe_max}.
\end{proof}

That $F_{i+1}$ can be parameterized in terms of the past as
$F_{i+1}=F_{i+1}(C_{i+1}Z_{i+1}; \Sfnt{C}_{\leq i}\Sfnt{Z_{\leq
    i}})$ allows for adapting the PEFs based on
$\Sfnt{C}\Sfnt{Z}$, but no other information in $\Sfnt{R}$
can be used. Since $T_{i}$ is a function of $\Sfnt{C}_{\leq
  i}\Sfnt{Z_{\leq i}}$, this enables stopping when a target value of
$T_{i}$ is achieved.

\Pc{Here is a trick by which one can
  adapt based on the above theorem, without invoking the next theorem.

  Whenever there is a need, one can insert special trials with
  $Z_{i+1}=\textrm{"special"}$ and $C_{i+1}$ containing information not
  available in $\Sfnt{C}_{\leq i}\Sfnt{Z}_{\leq i}$. This requires
  modifying the $\cC_{i+1}$ so that the original constraints are
  preserved conditional on ${Z_{i+1}\ne\textrm{"special"}}$ and no other
  constraints are added.  The value of $F_{i+1}$ on these new events is
  set to $1$.  With this trick, it is possible to use general adaptation
  strategies, at the cost of increasing the length of the extractor
  input according to the size of the information included in the inserted
  special trials.  However, the next theorem shows that it is not
  necessary to use this trick for randomness generation. } 

In order to use additional information available in $\Sfnt{R}$,
we now treat the case where $C,R\not= D$. While not necessary at this point,
we do this with respect to a softening of the PEF properties.

\begin{definition} \label{def_soft_pef} Let $\beta>0$, and let $\cC$
  be any model, not necessarily convex.  A \emph{soft probability
    estimation factor (soft PEF) with power $\beta$ for $\cC$ and
    $D|Z$} is a non-negative RV $F=F(CZ)$ such that for all
  $\nu\in\cC$, there exists a function $q(C|DZ)\geq 0$ of $CDZ$ with
  $\sum_c q(c|dz)\leq 1$ for all $dz$ and
  \begin{equation}
    \Exp_{\nu}\left(\frac{F(CZ)}{q(C|DZ)^\beta}\nu(C|Z)^\beta\right)\leq 1.
    \label{eq:spefdef}
  \end{equation}
\end{definition}
If $\sum_{c}q(c|dz)<1$ for some $dz$, we can increase $q$ to ensure
$\sum_{c}q(c|dz)=1$ without increasing the left-hand side of
Eq.~\ref{eq:spefdef}. Hence, without loss of generality, we can let
$\sum_c q(c|dz)=1$ in the above definition. We always set $q(c|dz)=0$
for probability-zero values $cdz$.  This does not cause problems with
Eq.~\ref{eq:spefdef} given the convention that when expanding an
expectation over a finite set of values, probability-zero values are
omitted from the sum.

The direct calculation in the proof of the next lemma shows that
PEFs are soft PEFs.

\begin{lemma}
  \label{lem:pef_softpef}
  If $F$ is a PEF with power $\beta$ for $\cC$, then it is
  a soft PEF with power $\beta$ for $\cC$.
\end{lemma}

\begin{proof}
  Let $F$ be a PEF with power $\beta$ for $\cC$ and
  $\nu\in\cC$. Define $q(C|DZ)=\nu(C|DZ)$. Then $q$ satisfies the
  condition in the definition of soft PEFs.  Since $D$ is a function
  of $C$, we have $\nu(dcz)=\nu(cz)\knuth{d=D(c)}$. From this identity
  we can deduce
  \begin{align}
    \Exp_{\nu}\left(\frac{F(CZ)}{q(C|DZ)^\beta}\nu(C|Z)^\beta\right)
    &= \Exp_{\nu} \left(F(CZ)
      \left(\frac{\nu(C|Z)}{q(C|DZ)}\right)^{\beta}\right)
    \notag\\
    &= \Exp_{\nu} \left(F(CZ)
      \left(\frac{\nu(CZ)/\nu(Z)}{\nu(CDZ)/\nu(DZ)}\right)^{\beta}\right)
    \notag\\
    &= \Exp_{\nu} \left(F(CZ)
      \left(\frac{\nu(DZ)}{\nu(Z)}\right)^{\beta}\right)
    \notag\\
    &= \Exp_{\nu}\left(F(CZ)\nu(D|Z)^{\beta}\right)\notag\\
    &\leq 1.
  \end{align}
  Since $\nu\in\cC$ is arbitrary, this verifies that every PEF $F$ is
  a soft PEF.
\end{proof}

\begin{lemma}
  \label{lem:spef_supe}
  Let $F(CZ)$ be a soft PEF with power $\beta$ for $\cH$.  Then
  $\cV=(F\epsilon)^{-1/\beta}$ is an $\epsilon$-soft $\UPE(D{:}C|Z;\cH)$.
\end{lemma}

\begin{proof}
  Fix $\nu\in\cH$ and let $q(C|DZ)$ be the function witnessing that
  $F$ is a soft PEF for this $\nu$.
  Define 
  \begin{equation}
    \tilde\nu_{z}(dc)=\nu(dc|z)\knuth{\cV(cz)\geq\nu(c|z)/q(c|dz)}.
  \end{equation}
  The Markov inequality
  and the definition of soft PEFs implies that
  \begin{align}
    \Prob_{\nu}(\cV(CZ)<\nu(C|Z)/q(C|DZ))
    &=\Prob_{\nu}\left((\nu(C|Z)/q(C|DZ))^{\beta}>\cV(CZ)^{\beta}\right)\notag\\
    &=\Prob_{\nu}(F(CZ)(\nu(C|Z)/q(C|DZ))^{\beta}>1/\epsilon)\notag\\
    &\leq \epsilon.
  \end{align}
  Hence
  \begin{align}
    \sum_{dcz}\tilde\nu_{z}(dc)\nu(z)
    &=\sum_{dcz}\nu(dcz)\knuth{\cV(cz)\geq\nu(c|z)/q(c|dz)}\notag\\
    &= \Prob_{\nu}(\cV(CZ)\geq\nu(C|Z)/q(C|DZ))\notag\\
    &\geq 1-\epsilon.
  \end{align}
  The definition of $\tilde \nu_{z}(dc)$ and
  $\nu(dc|z)=\nu(c|z)\knuth{d=D(c)}$ 
  ensure that
  \begin{align}
    \tilde\nu_{z}(dc)\leq \cV(cz)q(c|dz),
  \end{align}
  as required for soft $\UPE$s. Since $\nu$ is arbitrary, the above
  verifies that $\cV=(F\epsilon)^{-1/\beta}$ is an $\epsilon$-soft
  $\UPE$.
\end{proof}

\begin{theorem}\label{thm:suest_constr}
  Fix $\beta> 0$.  Let $\Sfnt{F}$ be a sequence of
  past-parameterized soft PEFs, with $F_{i}$ a soft PEF with power
  $\beta$ for $\cC_{i}$. Then $\cU_{\beta}(\Sfnt{RZ},1-\epsilon)$ as defined
  in Eqs.~\eqref{eq:upe_final} is an $\epsilon$-soft
  $\UPE(\Sfnt{D}{:}\Sfnt{R}|\Sfnt{Z};\cH(\cC))$.    
\end{theorem}

\begin{proof}
  Below we first show that chaining soft PEFs yields soft PEFs.  Then,
  by applying Lem.~\ref{lem:spef_supe} we prove the theorem.

  For this result, direct chaining of the probabilities
  fails. Instead, we decompose alternately according to $C_{i}$ and
  $R_{i}$ given $C_{i}$, as follows: 
  \begin{align}
    \Prob(\Sfnt{C}_{\leq i+1}\Sfnt{R}_{\leq i+1}|\Sfnt{Z}_{\leq i+1})
    &=\Prob(\Sfnt{C}_{\leq i}\Sfnt{R}_{\leq i}|Z_{i+1}\Sfnt{Z}_{\leq i})\notag\\
    &\hphantom{=\;\;}\times
    \Prob(C_{i+1}|Z_{i+1}\Sfnt{C}_{\leq i}\Sfnt{R}_{\leq i}\Sfnt{Z}_{\leq i})
    \Prob(R_{i+1}|C_{i+1}Z_{i+1}\Sfnt{C}_{\leq i}\Sfnt{R}_{\leq i}\Sfnt{Z}_{\leq i})\notag\\
    &=\Prob(\Sfnt{C}_{\leq i}\Sfnt{R}_{\leq i}|\Sfnt{Z}_{\leq i})
    \Prob(C_{i+1}|Z_{i+1}\Sfnt{R}_{\leq i}\Sfnt{Z}_{\leq i})
    \Prob(R_{i+1}|C_{i+1}Z_{i+1}\Sfnt{R}_{\leq i}\Sfnt{Z}_{\leq i}),
  \end{align}
  where we simplified in the last step by applying the fact that
  $Z_{i+1}$ is conditionally independent of $\Sfnt{R}_{\leq i}$ given
  $\Sfnt{Z}_{\leq i}$ and taking advantage of the assumption that
  $C_{i}$ is determined by $R_{i}$. The identity can be expanded
  recursively, replacing the first term in the last expression
  obtained each time.  We identify two products after the expansion, namely
  \begin{equation}
    P_{C,i}(\Sfnt{C}_{\leq i}\Sfnt{R}_{\leq i}\Sfnt{Z}_{\leq i}) = 
    \prod_{j=0}^{i-1}\Prob(C_{j+1}|Z_{j+1}\Sfnt{R}_{\leq j}\Sfnt{Z}_{\leq j})
  \end{equation}
  and
  \begin{equation}
    P_{R,i}(\Sfnt{R}_{\leq i}\Sfnt{Z}_{\leq i}) =
    \prod_{j=0}^{i-1}\Prob(R_{j+1}|C_{j+1}Z_{j+1}\Sfnt{R}_{\leq j}\Sfnt{Z}_{\leq j}),
  \end{equation}
  which satisfy
  \begin{equation}
    \Prob(\Sfnt{C}_{\leq i}\Sfnt{R}_{\leq i}|\Sfnt{Z}_{\leq i})
    =\Prob(\Sfnt{R}_{\leq i}|\Sfnt{Z}_{\leq i}) =
    P_{C,i}P_{R,i},
    \label{eq:thm:suest_constr1}
  \end{equation}
  whenever $\Prob(\Sfnt{C}_{\leq i}\Sfnt{R}_{\leq i}|\Sfnt{Z}_{\leq i})$ is not zero.
  Again, we took advantage of our assumption that
  $C_{i}$ is a function of $R_{i}$ to omit unnecessary arguments.

  Let $\mu$ be an arbitrary distribution in $\cH(\cC)$ and note that
  for all $\Sfnt{r}_{\leq i}\Sfnt{z}_{\leq i}$,
  $\nu_{i+1}=\mu[R_{i+1}Z_{i+1}|\Sfnt{r}_{\leq i}\Sfnt{z}_{\leq i}]$
  satisfies that $\nu_{i+1}[C_{i+1}Z_{i+1}]\in\cC_{i+1}$. According to
  the definition of soft PEFs, there exists
  $q_{i+1}(C_{i+1}|D_{i+1}Z_{i+1};\Sfnt{r}_{\leq i}\Sfnt{z}_{\leq i})$
  such that
  \begin{equation}
    \Exp(F_{i+1} (\nu_{i+1}(C_{i+1}|Z_{i+1})/q_{i+1}(C_{i+1}|D_{i+1}Z_{i+1}))^{\beta})
    \leq 1
  \end{equation}
  and $\sum_{c_{i+1}}q_{i+1}(c_{i+1}|d_{i+1}z_{i+1})\leq 1$ for all $d_{i+1}z_{i+1}$, 
   where the past parameters are implicit. We define the additional product
  \begin{equation} 
    Q_{i}(\Sfnt{C}_{\leq i}|\Sfnt{D}_{\leq i}\Sfnt{Z}_{\leq i};\Sfnt{R}_{\leq i-1})
    = \prod_{j=0}^{i-1}q_{j+1}(C_{j+1}|D_{j+1}Z_{j+1};\Sfnt{R}_{\leq j}\Sfnt{Z}_{\leq j}).
  \end{equation}
  Later we find that the function $q(\Sfnt{R}|\Sfnt{D}\Sfnt{Z})$
  needed for the final soft UPE construction according to
    Def.~\ref{def_soft_UPE} is the $i=n$ instance of
  \begin{equation}
    p_{\leq i}(\Sfnt{R}_{\leq i}|\Sfnt{D}_{\leq i}\Sfnt{Z}_{\leq i}) = 
    Q_{i}(\Sfnt{C}_{\leq i}|\Sfnt{D}_{\leq i}\Sfnt{Z}_{\leq i};\Sfnt{R}_{\leq i-1})P_{R,i}(\Sfnt{R}_{\leq i}\Sfnt{Z}_{\leq i}).\notag
  \end{equation}
  Write 
  \begin{equation}  
    p_{i+1}(R_{i+1}|D_{i+1}Z_{i+1};\Sfnt{R}_{\leq
      i}\Sfnt{Z}_{\leq i})=q_{i+1}(C_{i+1}|D_{i+1}Z_{i+1};\Sfnt{R}_{\leq
      i}\Sfnt{Z}_{\leq i})
    \Prob(R_{i+1}|C_{i+1}Z_{i+1}\Sfnt{R}_{\leq i} \Sfnt{Z}_{\leq i}) 
  \label{eq:thm:suest_constr2}
  \end{equation}    
  for the incremental factor defining
  $p_{\leq i+1}$, so that $p_{\leq i+1}=p_{i+1}p_{\leq i}$.  We show by
  induction that $p_{\leq i}$ satisfies the required normalization
  condition. The initial value is $p_{\leq 0}=1$ (empty products
  evaluate to $1$), and for the induction step,
    for all $\Sfnt{d}_{\leq i+1}\Sfnt{z}_{\leq i+1}$
  \begin{align}
    \sum_{\Sfnt{r}_{\leq i+1}}p_{\leq i+1}(\Sfnt{r}_{\leq
      i+1}|\Sfnt{d}_{\leq i+1}\Sfnt{z}_{\leq i+1})
    \hspace*{-1in}&\notag\\
    &= \sum_{\Sfnt{r}_{\leq i}}p_{\leq i}(\Sfnt{r}_{\leq
      i}|\Sfnt{d}_{\leq i}\Sfnt{z}_{\leq i})
    \notag\\
    &\hphantom{=\;\;}\times
    \sum_{r_{i+1}}q_{i+1}(c_{i+1}|d_{i+1}z_{i+1};\Sfnt{r}_{\leq
      i}\Sfnt{z}_{\leq i})
    \Prob(r_{i+1}|c_{i+1}z_{i+1}\Sfnt{r}_{\leq i} \Sfnt{z}_{\leq i}).
    \label{eq:thm:suest_constr3}
  \end{align}
  The inner sum on the right-hand side evaluates to
  \begin{align}
    \sum_{r_{i+1}}q_{i+1}(c_{i+1}|d_{i+1}z_{i+1};\Sfnt{r}_{\leq
      i}\Sfnt{z}_{\leq i})
    \Prob(r_{i+1}|c_{i+1}z_{i+1}\Sfnt{r}_{\leq i} \Sfnt{z}_{\leq i})\hspace*{-3in}&
    \notag\\
    &= \sum_{c_{i+1}}\left(\sum_{r_{i+1}:C_{i+1}(r_{i+1})=c_{i+1}}
      \Prob(r_{i+1}|c_{i+1}z_{i+1}\Sfnt{r}_{\leq i} \Sfnt{z}_{\leq i})\right)
    q_{i+1}(c_{i+1}|d_{i+1}z_{i+1};\Sfnt{r}_{\leq i}\Sfnt{z}_{\leq i})    \notag\\
    &=\sum_{c_{i+1}}q_{i+1}(c_{i+1}|d_{i+1}z_{i+1};\Sfnt{r}_{\leq i}\Sfnt{z}_{\leq i})
    \notag\\
    &\leq 1.
  \end{align}
  Substituting back in Eq.~\ref{eq:thm:suest_constr3} and applying
  the induction hypothesis gives the normalization condition.

  We claim that 
  \begin{align}
    G_{i+1}(R_{i+1}Z_{i+1};\Sfnt{R}_{\leq i}\Sfnt{Z}_{\leq i})
    \hspace*{-1in}&
    \notag\\
    &= \frac{F_{i+1}(C_{i+1}Z_{i+1};\Sfnt{R}_{\leq i}\Sfnt{Z}_{\leq
        i})}{p_{i+1}(R_{i+1}|D_{i+1}Z_{i+1};\Sfnt{R}_{\leq
        i}\Sfnt{Z}_{\leq
        i})^{\beta}}\Probv(R_{i+1}|Z_{i+1}\Sfnt{R}_{\leq
      i}\Sfnt{Z}_{\leq i})^{\beta}
  \end{align} 
  is a test factor determined by $\Sfnt{R}_{\leq i+1}\Sfnt{Z}_{\leq
    i+1}$.  To prove this claim, with
  $F_{i+1}=F_{i+1}(C_{i+1}Z_{i+1};\Sfnt{r}_{\leq i}\Sfnt{z}_{\leq
    i})$ and omitting the past parameters of $p_{i+1}$, 
    for each $\Sfnt{r}_{\leq i}\Sfnt{z}_{\leq i}$ we compute 
  \begin{align}
    \Exp\left(F_{i+1}
      \left(\frac{\Probv(R_{i+1}|Z_{i+1}\Sfnt{r}_{\leq i}\Sfnt{z}_{\leq i})}{p_{i+1}(R_{i+1}|D_{i+1}Z_{i+1})}\right)^{\beta}\middle|\Sfnt{r}_{\leq
        i}\Sfnt{z}_{\leq i}\right)
    \hspace*{-2.5in}\notag\\
    &=
    \Exp\left(F_{i+1}
      \left(\frac{
          \Probv(R_{i+1}|C_{i+1}Z_{i+1}\Sfnt{r}_{\leq i}\Sfnt{z}_{\leq i})
          \Probv(C_{i+1}|Z_{i+1}\Sfnt{r}_{\leq i}\Sfnt{z}_{\leq i})}
        {\Probv(R_{i+1}|C_{i+1}Z_{i+1}\Sfnt{r}_{\leq i}\Sfnt{z}_{\leq i})
          q_{i+1}(C_{i+1}|D_{i+1}Z_{i+1})}\right)^{\beta}\middle|\Sfnt{r}_{\leq
        i}\Sfnt{z}_{\leq i}\right)\notag\\
    &=\Exp_{\nu_{i+1}}\left(F_{i+1}\left(     
        \frac{\Probv_{\nu_{i+1}}(C_{i+1}|Z_{i+1})}
        {q_{i+1}(C_{i+1}|D_{i+1}Z_{i+1})}\right)^{\beta}\right)
    \notag\\
    &\leq 1,
  \end{align}
  by the definition of soft PEFs, where we used
  Eqs.~\ref{eq:thm:suest_constr1} and~\ref{eq:thm:suest_constr2} to
  get the second line. By arbitrariness of $\Sfnt{r}_{\leq
    i}\Sfnt{z}_{\leq i}$, the claim follows.

  Since $\prod_{j=0}^{n-1}p_{j+1} = p_{\leq n}$, if we set the
  function $q$ in the definition Def.~\ref{def_soft_pef} to $p_{\leq
    n}$, we conclude that $T_{n}(\Sfnt{RZ})=\prod_{j=0}^{n-1}F_{j+1}$
  is a soft PEF with power $\beta$ for $\Sfnt{D}|\Sfnt{Z}$ and
  $\cH(\cC)$. Hence, chaining soft PEFs yields soft PEFs for standard
  models.  Since
  $\cU_{\beta}(\Sfnt{RZ},(1-\epsilon))=(T_{n}\epsilon)^{-1/\beta}$, to
  finish the proof of the theorem, we apply Lem.~\ref{lem:spef_supe}
  with $C$ there replaced by $\Sfnt{R}$ here, $D$ there with
  $\Sfnt{D}$ here and $Z$ there with $\Sfnt{Z}$ here.  The proof of
  Lem.~\ref{lem:spef_supe} shows that the softness witness
  $q(\Sfnt{R}|\Sfnt{D}\Sfnt{Z})$ for
  $\cU_{\beta}(\Sfnt{RZ},1-\epsilon)$ is also given by $p_{\leq n}$.
\end{proof}

\subsection{Effectiveness of PEF Constraints}
\label{sec:first_pef_construction}

The constraints on PEFs $F_{i}$ are linear, but it is not obvious that
this set of linear constraints has a practical realization as a linear
or semidefinite program.  Thm.~\ref{thm:prob_constr} below shows that
when $C=D$, it suffices to use the extreme points $\xtrm{\cC}$
of $\cC$, which implies that if $\cC$ is a known polytope, then a
finite number of constraints suffice. This includes the cases where
$\cC$ is characterized by typical settings and non-signaling
constraints for $(k,l,m)$ Bell-test configurations and
generalizes to cases where the settings distributions are not fixed
but constrained to a polytope of settings distributions. See
Sect.~\ref{subsec:freeforz}.  For the case where $C\not=D$, we
show that soft PEFs for $\xtrm{\cC}$ are soft PEFs for $\cC$. In
particular PEFs for $\xtrm{\cC}$ yield soft PEFs for $\cC$. When
$\cC$ is a polytope this yields an effective construction for soft
PEFs.

\begin{theorem}\label{thm:prob_constr}
  Let $\cC$ be a convex set of distributions of $CZ$.  Then for
  $\beta>0$ the family of inequalities
  \begin{equation}
    \begin{array}{rcll}
      F&\geq& 0 & ,\\
      \Exp_{\nu}\left(F(CZ)\nu(C|Z)^{\beta}\right)&\leq& 1 & \textrm{for $\nu\in\cC$}
    \end{array}
    \label{eq:probest_constr}
  \end{equation}
  is implied by the subset with $\nu\in\xtrm{\cC}$.
\end{theorem}

\begin{proof}
  For a given $\nu\in\cC$, we can write the constraint on $F$ as
  \begin{equation}
    \Exp_{\nu}\left(F(CZ)\nu(C|Z)^{\beta}\right) =
    \sum_{cz}F(cz)\nu(cz)^{1+\beta}/\nu(z)^{\beta} \leq 1.
  \end{equation}
  Thus, each $\nu$ determines a linear inequality for $F$.
  For distributions $\sigma$ on $CZ$ and $\rho$ on $Z$ with
  $\Supp(\sigma[Z])\subseteq\Supp(\rho)$ define
  \begin{equation}
    l_{\sigma|\rho}(cz) = \sigma(cz)^{1+\beta}/\rho(z)^{\beta},
  \end{equation}
  so that $l_{\nu|\nu}(cz)$ are the coefficients of $F$ in the
  inequality above.  If $\rho(z)=0$, then $\sigma(cz)\leq\sigma(z)=0$,
  so we set $l_{\nu|\rho}(cz)=0$, which is consistent since the power
  in the numerator is larger than that in the denominator.  Since $F$
  is non-negative, it suffices to show that for each $cz$, the
  coefficients $l_{\nu|\rho}(cz)$ are biconvex functions of $\nu$ and
  $\rho$.  If $\nu=\lambda\nu_{1}+(1-\lambda)\nu_{2}$ with
  $0\leq\lambda\leq 1$ and $\nu_{i}\in\cC$, then
  \begin{equation}
    l_{\nu|\nu}(cz)\leq
    \lambda l_{\nu_{1}|\nu_{1}}(cz)+(1-\lambda) l_{\nu_{2}|\nu_{2}}(cz).
    \label{eq:thm:prob_constr_1}
  \end{equation}
  Given that $\sum_{cz}F(cz)
  l_{\nu'|\nu'}(cz)\leq 1$ for $\nu'=\nu_{1}$ and $\nu'=\nu_{2}$, we find
  \begin{align}
    \sum_{cz}F(cz) l_{\nu|\nu}(cz)&\leq \sum_{cz}F(cz)(\lambda
    l_{\nu_{1}|\nu_{1}}(cz)+
    (1-\lambda) l_{\nu_{2}|\nu_{2}}(cz))\notag\\
    &\leq \lambda\sum_{cz}F(cz) l_{\nu_{1}|\nu_{1}}(cz)
    +(1-\lambda)\sum_{cz}F(cz) l_{\nu_{2}|\nu_{2}}(cz)\notag\\
    &\leq \lambda+(1-\lambda)=1.
  \end{align}
  This observation extends to arbitrary finite convex combinations by reduction
  to the above case. Consequently the constraints associated with
  extremal distributions suffice.

  Let $r(x|y) = x^{1+\beta}y^{-\beta}$, where we define $r(0|0)=0$.
  The desired convexity expressed in Eq.~\ref{eq:thm:prob_constr_1} is
  implied by the fact that $r(x|y)$ is
  biconvex in its arguments for relevant values, which is
  Lem.~\ref{lem:renybiconvex} below.
\end{proof}

\begin{lemma}
  \label{lem:renybiconvex}
  Define $r(x|y) = x^{1+\beta}y^{-\beta}$ with $\beta> 0$.  If $(x,y)$ is expressed
  as a finite convex combination
  $(x,y)=\sum_{i}\lambda_{i}(x_{i},y_{i})$ where $x_{i}\geq 0$,
  $y_{i}\geq 0$ and $y_{i}=0$ implies $x_{i}=0$, then $r(x|y)\leq
  \sum_{i}\lambda_{i}r(x_{i}|y_{i})$.
\end{lemma}

\begin{proof}
  This is a special case of the biconvexity of sandwiched R\'enyi
  powers, see~\cite{mueller-lennert:qc2013a,frank:qc2013a}.  To
  prove biconvexity of $r(x|y)$ directly, it suffices to consider
  convex combinations of two arguments, so let $x=\lambda
  x_{1}+(1-\lambda) x_{2}$ and $y=\lambda y_{1}+(1-\lambda) y_{2}$
  be convex combinations as expressed in the theorem statement.
  Consider
  \begin{equation}
    f(\lambda) = 
    r(\lambda x_{1}+(1-\lambda)x_{2}|\lambda y_{1} + (1-\lambda)y_{2}).
  \end{equation}
  If $y_{1}=y_{2}=0$, then $f(\lambda)=0$ for all $\lambda$, which
  is convex. If $y_{1}=0<y_{2}$, then $f(\lambda)=(1-\lambda)
  x_{2}^{1+\beta}/y_{2}^{\beta}$, which again is convex. By
  symmetry, so is $f(\lambda)$ for the case where $y_{2}=0<y_{1}$.  It
  therefore suffices to consider the case where $y_{1}>0$ and
  $y_{2}>0$. Now $f(\lambda)$ is continuous for $\lambda\in[0,1]$
  and smooth for $\lambda\in(0,1)$, so it suffices to show that
  $f''(\lambda)\geq 0$ as follows:
  \begin{align}
    f'(\lambda) &= (\lambda x_1+(1-\lambda)x_{2})^{\beta}(\lambda y_1+(1-\lambda)y_2)^{-\beta-1}\notag\\
    &\hphantom{=\;\;} \times \Big( (1+\beta) (x_{1}-x_{2})(\lambda y_1+(1-\lambda)y_{2}) + 
    (-\beta) (\lambda x_{1}+(1-\lambda)x_{2})(y_1-y_{2}) \Big)\label{eq:thm:prob_constr_2}\\
    f''(\lambda) &= 
    (\lambda x_{1}+(1-\lambda)x_{2})^{\beta-1}(\lambda y_1+(1-\lambda)y_{2})^{-\beta-2}\notag\\
    &\hphantom{=\;\;} \times \Big( \beta(1+\beta) (x_{1}-x_{2})^{2}(\lambda y_1+(1-\lambda)y_{2})^{2} 
    \notag\\
    &\hphantom{=\;\;} \hphantom{\times\Big(} + 2(-\beta)(1+\beta) (x_{1}-x_{2})(y_1-y_{2})(\lambda x_{1}+(1-\lambda) x_{2})(\lambda y_1+(1-\lambda)y_{2})\notag\\
    &\hphantom{=\;\;} \hphantom{\times\Big(} + (-\beta)(-1-\beta) (y_1-y_{2})^{2}(\lambda x_{1}+(1-\lambda )x_{2})^{2}\Big)\notag\\
    &= (\lambda x_{1}+(1-\lambda)x_{2})^{\beta-1}(\lambda y_1+(1-\lambda)y_{2})^{-\beta-2}\notag\\
    &\hphantom{=\;\;} \times \beta(1+\beta)\Big((x_{1}-x_{2})(\lambda y_1+(1-\lambda) y_{2}) -
    (y_1-y_{2})(\lambda x_{1}+(1-\lambda) x_{2})\Big)^{2},
  \end{align}
  which is a non-negative multiple of a square. Consequently $f(\lambda)$
  is convex in $\lambda$, and biconvexity of $r(x|y)$ follows. 
\end{proof}

\begin{corollary}
  If $F(CZ)$ is a PEF with power $\beta$ for $C|Z$ and $\xtrm{\cC}$,
  then $F(CZ)$ is a PEF with power $\beta$ for $C|Z$ and the
  convex closure of $\cC$.
\end{corollary}

The same observation holds for soft PEFs.

\begin{theorem}
  \label{thm:softpefxtrm}
  If $F(CZ)$ is a soft PEF with power $\beta$ for $D|Z$ and $\xtrm{\cC}$,
  then $F(CZ)$ is a soft PEF with power $\beta$ for $D|Z$ and the
  convex closure of $\cC$.
\end{theorem}

\begin{proof}
  Let $\nu_i\in\xtrm{\cC}$ and $\nu=\sum_i \lambda_i \nu_i$ be a
  finite convex combination of the $\nu_i$.
  Let $q_i(C|DZ)$ witness that $F$ is a soft PEF at
  $\nu_i$. From the joint convexity of $\nu^{1+\beta}/\sigma^\beta$
  according to Lem.~\ref{lem:renybiconvex}, for all $cdz$
  \begin{equation}
    \frac{\nu(cz)^{1+\beta}}{(\sum_i\lambda_i q_i(c|dz)\nu_i(z))^\beta}\leq
    \sum_i \lambda_i\frac{\nu_i(cz)^{1+\beta}}{(q_i(c|dz)\nu_i(z))^\beta}.
  \end{equation}
  To apply Lem~\ref{lem:renybiconvex}, note that the definition of
  soft PEFs implies that $q_{i}(c|dz)>0$ whenever $\nu_{i}(cz)>0$ as
  otherwise the defining inequality cannot be satisfied.  Define
  $q(c|dz)=\sum_i\lambda_iq_i(c|dz)\nu_i(z)/\nu(z)$ if 
  $\nu(z)>0$ and $q(c|dz)=0$ otherwise. We have
  \begin{equation}
    \sum_c q(c|dz)=\sum_i\lambda_i \sum_c q_i(c|dz)\nu_i(z)/\nu(z)
    \leq \sum_i\lambda_i \nu_i(z)/\nu(z)= \nu(z)/\nu(z)= 1.
  \end{equation}
  Now
  \begin{align}
    \frac{F(cz)}{q(c|dz)^\beta}\nu(c|z)^{\beta}\nu(cz)
    &= F(cz)\frac{\nu(cz)^{1+\beta}}{q(c|dz)^\beta\nu(z)^\beta}
    \notag\\
    &= F(cz) \frac{\nu(cz)^{1+\beta}}{(\sum_i\lambda_iq_i(c|dz)\nu_i(z))^\beta}
    \notag\\
    &\leq \sum_i\lambda_iF(cz) \nu_i(cz)^{1+\beta}/(q_i(c|dz)\nu_i(z))^\beta
    \notag\\
    &= \sum_i\lambda_i\frac{F(cz)}{q_i(c|dz)^\beta}
    \nu_i(c|z)^\beta\nu_i(cz).
  \end{align}
  Summing both sides over $dcz$ gives the desired inequality
  \begin{equation}
    \Exp_\nu\left(\frac{F(CZ)}{q(C|DZ)^\beta} \nu(C|Z)^\beta\right)
    \leq \sum_i\lambda_i \Exp_{\nu_i}\left(\frac{F(CZ)}{q_i(C|DZ)^\beta} \nu_i(C|Z)^\beta\right) \leq 1.
  \end{equation}
\end{proof}

\begin{corollary}
  If $F(CZ)$ is a PEF with power $\beta$ for $D|Z$ and $\xtrm{\cC}$,
  then $F(CZ)$ is a soft PEF with power $\beta$ for $D|Z$ and 
  the convex closure of $\cC$.
\end{corollary}

\begin{proof}
  By Lem.~\ref{lem:pef_softpef}, $F(CZ)$ is a soft PEF with power
  $\beta$ for $D|Z$ and $\xtrm{\cC}$, so we can apply
  Thm.~\ref{thm:softpefxtrm}.
\end{proof}

\subsection{Log-Probs and PEF Optimization Strategy}

Given the probability estimation framework, the number of random bits
that can be extracted is determined by the negative logarithm of the
final probability estimate after adjustments for error bounds and
extractor constraints. We focus on the quantities obtained from
Eq.~\ref{eq:upe_final}, which reflect the final values of the PEF
products. 

\begin{definition}
  Let $F_{i}$ be PEFs with power $\beta>0$ for $\cC_{i}$.
  The \emph{log-prob of $\Sfnt{F}$} is
  $\sum_{i}\log(F_{i}(c_{i}z_{i}))/\beta$. Given an error
  bound $\epsilon$, the \emph{net log-prob} is given by
  $\sum_{i}\log(F_{i}(c_{i}z_{i}))/\beta-\log(1/\epsilon)/\beta$.
  \label{def:net_log_prob}
\end{definition}

In the randomness generation protocols of Sect.~\ref{subsec:protocols}
with the UPEs of Sect.~\ref{subsec:upe_constructions}, the net
log-prob is effectively a ``raw'' entropy base $e$ available for
extraction. For the protocol to succeed, the net log-prob must exceed
$\sige\log(2)$, where $\sige$ is the extractor parameter for the input
min-entropy in bits.  In the protocol with banked randomness, the net
log-prob contributes almost fully to the input min-entropy provided to
the extractor.  In Sect.~\ref{sec:apps}, we further adjust the net
log-prob by subtracting the entropy used for settings choices and call
the resulting quantity the \emph{net entropy}. This does not yet take
into account the entropy needed for the extractor seed, nor the fact
that the number of output bits is smaller than the input min-entropy
according to the extractor constraints.

The log-prob and net log-prob are empirical quantities. For the
average quantities with respect to a distribution we consider expected
contributions per trial.

\begin{definition}
  If $\nu\in\cC$ is the probability distribution of $CZ$ at a
  trial and $F$ is a PEF with power $\beta>0$ for $\cC$, then we define
  the \emph{log-prob rate of $F$ at $\nu$} as
  \begin{equation}
    \cO_{\nu}(F;\beta)=\sum_{cz}\nu(cz)\log(F(cz))/\beta
    =\Exp_{\nu}(\log(F))/\beta.
  \end{equation}
\end{definition}

Let $\nu$ be a distribution of $CZ$.  If
$\mu(c_{i+1}z_{i+1}|\Sfnt{R}_{\leq i}\Sfnt{Z}_{\leq i}) =
\nu(C=c_{i+1},Z=z_{i+1})$ for each $i$ and $c_{i+1},z_{i+1}$, then the
trials are \emph{empirically i.i.d.} with distribution $\nu$.  That
is, they are i.i.d.~from the point of view of the user, but not
necessarily from the point of view of $\Pfnt{E}$.  In this case, the
expected log-prob after $n$ trials is given by $n\cO_{\nu}(F;\beta)$.
When using the term ``rate'', we implicitly assume a situation where
we expect many trials with the same distribution.  We omit the
parameters $\beta$ or $\nu$ when they are clear from context.  The
quantity $\cO(F)$ could be called the trial's nominal min-entropy
increase.  We refer to it as ``nominal'' because the final min-entropy
provided to the extractor has to be reduced by
$\log(1/\epsilon)/\beta$ (see Def.~\ref{def:net_log_prob}).  This
reduction diverges as $\beta$ goes to zero.

The log-prob rate $\cO(F)$ is a concave function of $F$, where $F$ is
constrained to a convex set through the inequalities determined by
$\cC$.  The problem of maximizing $\cO(F)$ is therefore that of
maximizing a concave function over a convex domain.  For $(k,l,m)$
Bell-test configurations with $k,l,m$ small and typical settings and
non-signaling constraints, it is feasible to solve this numerically.
However, the complexity of the full problem grows very rapidly with
$k,l,m$.

\subsection{Schema for Probability Estimation}

The considerations above lead to the following schematic probability
estimation protocol. We are given the maximum number of available
trials $n$, a probability estimation goal $q$ and the error bound
$\epse $.  In general, the protocol yields a soft UPE. The protocol
can take advantage of experimentally relevant operational details to
improve probability estimation by adapting PEFs while it is running.
In view of Thm.~\ref{thm:uest_constr} (if $C_{i}=D_{i}$) or
Thm.~\ref{thm:suest_constr} (if not), each trial's PEF can be adapted
by taking advantage of the values of the $R_{j}$ and $Z_{j}$ from
previous trials as well as any initial information that may be part of
$E$.  However, the values of the parameters $n$, $q$, $\epse$ and
$\beta$ in the protocol must be fully determined by initial
information and cannot be adapted later.
\begin{description}[\compact]
\item[Schema for Probability Estimation:]
\item[Given:] $n$, the maximum number of trials, $q$, the probability
  estimation goal, and $\epse$, the error bound.
\item[]
  \begin{description}[\compact]
  \item[1.] Set $T_{0}=1$. Choose $\beta>0$ based on prior data from
    device commissioning or other earlier uses of the devices.
  \item[2.] For each $i=0,\ldots,n-1$:
    \begin{description}[\compact]
    \item[2.1] After the $i$'th and before the $i+1$'th trial:
      \begin{description}[\compact]
      \item[2.1.1] Determine the set $\cC_{i+1}$ constraining
        the distributions at the $i+1$'th trial.
      \item[2.1.2] Estimate the distribution $\nu\in\cC_{i+1}$ of
        $C_{i+1}Z_{i+1}$ from previous measurements (see the remark
        after the schema).
      \item[2.1.3] Based on the estimate $\nu$ and other
          information available in $\Sfnt{R}_{\leq i}$ and
          $\Sfnt{Z}_{\leq i}$, determine a soft PEF $F_{i+1}$ with
        power $\beta$ for $D_{i+1}|Z_{i+1}$ and $\cC_{i+1}$, optimized
        for log-prob rate at $\nu$.
      \end{description}
    \item[2.2] Obtain the $i+1$'th trial result $c_{i+1}z_{i+1}$.
    \item[2.3] After the $i+1$'th trial, compute
      $T_{i+1}=T_{i}F_{i+1}(c_{i+1}z_{i+1})$.
    \item[2.4] If $U=(T_{i+1}\epse )^{-1/\beta}\leq q$, stop and return
      $U$.
    \end{description}
  \item[3.] Return $U=(T_{n}\epse )^{-1/\beta}$.  
  \end{description}
\end{description}
In this schema, we have taken advantage of the ability to stop early
while preserving the required coverage properties.  According to
Thm.~\ref{thm:suest_constr} an implementation of the schema returns
the value of an $\epsilon_h$-soft UPE suitable for use in a randomness
generation protocol with Thm.~\ref{thm:prot_chainlemmas},
Thm.~\ref{thm:estext_nofail}, or Thm.~\ref{thm:estext_fail}.

We remark that there are different ways of estimating the distribution
$\nu[C_{i+1}Z_{i+1}]$.  Our method works as along as the estimated
distribution is determined by $\Sfnt{R}_{\leq i}\Sfnt{Z}_{\leq i}$ and
initial information.  In practice, we estimate the distribution
$\nu\in\cC_{i+1}$ by maximum likelihood using the past trial results
$\Sfnt{C}_{\leq i}\Sfnt{Z}_{\leq i}$, where the likelihood is computed
with the assumption that these trial results are i.i.d. (see
Eq.~\ref{eq:maximum_likelihood} for details).  We stress that the use
of the i.i.d. assumption is only for performing the estimate.  The
probability estimation protocol and the derived randomness generation
protocol are sound regardless of how we obtain the estimate.

\section{Some PEF Properties and PEF Optimality}
\label{sec:pef_props}
\subsection{Dependence on Power}

\begin{lemma}
  \label{lm:power_decrease}
  If $F$ is a PEF with power $\beta$ for $\cC$, then for all
  $0<\gamma\leq 1$, $F^{\gamma}$ is a PEF with power $\gamma\beta$
  for $\cC$.  Furthermore, given a distribution $\rho\in \cC$, 
  the log-prob rate at $\rho$ is the same for $F$ and $F^{\gamma}$.
\end{lemma}

We refer to the transformation of PEFs in the lemma as ``power reduction''
by $\gamma$.

\begin{proof}
  We have, for all $\nu\in\cC$,
  \begin{equation}
    \Exp_{\nu}\left(F^{\gamma}(CZ)\nu(D|Z)^{\gamma\beta}\right)
    = \Exp_{\nu}\left(\left(F(CZ)\nu(D|Z)^{\beta}\right)^{\gamma}\right)
    \leq \left(\Exp_{\nu}\left(F(CZ)\nu(D|Z)^{\beta}\right)\right)^{\gamma} \leq 1,
  \end{equation}
  by concavity of $x\mapsto x^{\gamma}$. Hence $F^{\gamma}$ is a PEF with 
  power $\gamma\beta$ for $\cC$.  Accordingly,
  \begin{equation}
    \cO_{\rho}(F;\beta)=\Exp_{\rho}(\log(F))/\beta
    =\Exp_{\rho}(\gamma\log(F))/(\gamma\beta)=
    \Exp_{\rho}(\log(F^{\gamma}))/(\gamma\beta) = \cO_{\rho}(F^{\gamma};\gamma\beta).
  \end{equation}
\end{proof}

The lemma implies the following corollary:

\begin{corollary}
  The supremum of the log-prob rates at $\rho$ for all PEFs with power
  $\beta$ is non-increasing in $\beta$.
\end{corollary}

Hence, to determine the best log-prob rate without regard to the error
bound, one can analyze PEFs in the limit where the power goes to $0$.

\subsection{PEFs and Entropy Estimation}
\label{subsec:entropy_estimator}

\begin{definition}
  An \emph{entropy estimator for $D|C$ and $\cC$} is a real-valued function
  $K$ of $CZ$ with the property that for all
  $\nu\in\cC$,
  \begin{equation}
    \Exp_{\nu}(K(CZ))\leq \Exp_{\nu}(-\log(\nu(D|Z))).
    \label{eq:def:entropy_estimator}
  \end{equation}
  The \emph{entropy estimate of $K$ at $\nu$} is
  $\Exp_{\nu}(K(CZ))$.
\end{definition}

Entropy estimators are the analog of affine min-tradeoff functions in
the entropy accumulation framework of
Ref.~\cite{arnon-friedman:2018}.  The following
results establish a close relationship between entropy estimators and
PEFs.  First we show that the set of entropy estimators contains the
log-$\beta$-roots of positive PEFs with power $\beta>0$.
\begin{theorem}\label{thm:pef_to_eest}
  Let $F$ be a function of $CZ$ such that $F\ge 0$.  If there exists
  $\beta>0$ such that $F^{\beta}$ is a PEF with power $\beta$ for
  $\cC$, then for all $\nu\in\cC$, $\Exp_{\nu}(\log(F(CZ)))\leq
  \Exp_{\nu}(-\log(\nu(D|Z)))$.
\end{theorem}

\begin{proof}
  Suppose that $F^{\beta}$ is a PEF with power $\beta>0$ for $\cC$.
  For any distribution $\nu\in\cC$,
  \begin{align}
    1 &\geq \Exp_{\nu}\left(F^{\beta}(CZ)\nu(D|Z)^{\beta}\right)\notag\\
    &= \Exp_{\nu}\left(e^{\beta\log(F(CZ))+\beta \log(\nu(D|Z))}\right)\notag\\
    &\geq \Exp_{\nu}\left(1+\beta\log(F(CZ))+\beta\log(\nu(D|Z))\right)\notag\\
    &= 1+\beta\left(\vphantom{\Big|}\Exp_{\nu}(\log(F(CZ))) - \Exp_{\nu}(-\log(\nu(D|Z)))\right).
    \label{eq:thm:pef_to_eest1}
  \end{align}
  Here we used that  $e^{x}\ge 1+x$. It follows
  that $\Exp_{\nu}(\log(F(CZ)))\leq \Exp_{\nu}(-\log(\nu(D|Z)))$.
\end{proof}

According to Thm.~\ref{thm:pef_to_eest}, positive PEFs
$F^{\beta}$ with power $\beta$ define entropy estimators
$\log(F)$ whose entropy estimate at a distribution $\rho\in \cC$
is the log-prob rate of the PEF. Our next goal is to relate
log-prob rates for PEFs with power approaching $0$ to entropy
estimates. For this we need the following definition:

\begin{definition}\label{def:agr}
  The \emph{asymptotic gain rate at $\rho$ of $\cC$} is the supremum of
  the log-prob rates at $\rho$ achievable by PEFs with positive power
  for $\cC$.
\end{definition}
By the properties of PEFs, the asymptotic gain rate at $\rho$ of
$\cC$ is the same as that of $\cCvx(\cC)$.

The asymptotic gain rate is defined in terms of a supremum over
all PEFs, including those whose support is properly contained in
$\Rng(CZ)$.  The next lemma shows that the same supremum is obtained
for positive PEFs.

\begin{lemma}
  \label{lem:g=g+}
  Let $g$ be the asymptotic gain rate at $\rho$ of $\cC$ and
  define 
  \begin{equation}
    g_{+}=\sup\{\cO_{\rho}(F^{\beta};\beta): 
    \textrm{$F>0$ and $F^{\beta}$ is a PEF with power $\beta>0$ for $\cC$}\}.
  \end{equation}
  Then $g=g_{+}$ and the supremum of the entropy estimates at $\rho$ of
  entropy estimators for $\cC$ is at least $g$.
\end{lemma}

\begin{proof}
  It is clear that $g\geq g_{+}$.  Both $g$ and $g_{+}$ are
  non-negative because we can always choose $F=1$ (see
  Thm.~\ref{thm:pef_convex_comb} below).  It suffices to
  consider the non-trivial case $g>0$.  Suppose that $F$
  contributes non-trivially to the supremum defining the asymptotic
  gain rate, namely $F^{\beta}$ is a PEF with power $\beta>0$ and
  log-prob rate $g_{F}=\cO_{\rho}(F^{\beta};\beta)>0$. We can
  choose $F$ so that $g_{F}$ is arbitrarily close to $g$. Consider
  $G^{\beta}=(1-\delta)F^{\beta}+\delta$ with $\delta> 0$
  sufficiently small. Then $G^{\beta}$ is a positive PEF with
  power $\beta$ (see Thm.~\ref{thm:pef_convex_comb}) and has log-prob
  rate
  \begin{align}
    g_{G}=\cO_{\rho}(G^{\beta};\beta)
    &=\Exp_{\rho}\left(\log(G^{\beta})\right)/\beta\notag\\
    &=\Exp_{\rho}\left(\log((1-\delta)F^{\beta}+\delta)\right)/\beta\notag\\
    &\geq \Exp_{\rho}\left(\log((1-\delta)F^{\beta})\right)/\beta\notag\\
    &= g_{F}+ \log(1-\delta)/\beta,
  \end{align}
  by monotonicity of the logarithm.
  We have $\log(1-\delta)/\beta = O(\delta)$. 
  Consequently,
  \begin{equation}
    g_{G}\geq g_{F}+O(\delta).
  \end{equation}
  Since $\delta$ is arbitrary, we can find positive $G$ so that
  $G^{\beta}$ is a PEF with power $\beta$ and log-prob rate $g_{G}$ at
  $\rho$ bounded below by a quantity arbitrarily close to
  $g_{F}$. Since $g_{F}$ is itself arbitrarily close to $g$, we have
  $g=g_{+}$. The last statement of the lemma follows from
  Thm.~\ref{thm:pef_to_eest} and the definitions of asymptotic gain
  rate and entropy estimators.
\end{proof}

According to the lemma above, the supremum of the entropy
estimates at $\rho$ by entropy estimators is an upper bound for the
asymptotic gain rate.  The upper bound could be strict.  But the next
theorem implies that in fact, the supremum of the entropy estimates
gives exactly the asymptotic gain rate. 

\begin{theorem}\label{thm:eest_to_pef}
  Suppose $K$ is an entropy estimator for $\cC$.  Then the asymptotic
  gain rate at $\rho$ is at least the entropy estimate of $K$ at
  $\rho$ given by $\sigma=\Exp_{\rho}(K(CZ))$.
\end{theorem}

\begin{proof}
  To prove the theorem requires constructing families of PEFs with
  small powers whose log-prob rates approach the entropy
  estimate of $K$.  Let $k_{\max}=\max(K)$ and
  $k_{\min}=\min(K)$. We may assume that $k_{\max}> 0$ as otherwise,
  the entropy estimate is not positive.  For sufficiently small
  $\epsilon>0$, we determine $\gamma>0$ such that
  $G(CZ)^{\gamma}=(e^{-\epsilon+K(CZ)})^{\gamma}$ is a PEF with
  power $\gamma$ for $\cC$.  We require that $\epsilon<1/2$.  Consider
  $\nu\in\cC$ and define
  \begin{equation}
    f(\gamma)=
    \Exp_{\nu}\left(G(CZ)^{\gamma}\nu(D|Z)^{\gamma}\right) 
    =\sum_{cz} (e^{-\epsilon+K(cz)}\nu(d|z))^{\gamma}\nu(cz).
  \end{equation}
  $G^{\gamma}$ is a PEF with power $\gamma$ provided that
  $f(\gamma)\leq 1$ for all $\nu$.  We Taylor-expand $f(\gamma)$ at
  $\gamma=0$ for $\gamma\geq 0$ with a second-order remainder.
  \begin{align}
    f(0)&=1,\notag\\
    \frac{d}{d\gamma} f(\gamma) &= \sum_{cz}
    \log(e^{-\epsilon+K(cz)}\nu(d|z))
    (e^{-\epsilon+K(cz)}\nu(d|z))^{\gamma}\nu(cz)\notag\\
    &= \sum_{cz}(K(cz)+\log(\nu(d|z))-\epsilon)
    (e^{-\epsilon+K(cz)}\nu(d|z))^{\gamma}\nu(cz),\notag\\
    \frac{d}{d\gamma} f(\gamma)\Bigr|_{\gamma=0} &=
    \sum_{cz}(K(cz)+\log(\nu(d|z))-\epsilon)
    \nu(cz)\notag\\
    &=\Exp_{\nu}(K(CZ)+\log(\nu(D|Z))) -\epsilon\notag\\
    &\leq -\epsilon,\\
    \frac{d^{2}}{d\gamma^{2}} f(\gamma) &=
    \sum_{cz}(K(cz)+\log(\nu(d|z))-\epsilon)^{2}
    (e^{-\epsilon+K(cz)}\nu(d|z))^{\gamma}\nu(cz).
  \end{align}
  The second derivative is positive, so $f(\gamma)$ is convex
  and
  \begin{equation}
    f(\gamma)\leq 1-\gamma\epsilon+\int_{0}^{\gamma}\frac{d^{2}}{d\gamma^{2}}
    f(\gamma)\Bigr|_{\gamma=\gamma'}(\gamma-\gamma') d\gamma'.
  \end{equation}
  To bound the remainder term, we use $\nu(cz)\leq\nu(dz)=\nu(d|z)\nu(z)$, 
  $e^{-\epsilon}\leq 1$ and $e^{K(CZ)\gamma}\leq e^{k_{\max}\gamma}$  for
  \begin{align}
    \frac{d^{2}}{d\gamma^{2}} f(\gamma)
    &\leq
    \sum_{cz}(K(cz)+\log(\nu(d|z))-\epsilon)^{2}
    \nu(d|z)^{1+\gamma}e^{k_{\max}\gamma}\nu(z).
  \end{align}
  As a function of $\nu(d|z)$ (which is bounded by $1$)
  the summand has a maximum value determined by the
  maximum of $g(x)=(a+\log(x))^{2}x^{1+\gamma}$ for $0<x\leq 1$
  with $a=K(cz)-\epsilon$.
  The critical points of $g(x)$ are obtained from solving
  \begin{equation}
    0 = \frac{d}{dx} g(x) =
    x^{\gamma}(a+\log(x))( (1+\gamma)(a+\log(x)) + 2).
  \end{equation}
  The solutions $x=0$ and $\log(x)=-a$ are minima with $g(x)=0$.
  The remaining solution is obtained from
  $\log(x_{0})=-(a+2/(1+\gamma))$.  The candidate maxima are at this
  solution and at the boundary $x_{1}=1$. 
  The values of $g$ at these points are
  \begin{align}
    g(x_{0})&= \frac{4}{(1+\gamma)^{2}}e^{-(a(1+\gamma)+2)},\notag\\
    g(x_{1})&= a^{2}.
  \end{align}
  If $x_{0}>1$, the solution at $x_{0}$ is irrelevant.  The
  condition $x_{0}\leq 1$ is equivalent to $-(a(1+\gamma)+2)\leq 0$,
  in which case $g(x_{0})\leq 4$.  We have $a^{2}\leq
  \max((k_{\max}-\epsilon)^{2},(k_{\min}-\epsilon)^{2})$, so let
  $u=\max(4,k_{\max}^{2},k_{\min}^{2}+|k_{\min}|+1/4)$, which is a
  loose upper bound on the maximum of $g(x)$ for $0<x\leq 1$. 
  For the bound, we used that $k_{\max}>0$ and $0<\epsilon<1/2$ imply
  $(k_{\max}-\epsilon)^{2}\leq \max(1/4,k_{\max}^{2})$.
  Returning to bounding
  $f$, we get
  \begin{align}
    \frac{d^{2}}{d\gamma^{2}} f(\gamma)
    &\leq
    \sum_{cz}u e^{k_{\max}\gamma}\nu(z)\notag\\
    &=u e^{k_{\max}\gamma}\sum_{c} 1 = u e^{k_{\max}\gamma}|\Rng(C)|,
  \end{align}
  from which
  \begin{align}
    f(\gamma) &\leq 1-\epsilon\gamma
    + \int_{0}^{\gamma} u e^{k_{\max}\gamma'}|\Rng(C)|(\gamma-\gamma') d\gamma'\notag\\
    &\leq 1-\epsilon\gamma + u|\Rng(C)|e^{k_{\max}\gamma}\int_{0}^{\gamma}(\gamma-\gamma')d\gamma'\notag\\
    &= 1-\epsilon\gamma  + u |\Rng(C)|e^{k_{\max}\gamma}\gamma^{2}/2.
  \end{align}
  To ensure that $f(\gamma)\leq 1$
  it suffices to satisfy
  \begin{equation}
    \gamma\leq \frac{2\epsilon e^{-k_{\max}\gamma}}{u|\Rng(C)|}.
  \end{equation}
  Let $w$ be the right-hand side of this inequality.  To remove the
  dependence of $w$ on $\gamma$ while maintaining $f(\gamma)\leq 1$,
  we reduce the upper bound on $\gamma$ in a few steps.  Since
  $e^{-k_{\max}\gamma}\geq 1-k_{\max}\gamma$,
  \begin{equation}
    w\geq w'=\frac{2\epsilon (1-k_{\max}\gamma)}{u|\Rng(C)|}.
  \end{equation}
  If $\gamma\leq w'$, then
  $\gamma\leq 2\epsilon/(u|\Rng(C)|)$,
  so we can substitute this in the right-hand side for 
  \begin{equation}
    w'\geq w''=\frac{2\epsilon (1-2k_{\max}\epsilon/(u|\Rng(C)|))}{u|\Rng(C)|}.
  \end{equation}
  For $\epsilon<u|\Rng(C)|/(4k_{\max})$, this simplifies further to
  $w''\geq w'''=\epsilon/(u|\Rng(C)|)$ by substituting the bound on
  $\epsilon$ for the second occurrence of $\epsilon$ in the
  expression for $w''$. The bound on $\epsilon$ is satisfied since
  $\epsilon<1/2$ is already assumed and from
  $u\geq\max(4,k_{\max}^{2})$, $u/(4k_{\max})\geq 1/2$.  We now
  require 
  \begin{equation}
    \gamma\leq w'''=\frac{\epsilon}{u|\Rng(C)|}=
    \frac{\epsilon}{\max(4,k_{\max}^{2},k_{\min}^{2}+|k_{\min}|+1/4)|\Rng(C)|},
  \end{equation}
  which ensures that $f(\gamma)\leq 1$.
  Since the bound on $f(\gamma)$ is independent of $\nu\in\cC$,
  $G^{\gamma}$ is a PEF with power $\gamma$ and the log-prob rate of
  $G^{\gamma}$ at $\rho$ is
  \begin{align}
    \cO_{\rho}(G^{\gamma};\gamma) &= \cO_{\rho}(e^{-\epsilon \gamma} e^{K \gamma};\gamma) \notag\\
    &=\Exp_{\rho}(K(CZ)-\epsilon)\notag\\
    &= \sigma-\epsilon,
  \end{align}
  which
  approaches $\sigma$ as $\epsilon$ goes to zero, as required to
  complete the proof of the theorem.
\end{proof}

\subsection{Error Bound Tradeoffs}

So far, the results of this section have ignored the contribution of
the error bound $\epse $ to the probability estimate,
giving the appearance that arbitrarily small powers are optimal. For a
given finite number of trials, or if $\epse $ is intended to
decrease exponentially with output min-entropy, the optimal power for
PEFs is bounded away from zero.  This is because $\epse $
increases the probability estimate by a factor of
$\epse ^{-1/\beta}$, which diverges as $\beta$ goes to zero.  We can 
analyze the situation for
the case where $\epse =e^{-\kappa \sigma n}$, where $\sigma$ is the
log-prob rate $\cO_{\rho}(F;\beta)$ at a given distribution
$\rho\in\cC$.  Per trial, we get a \emph{net log-prob rate} of
\begin{equation}
  \cO_{\rho}(F;\kappa;\beta) 
  = \cO_{\rho}(F;\beta)-\kappa\cO_{\rho}(F;\beta)/\beta
  = (1-\kappa/\beta)\Exp_{\rho}(\log(F(CZ)))/\beta.
  \label{eq:netlogprobdef}
\end{equation}

We can consider the family $F^{\gamma}$ of PEFs with power
$\gamma\beta$ obtained by power reduction of $F$. The net log-prob
rate at power $\gamma\beta$ is
\begin{equation}
  \cO_{\rho}(F^{\gamma},\kappa;\gamma\beta) =  \left(1- 
    \frac{\kappa}{\gamma\beta}\right)\cO_{\rho}(F;\beta).
\end{equation}
Consequently, the net log-prob rate is never improved by power
reduction of a given PEF.  However, there are usually PEFs with higher
net log-prob rates but lower powers. There is therefore a tradeoff
between the log-prob rate and the term $\kappa/\beta$ that we
expect to be optimized at some finite non-zero $\beta$. This effect is
demonstrated by example in Sect.~\ref{sec:apps}, see
Fig.~\ref{fig:rosenfeld_lpg(beta)}.

We remark that in the above definitions, we may consider
$\kappa=\kappa(n)$ as a function of $n$.  

\subsection{Optimality of PEFs}
\label{subsect:optimality}

Consider the situation where $R_{i}Z_{i}E_{i}$ are 
i.i.d.~with a distribution $\nu$ such that for all $e$,
$\nu[C_{i}Z_{i}|E_{i}=e_{i}]$ is in $\cC$. Here we have structured the
RV $E$ as a sequence RV, $\Sfnt{E}$.  We call this the
\emph{i.i.d.~scenario} with respect to $\nu$, where we specify $\nu$ as a
distribution of $RZE$ without indices.  We consider the error bound to
be $e^{-o(n)}$.  Our results so far establish that if the
asymptotic gain rate at $\rho=\nu[CZ]$ is $g$, we can certify
smooth conditional min-entropy and extract near-uniform random bits
at an asymptotic rate arbitrarily close to $g$. This is
formalized by the next theorem.

\begin{theorem}
  \label{thm:gainrate_asymptotics}
  Let $g$ be the asymptotic gain rate at $\rho$ of $\cC$ and
  $\epsilon_{n}=e^{-\kappa(n) n}$ with $\kappa(n)=o(1)$. Assume
  the i.i.d.~scenario. Then for any $\delta>0$, there exists a PEF
  $G^{\beta}$ with power $\beta$ for $\cC$ such that with asymptotic
  probability one, the net log-prob after $n$ trials with PEFs
  $F_{i}=G(C_{i}Z_{i})^{\beta}$ and error bound $\epsilon_{n}$
  exceeds $(g-\delta)n$. 
\end{theorem}

\begin{proof}
  By the definition of asymptotic gain rate, there exists a PEF
  $G^{\beta}$ with power $\beta$ such that
  $\Exp(\log(G))=g-\delta/3$. We may assume $g-\delta/3>0$.  The RV
  $S_{n}=\sum_{i=1}^{n}\log(F_{i})/\beta=\sum_{i=1}^{n}\log(G(C_{i}Z_{i}))$
  is a sum of bounded
  i.i.d. RVs, so according to the weak law of
  large numbers, the probability that $S_{n}\geq (g-2\delta/3)n$
  goes to $1$ as $n\rightarrow\infty$.  The net log-prob is given by
  $S_{n}-\log(1/\epsilon)/\beta = S_{n}-\kappa(n)n/\beta$. Since
  $\kappa(n)=o(1)$, for $n$ sufficiently large, $\kappa(n)/\beta\leq
  \delta/3$.  Thus with asymptotic probability $1$, the net log-prob
  after $n$ trials exceeds $(g-\delta)n$. 
\end{proof}

We claim that the asymptotic gain rate $g$ at $\rho$ is equal 
to the minimum of the conditional entropies $H(D|ZE;\nu)$ over all distributions
$\nu$ of $CZE$ such that $\nu[CZ]=\rho$, where the conditional
entropy base $e$ is defined by
\begin{align}
  H(D|ZE;\nu)  &= \sum_{ze}H(\nu[D|ze])\nu(ze)\notag\\
  &= \Exp_{\nu}\left(-\log\left(\nu(D|ZE)\right)\right).
\end{align}
According to the asymptotic equipartition property
(AEP)~\cite{tomamichel:qc2009a} specialized to the case of finite
classical-classical states, the infimum of the conditional entropies
is the optimal randomness-generation rate. In this sense, our method
is optimal. The claim is established by the next theorem.

\begin{theorem} 
  \label{thm:gainrate_optimality}
  Let $g $ be the asymptotic gain rate at $\rho$ of $\cC$.  For all
  distributions $\nu$ of $CZE$ with $\nu[CZ]=\rho[CZ]$, we have $g\leq
  H(D|ZE;\nu)$.  For $|\Rng(E)|$ sufficiently large, there is a
  distribution $\nu$ of $CZE$ with $\nu[CZ]=\rho[CZ]$ such that
  $g=H(D|ZE;\nu)$.
\end{theorem}

\begin{proof}
  Suppose that $\nu$ is a distribution of $CZE$ such that
  $\nu[CZ]=\rho$.  Define $\rho_{e}=\nu[CZ|E=e]\in\cC$ and
  $\lambda_{e}=\nu(e)$.  Then, $\rho=\sum_{e}\lambda_{e}\rho_{e}$.
  Consider an arbitrary entropy estimator $K$ for $\cC$ and write
  $f(\sigma)=\Exp_{\sigma}(K(CZ))$ for its entropy estimate at
  $\sigma$. By definition, $f(\sigma)\leq H(D|Z;\sigma)$ for all
  $\sigma\in\cC$.  Since $f(\sigma)$ is linear in $\sigma$, we have
  \begin{equation} 
    f(\rho)=\sum_{e}\lambda_{e}f(\rho_{e})\leq \sum_{e}\lambda_{e}H(D|Z;\rho_{e})
    = H(D|ZE;\nu).
  \end{equation}
  By arbitrariness of the entropy estimator, the supremum of the
  entropy estimates at $\rho$ is bounded by $H(D|ZE;\nu)$, so by
  Lem. \ref{lem:g=g+}, we have $g\leq H(D|ZE;\nu)$.

  For the second part of the theorem, since the asymptotic gain
  rate of $\cC$ is the same as that of $\cCvx(\cC)$, without loss of
  generality assume that $\cC$ is convex closed.  For any
  distribution $\sigma\in\cC$ define
  \begin{equation}
    h_{\min}(\sigma)=\inf \left\{\sum_{e}\lambda_{e}H(D|Z;\sigma_{e}):
      \sigma_{e}\in\cC,\lambda_{e}\geq 0,\sum_{e}\lambda_{e}=1,\sum_{e}\lambda_{e}\sigma_{e}=\sigma\right\}.
  \end{equation}
  We claim that $h_{\min}(\rho)$ is the supremum of entropy
  estimates at $\rho$ for $\cC$ and that the infimum in the
  definition of $h_{\min}(\rho)$ is achieved by a sum involving a
  bounded number of terms.  The conditional entropy is concave in
  the joint distribution of its variables (see, for example,
  Ref.~\cite{nielsen:qc2001a}, Cor. 11.13, which is readily
  specialized to the classical case).  It follows that if one of the
  $\sigma_e$ contributing to the sum defining $h_{\min}(\sigma)$ is
  not extremal, we can replace it by a convex combination of
  extremal distributions to decrease the value of the sum.  Thus, we
  only have to consider $\sigma_e\in\xtrm{\cC}$ for defining
  $h_{\min}$. It follows that $h_{\min}(\sigma)$ is the convex roof
  extension of the function $\sigma\mapsto H(D|Z;\sigma)$ on
  $\xtrm{\cC}$.  Convex roof extensions are defined in
  Ref.~\cite{uhlmann:qc1998a}.  See this reference for a proof that
  $h_{\min}$ is convex.  In fact, the graph of $h_{\min}$ on $\cC$
  is the lower boundary of the convex closure of the set
  $\{(\sigma,H(D|Z;\sigma )):\sigma\in\xtrm{\cC}\}$.  Specializing
  $h_{\min}(\sigma)$ to the case $\sigma=\rho$, since the dimension
  is finite we can apply Carath\'eodory's theorem and express
  $h_{\min}(\rho)$ as finite convex combinations
  $\sum_{e}\lambda_{e}H(D|Z;\rho_{e})$ with $\rho_{e}\in\xtrm{\cC}$.
  The number of terms required is at most $d+2$, where $d$ is the
  dimension of $\cC$.  Since $h_{\min}$ is convex and has a closed
  epigraph, for any $\epsilon>0$ there exists an affine function $f$
  on $\cC$ such that $f(\rho)\geq h_{\min}(\rho)-\epsilon$ and the
  graph of $f$ is below that of $h_{\min}$ (See
  Ref.~\cite{boyd:qc2004a}, Sect. 3.3.2 and Exercise~3.28).  We can
  extend $f$ to all distributions $\sigma$ of $CZ$ and express
  $f(\sigma)$ as an expectation $f(\sigma)=\Exp_{ \sigma}(K(CZ))$
  for some real-valued function $K$ of $CZ$.  Relevant existence and
  extension theorems can be found in textbooks on convex analysis,
  topological vector spaces or operator theory. For example, see
  Ref.~\cite{kadison:qf1997a}, Ch. 1.  We now have that for all
  $\sigma\in\cC$,
  \begin{equation}
    \Exp_{\sigma}(K(CZ))\leq h_{\min}(\sigma)\leq H(D|Z;\sigma)
    = \Exp_{\sigma}(-\log(\sigma(D|Z))).
  \end{equation}
  It follows that $K$ is an entropy estimator and that its entropy
  estimate at $\rho$ is $h_{\min}(\rho)-\epsilon$. By arbitrariness
  of $\epsilon$, $h_{\min}(\rho)$ is the supremum of entropy
  estimates at $\rho$ for $\cC$, and hence $g=h_{\min}(\rho)$ by
  Lem. \ref{lem:g=g+} and Thm.~\ref{thm:eest_to_pef}. 
  From above, we can write
  $h_{\min}(\rho)=\sum_{e=1}^{d+2}\lambda_{e}H(D|Z;\rho_{e})$, with
  $\lambda_{e}\geq 0$, $\sum_{e=1}^{d+2}\lambda_{e}=1$ and
  $\rho_{e}\in\xtrm{\cC}$. We can set $\Rng(E)=\{1,\ldots,d+2\}$ and
  define $\nu$ on $CZE$ by $\nu(cze)=\rho_{e}(cz)\lambda_{e}$.  Then
  $H(D|ZE;\nu)=\sum_{e}\lambda_{e}H(D|Z;\rho_{e})=h_{\min}(\rho)
  =g$, which is what we aimed to prove.
\end{proof}

\stepcounter{equation}

\section{Additional PEF Constructions}
\label{sec:pef_constructions}

Effectiveness of PEF constraints as discussed in
Sect.~\ref{sec:first_pef_construction} provides a practical way to
construct PEFs when the convex closure of the set $\cC$ has a finite
number of extreme points. We demonstrate this construction in
Sect.~\ref{sec:apps}.  In addition, we can construct PEFs by the
strategies discussed in the following two subsections.

\subsection{PEFs from Maximum Probability Estimators}
\label{sec:maxe_pefs}

A strategy for constructing PEFs is to determine functions $F$ that
solve the inequalities
\begin{equation}
  \begin{array}{rcll}
    F&\geq& 0 & ,\\
    \Exp_{\nu}(F(CZ)\max_{dz}\nu(d|z)^{\beta})&\leq& 1 & \textrm{for $\nu\in\xtrm{\cC}$}.
  \end{array}
  \label{eq:probest_constr_max}
\end{equation}
Any such $F$ is a PEF with power $\beta$ for $\cC$.  The next theorem
provides families of PEFs satisfying these inequalities. Since
the expectations of such PEFs witness the maximum conditional
probability $\max_{dz}\nu(d|z)$ in a trial, the corresponding
probability estimators may be referred to as \emph{maximum
probability estimators}.  

\begin{theorem}\label{thm:pef_uniformbnd}
  Suppose that $B$ is a function of $CZ$ such that
  \begin{equation}
    1-\Exp_{\nu}(B(CZ)) \geq \max_{dz}\nu(d|z)
    \textrm{\ for all $\nu\in\cC$}.\label{eq:maxprobbnd}
  \end{equation}
  Let $\alpha<1$, $\beta>0$, $0\leq \lambda\leq 1$ and define
  \begin{equation}
    F(CZ)=(1-\alpha)^{-\beta}\left(1+\beta\frac{B(CZ)-\alpha}{1-\alpha}\right).
    \label{eq:approxbeta}
  \end{equation}
  If $F'=1+\lambda(F-1)\geq 0$, then $F'$ is a PEF with power $\beta$
  for $\cC$ satisfying Eq.~\ref{eq:probest_constr_max}.
\end{theorem}

A reasonable choice for $\alpha$ in the theorem is $\alpha=\bar
b=\Exp_{\rho}(B(CZ))$ with $\rho$ our best guess for the true
distribution of the trial. The inequality in Eq.~\ref{eq:maxprobbnd}
and the expression for $F$ suggests that we should maximize $\bar b$
for the best results at $\rho$. To optimize the log-prob rate or
the net log-prob rate, we can then vary $\beta$ and $\alpha$.  If the
condition $F'\geq 0$ is not satisfied, we can either reduce $\lambda$
in the definition or replace $B$ by $\gamma B$ and $\alpha$ by $\gamma
\alpha$ for an appropriate $\gamma\in(0,1)$. The latter replacement
preserves the validity of Eq.~\ref{eq:maxprobbnd}.
Thm.~\ref{thm:agrfree} shows that PEFs obtained by these methods
for a given $B$ witness an asymptotic gain rate of at least
$-\log(1-\bar b)$, which justifies the goal of maximizing $\bar b$
and is what we would hope for given the interpretation of the
right-hand-side of Eq.~\ref{eq:maxprobbnd} as a worst-case
conditional max-prob.

Functions $B$ satisfying the conditions in the theorem with $\bar b>0$
for a given non-LR distribution are readily constructed for a
large class of Bell-test configurations.  See the discussions after
the proofs of this and the next theorem. The family of PEFs
constructed accordingly contains PEFs with good log-prob rates,
as witnessed by Thm.~\ref{thm:agrfree} below, which quantifies the
performance of $B$ in terms of $\bar b$.  The family can also be used
as a tool for proving exponential randomness expansion, see
Thm.~\ref{thm:expexp}.

\begin{proof}
  We use the following general inequality: For $x<1$,
  \begin{align}
    (1-x)^{-\beta} &= \left(1-\alpha - (x-\alpha)\right)^{-\beta}\notag\\  
    &= (1-\alpha)^{-\beta}\left(1-\frac{x-\alpha}{1-\alpha}\right)^{-\beta}\notag\\
    &\geq (1-\alpha)^{-\beta}\left(1+\beta\frac{x-\alpha}{1-\alpha}\right),
    \label{eq:powertangent}
  \end{align}
  since the right-hand-side defines the tangent line of the graph of the
  convex function $(1-x)^{-\beta}$ at $x=\alpha$. Provided $F'\geq 0$,
  given any $\nu\in \cC$ we can compute
  \begin{align}
    \Exp_{\nu}\left(F'(CZ)\nu(D|Z)^{\beta}\right) &\leq
    \Exp_{\nu}(F'(CZ))\max_{dz}\nu(d|z)^{\beta}\notag\\
    &= \Exp_{\nu}(1+\lambda(F(CZ)-1))\max_{dz}\nu(d|z)^{\beta}\notag\\
    &= (1-\lambda)\max_{dz}\nu(d|z)^{\beta}
    +\lambda\Exp_{\nu}(F(CZ))\max_{dz}\nu(d|z)^{\beta}\notag\\
    &\leq (1-\lambda) + \lambda\Exp_{\nu}(F(CZ))\max_{dz}\nu(d|z)^{\beta}.
    \label{eq:thm:pef_ub0}
  \end{align}
  To continue, with Eq.~\ref{eq:powertangent} we bound
  \begin{align}
    \Exp_{\nu}(F(CZ))&=
    \Exp_{\nu}\left((1-\alpha)^{-\beta}\left(1+\beta
        \frac{B(CZ)-\alpha}{1-\alpha}\right)\right)\notag\\
    &= \left((1-\alpha)^{-\beta}\left(1+\beta
        \frac{\Exp_{\nu}(B(CZ))-\alpha}{1-\alpha}\right)\right)\notag\\
    &\leq (1-\Exp_{\nu}(B(CZ)))^{-\beta}.
  \end{align}
  To apply Eq.~\ref{eq:powertangent}, note that
  Eq.~\ref{eq:maxprobbnd} implies $\Exp_{\nu}(B(CZ))<1$.
  Substituting in Eq.~\ref{eq:thm:pef_ub0} and applying
  Eq.~\ref{eq:maxprobbnd} gives
  \begin{equation}
    \Exp_{\nu}\left(F'(CZ)\nu(D|Z)^{\beta}\right) \leq
    (1-\lambda)+\lambda(1-\Exp_{\nu}(B(CZ)))^{-\beta}\max_{dz}\nu(d|z)^{\beta}
    \leq 1.
    \label{eq:thm:pef_uniformbnd1}
  \end{equation}
  Since $\nu$ is an arbitrary distribution in $\cC$, we conclude
  that if $F'\geq 0$, it is a PEF with power $\beta$ and, in
  consideration of the first line of Eq.~\ref{eq:thm:pef_ub0}, it
  satisfies Eq.~\ref{eq:probest_constr_max}.
\end{proof}

For $(2,2,2)$ Bell tests with known settings probabilities, starting
from a Bell function $B_{0}$ with positive expectation at $\rho$, one
can determine $m>0$ such that $B=B_{0}/2m$ satisfies
Eq.~\ref{eq:maxprobbnd}.  Computationally, $m$ can be found by
checking the constraints of Eq.~\ref{eq:maxprobbnd} at extremal
$\nu[C|Z]$.  This construction was exploited in
Ref.~\cite{bierhorst:qc2017a}. See the lemma in the proof of the
``Entropy Production Theorem'' in this reference, where setting
$B=(T-1)/2m$ with respect to the notation there defines a function
satisfying Eq.~\ref{eq:maxprobbnd}.  This observation about
non-trivial Bell functions is a consequence of the fact
that
positive expectations of such Bell functions witness the presence of a
Popescu-Rohrlich (PR) box in the distribution, and such a box has
maximum outcome probability $1/2$ for each setting $z$.  (See
Sect.~\ref{sec:apps} for the definition of PR boxes.)

For $\lambda=1$,
Thm.~\ref{thm:pef_uniformbnd} can be generalized to a weighted form
with minor modifications to the proof.

\begin{theorem}\label{thm:pef_uniformbnd2}
  Let $\gamma$ be a positive function of $DZ$.
  Suppose that $B$ is a function of $CZ$ such that
  \begin{equation}
    1-\Exp_{\nu}(B(CZ)) \geq \max_{dz}\gamma(dz)\nu(d|z)
    \textrm{\ for all $\nu\in\cC$}.\label{eq:wmaxprobbnd}
  \end{equation}
  Let $\alpha<1$, $\beta>0$ and define
  \begin{equation}
    F(CZ)=\gamma(DZ)^{\beta}(1-\alpha)^{-\beta}
    \left(1+\beta\frac{B(CZ)-\alpha}{1-\alpha}\right).
    \label{eq:wapproxbeta}
  \end{equation}
  If $F\geq 0$, then $F$ is a PEF with power $\beta$ for $\cC$.
\end{theorem}

\begin{proof}
  It suffices to adjust the sequence of inequalities leading to
  Eq.~\ref{eq:thm:pef_uniformbnd1} as follows
  \begin{align}
    \Exp_{\nu}\left(F(CZ)\nu(D|Z)^{\beta}\right) 
    &=      \Exp_{\nu}\left((1-\alpha)^{-\beta}\left(1+\beta
        \frac{B(CZ)-\alpha}{1-\alpha}\right)(\gamma(DZ)\nu(D|Z))^{\beta}\right) 
    \notag\\
    &\leq
    \Exp_{\nu}\left((1-\alpha)^{-\beta}\left(1+\beta
        \frac{B(CZ)-\alpha}{1-\alpha}\right)\right)\max_{dz}(\gamma(dz)\nu(d|z))^{\beta}\notag\\
    &= \left((1-\alpha)^{-\beta}\left(1+\beta
        \frac{\Exp_{\nu}(B(CZ))-\alpha}{1-\alpha}\right)\right)
    \max_{dz}(\gamma(dz)\nu(d|z))^{\beta} \notag\\
    &\leq (1-\Exp_{\nu}(B(CZ)))^{-\beta}(\max_{dz}\gamma(d|z)\nu(d|z))^{\beta}\notag\\
    &\leq 1. 
  \end{align}
\end{proof}

Eq.~\ref{eq:wmaxprobbnd} defines a set of convex constraints on $B$.
Namely
\begin{equation}
  1-\Exp_{\nu}(B(CZ)) \geq \gamma(dz)\nu(d|z) \textrm{\ for all $dz$ and $\nu\in\cC$}.\label{eq:maxprobbnd_m}
\end{equation}
If we use $B$ to construct PEFs according to
Thm.~\ref{thm:pef_uniformbnd2}, a reasonable goal is to maximize $\bar
b = \Exp_{\rho}(B(CZ))$ subject to these constraints, where $\rho$ is
an estimate of the true distribution.  We remark that the asymptotic
gain rate witnessed by PEFs constructed according to
Thm.~\ref{thm:pef_uniformbnd2} from a given $B$ is at least
$-\log(1-\bar b)+\Exp_{\rho}\log(\gamma(DZ))$, see Thm.~\ref{thm:agrfree}.

To maximize $\bar b$, we can define $\tilde B=1-B$ and solve the
following problem:
\begin{equation}
  \begin{array}[b]{lll}
    \textrm{Minimize:}& \Exp_{\rho}(\tilde B(CZ))&\\
    \textrm{Subject to:}& \Exp_{\nu}(\tilde B(CZ)) \geq \gamma(dz)\nu(d|z) \textrm{\ for all $dz$ and $\nu\in\cC$}.
  \end{array}\label{eq:opt_bancal}
\end{equation}
Given an optimal $B$, one can use Eq.~\ref{eq:wapproxbeta}
to define PEFs, choosing parameters to optimize the log-prob rate
at $\rho$. If $F$ as constructed does not satisfy 
$F\geq 0$, we can replace $B$ by $\gamma' B$
and $\alpha$ by $\gamma' \alpha$ for an appropriate $\gamma'\in(0,1)$.
Alternatively, we can reduce the power $\beta$.

For general $\cC$, solving Eq.~\ref{eq:opt_bancal} may be difficult.
But suppose that the measurement settings distribution $\nu[Z]$ is
fixed, $\nu[Z]=\rho[Z]$ for all $\nu\in\cC$, and the conditional
distributions $\nu[C|Z]$ belong to a convex set $\cC_{C|Z}$ determined
by semidefinite constraints. Now
$\cC=\{\nu(C|Z)\rho(Z):\nu[C|Z]\in\cC_{C|Z}\}$, which is a special
case of the sets free for $Z$ defined in Sect.~\ref{subsec:freeforz}.
This is the standard situation for Bell configurations, in which case
Eq.~\ref{eq:opt_bancal} is related to the optimization problems
described in
Refs.~\cite{nieto:2014,bancal:2014,nieto-silleras:qc2016a}.  These
references define convex programs that determine the maximum available
min-entropy for one trial with distribution $\rho$. In fact, the
program given in Eq.~(8) of Ref.~\cite{bancal:2014} is related to the
dual of Eq.~\ref{eq:opt_bancal} when $\cC_{C|Z}$ is the set of quantum
realizable conditional probability distributions. To make the
relationship explicit and show that with the given assumptions,
Eq.~\ref{eq:opt_bancal} is effectively solvable, define $\hat B$ by
$\hat B(cz)=\tilde B(cz)\rho(z)$.  With explicit sums,
Eq.~\ref{eq:opt_bancal} becomes
\begin{equation}
  \begin{array}[b]{lll}
    \textrm{Minimize:}& \sum_{c'z'}\hat B(c'z')\rho(c'|z')&\\
    \textrm{Subject to:}& \sum_{c'z'}\hat B(c'z')\nu(c'|z') \geq \gamma(dz)\nu(d|z) \textrm{\ for all $dz$ and $\nu[C|Z]\in\cC_{C|Z}$}.
  \end{array}\label{eq:opt_bancalsums}
\end{equation}
Note that $\nu[d|z]$ is a linear function of $\nu[C|Z]$. Let $\hat
B_{\min}$ minimize the objective function of
Eq.~\ref{eq:opt_bancalsums} with minimum value $1-\bar b$.  Such a
solution exists since $\hat B(cz)=\max(\gamma)\rho(z)$ is feasible,
which implies that the minimum value of the objective function is less
than or equal to $\max(\gamma)$. Moreover, the constraints imply that
the minimum is positive.  By construction, $(\hat B_{\min},-1)$ is in
the cone $\cD^{*}$ dual to the closed convex cone $\cD$ generated by
$\cD_{1}=\{(\nu[C|Z],\gamma(dz)\nu(d|z)):\nu[C|Z]\in\cC_{C|Z},dz\in\Rng(DZ)\}$.
Further, $\vec x=(\rho(C|Z),(1-\bar b))$ defines a supporting
hyperplane of $\cD^{*}$ at $(\hat B_{\min},-1)$: Since $(\hat
B_{\min},-1)\cdot (\rho(C|Z),(1-\bar b))=0$, it suffices to check that
$\vec x$ is in the dual of $\cD^{*}$.  Non-zero elements of $\cD^{*}$
are positive linear combinations of elements of the form $(\hat B,-1)$
for some $\hat B$ and $(\sigma,a)$ for some $\sigma$ and $a\geq 0$.
The first satisfies that $\hat B$ is a feasible solution of
Eq.~\ref{eq:opt_bancalsums} and by definition of $\bar b$, $(\hat
B,-1)\cdot \vec x=\sum_{c'z'}\hat B(c'z')\rho(c'|z')-(1-\bar b)\geq
0$. For the second form, we have $(\sigma,a)\cdot \vec x =
\sum_{c'z'}\sigma(c'z')\rho(c'|z')+a(1-\bar b)$. For all $dz$,
$\sum_{c'z'}\sigma(c'z')\rho(c'|z')\geq -a\gamma(dz)\rho(d|z)$, and
since $(\hat B_{\min},-1)\in\cD^{*}$, $(1-\bar b)\geq
\gamma(dz)\rho(d|z)$.  Hence $(\sigma,a)\cdot\vec x \geq
-a\gamma(dz)\rho(d|z) + a\gamma(dz)\rho(d|z)=0$, in consideration of
$a\geq 0$.  The dual of $\cD^{*}$ is $\cD$ again, so $\vec x\in\cD$.
Because of probability normalization, $\cD_{1}$ is contained in
several hyperplanes not containing the origin, so $\cD$ is
pointed. The intersection of any of these hyperplanes with $\cD$ is
$\cD_{1}$. Since the set $\cD_{1}$ is closed and bounded, it meets
every extremal ray of $\cD$ where it intersects these hyperplanes (Thm
1.4.5 of Ref.~\cite{kadison:qf1997a}).  Consequently,
$(\rho[C|Z],(1-\bar b))$ is a convex combination of elements of
$\cD_{1}$. Since we are in finite dimensions we can apply
Carath\'eodory's theorem and express $\rho[C|Z]$ as a finite convex
combination $\sum_{dzk}\lambda_{dzk}\nu_{dzk}(c'|z')=\rho(c'|z')$
satisfying $\sum_{dzk}\lambda_{dzk}\gamma(dz)\nu_{dzk}(d|z)=(1-\bar
b)$ with $\nu_{dzk}\in\cC_{C|Z}$, $\sum_{dzk}\lambda_{dzk}=1$ and
$\lambda_{dzk}\geq 0$.  If we define
$\lambda_{dz}=\sum_{k}\lambda_{dzk}$ and
$\nu_{dz}=\sum_{k}\lambda_{dzk}\nu_{dzk}/\lambda_{dz}$ if
$\lambda_{dz}>0$ (otherwise, $\nu_{dz}$ can be any member of the set
$\cC_{C|Z}$), then $k$ can be eliminated in the convex combination.
By construction, $(1-\bar b)$ is the maximum value of
$\sum_{dz}\lambda_{dz}\gamma(dz)\nu_{dz}(d|z)$ for any family
$\nu_{dz}\in\cC_{C|Z}$ and $\lambda_{dz}\geq 0$ satisfying that for
all $c'z'$, $\sum_{dz}\lambda_{dz}\nu_{dz}(c'|z')=\rho(c'|z')$ and
$\sum_{dz}\lambda_{dz}=1$. We can write
$\tilde\nu_{dz}=\lambda_{dz}\nu_{dz}$ to absorb the coefficients
$\lambda_{dz}$.  With this, the value of the following problem is
$(1-\bar b)$:
\begin{equation}
  \begin{array}[b]{lll}
    \textrm{Maximize:}& \sum_{dz}\gamma(dz)\tilde\nu_{dz}(d|z)&\\
    \textrm{Subject to:}& \sum_{dz}\tilde\nu_{dz}(c'|z')=\rho(c'|z')&
    \textrm{for all $c'z'$},\\
    & \tilde\nu_{dz}\in[0,\infty)\cC_{C|Z}& \textrm{for all $dz$}.
  \end{array}\label{eq:optorig_bancal}
\end{equation}
The set $[0,\infty)\cC_{C|Z}$ is the cone generated by $\cC_{C|Z}$.
If $\cC_{C|Z}$ is characterized by a semidefinite program, then so
is $[0,\infty)\cC_{C|Z}$ (by eliminating inhomogenous constraints in
the semidefinite program's standard form, see
Ref.~\cite{boyd:qc2004a}, Eq. (4.51)).  Eq.~\ref{eq:optorig_bancal}
can then be cast as a semidefinite program also and solved
effectively. By semidefinite-programming duality, this can then be
used to obtain effective solutions of Eq.~\ref{eq:opt_bancal},
provided the semidefinite program is formulated to satisfy strong
duality. For related programs, Ref.~\cite{nieto:2014} claim strong
duality by exhibiting strictly feasible solutions. 

We conclude with remarks on the relationship between the
optimization problem in Eq.~\ref{eq:optorig_bancal} and that in
Eq.~(8) of Ref.~\cite{bancal:2014} (referenced as ``P8'' below). To
relate Eq.~\ref{eq:optorig_bancal} to P8, set $\gamma(dz)=1$ and let
$\cC_{C|Z}$ be the set of quantum achievable conditional probability
distributions, for which there is a hierarchy of
semidefinite-programming relaxations~\cite{navascues:2007}.  Then
identify both $c$ and $d$ here with $ab$ there (so $c=d$), $z$ here
with $xy$ there, and $\tilde\nu_{dz}(c'|z')$ here with
$\rho(z)P_{cz}(c'|z')$ there, where
$P_{cz}(c'|z')=\sum_{\mathbf{\alpha\beta}:\mathbf{\alpha}_{z}\mathbf{\beta}_{z}=c}
P_{\mathbf{\alpha\beta}}(c'|z')$.  The objective function of
Eq.~\ref{eq:optorig_bancal} now matches that of P8.  But unless there
is only one setting, the equality constraints in
Eq.~\ref{eq:optorig_bancal} are a proper subset of those of P8.
Observe that the equality constraints of P8 when expressed in terms of
the variables $P_{cz}(c'|z')$ require that for each $z$ and $c'z'$,
$\sum_{c}P_{cz}(c'|z')=\rho(c'|z')$. For this, note that P8 includes
the constraints $\sum_{\mathbf{\alpha\beta}}
P_{\mathbf{\alpha\beta}}(c'|z')=\rho(c'|z')$ for all $c'z'$.  For any
$z$, the left-hand side can be written as
$\sum_{c}\sum_{\mathbf{\alpha\beta}:\mathbf{\alpha}_{z}\mathbf{\beta}_{z}=c}
P_{\mathbf{\alpha\beta}}(c'|z')=\sum_{c}P_{cz}(c'|z')$.  These
identities imply that
$\sum_{cz}\rho(z)P_{cz}(c'|z')=\sum_{z}\rho(z)\rho(c'|z')=\rho(c'|z')$,
but are stronger when $|\Rng(Z)|>1$.  Furthermore, according to P8,
the $P_{cz}$ must be expressed as sums of the
$P_{\mathbf{\alpha}\mathbf{\beta}}\in\cC_{C|Z}$ as specified above,
and this implies additional constraints.  As a result the optimal
value for P8 is in general smaller.

\subsection{Convex Combination}

Optimization over all PEFs with a given power can be highly demanding,
because both the size of the range of $CZ$ and the number of extreme
points of $\cC$ in a general $(k,l,m)$ Bell-test configuration are
large. Furthermore, the size of the range of $CZ$ determines the
dimension of the search space, so if this size is large and the amount
of data available for making an estimate of the true distribution is
limited, there is a risk of overfitting when optimizing the log-prob
rate for the estimated distribution at a given power
$\beta$. See Eq.~\eqref{eq:opt_pef} of Sect.~\ref{sec:apps} for the
explicit formulation of the optimization problem.  
However, in many cases we have a small set of candidate PEFs expected
to be helpful in a given situation, where the set was obtained in
earlier studies of the experimental devices before running any
protocols.  Then, the
optimization problem can be greatly simplified by means of the next
theorem.

\begin{theorem}
  \label{thm:pef_convex_comb}
  Let $(F_{i})_{i=1}^{r}$ be a collection of PEFs with power $\beta$
  and define $F_{0}=1$. Then every weighted average
  $F=\sum_{i=0}^{r}\lambda_{i}F_{i}$ with $\lambda_{i}\geq 0$
  and $\sum_{i}\lambda_{i}=1$ is a PEF with power $\beta$.
\end{theorem}

\begin{proof}
  It suffices to observe that $F_{0}=1$ is a PEF for all powers,
  and the set of PEFs with power $\beta$ is convex closed 
  since it is defined by a family of linear constraints.
\end{proof}

Optimizing PEFs with respect to the coefficients of weighted averages
is efficient and less susceptible to overfitting issues.  If $\rho$ is
estimated by empirical frequencies from past trials, one can evaluate
the objective function directly on the past data, in which case the
technique is not limited to configurations with computationally
manageable $|\Rng(CZ)|$.  This strategy was proposed for optimizing
test factors for rejecting LR in Ref.~\cite{zhang:2013} and used in
Refs.~\cite{knill:qc2014a,christensen:qc2015a}.

\section{Reducing Settings Entropy}
\label{sec:reduce_entropy}

\subsection{Sets of Distributions that are Free for $Z$}
\label{subsec:freeforz}    

In our applications, the sets of distributions $\cC$ are
determined by constraints on the $Z$-conditional distributions of $C$
and separate constraints on the marginal distributions of $Z$.  We
first formalize this class of sets. 
As defined in Sect.~\ref{sec:prelims}, for any
RV $X$, $\cS_{X}=\{\rho: \text{$\rho$ is a distribution of $X$}\}$
and $\cCvx(\cX)$ is the convex closure of $\cX$.

\begin{definition}
  Let $\cS_{C|Z}=\{(\rho_{C|z})_{z\in\Rng(Z)}: \text{for each $z$,
    $\rho_{C|z}$ is a distribution of $C$}\}$. 
  The set of distributions $\cC$ of $CZ$ is \emph{free for $Z$} if
  there are closed convex sets $\cC_{C|Z}\subseteq\cS_{C|Z}$ and
  $\cC_{Z}\subseteq\cS_{Z}$ such that
  \begin{equation}\label{eq:freeforZ}
    \cC=\cCvx\left(\{\nu\in\cS_{CZ}: (\nu[C|Z=z])_{z}\in\cC_{C|Z} \textrm{\ and\ } \nu[Z]\in\cC_{Z}\}\right).
  \end{equation}
\end{definition}

In the case of Bell tests, $\cC_{C|Z}$ consists of the set of
non-signaling distributions, possibly satisfying additional quantum
constraints.  If $Z$ is a settings choice RV with known distribution
$\nu$, then $\cC_{Z}=\{\nu\}$.  For $\rho_{C|Z} =
(\rho_{C|z})_{z}\in\cS_{C|Z}$, we define $\rho_{C|Z} \rtimes\nu$ by
$(\rho_{C|Z}\rtimes\nu)(cz) = \rho_{C|z}(c)\nu(z)$.  With this, $\cC$
in Eq.~\ref{eq:freeforZ} can be written as
$\cCvx(\cC_{C|Z}\rtimes\cC_{Z})$.  If $\cC_{Z}=\{\nu\}$, then
$\cC_{C|Z}\rtimes\cC_{Z}=\cC_{C|Z}\rtimes\nu$ is already convex
closed, so we can omit the convex-closure operation.

The extreme points of $\cC_{C|Z}\rtimes\nu$ consist of the set of
$\rho_{C|Z}\rtimes\nu$ with $\rho_{C|Z}$ extremal in $\cC_{C|Z}$.  In
general, we have:

\begin{lemma}\label{lm:cvxrtimes}
  $\xtrm{\cCvx(\cC_{C|Z}\rtimes\cC_{Z})}\subseteq\overline{\xtrm{\cC_{C|Z}}}\rtimes\overline{\xtrm{\cC_{Z}}}$
\end{lemma}
Here, the overline on the right-hand side denotes topological closure.  

\begin{proof}
  $\cC_{C|Z}$ and $\cC_{Z}$ are bounded closed sets with finite
  dimension, hence compact. The operation $\rtimes$ is continuous and
  therefore maps bounded closed sets to bounded closed sets. So
  $\cC_{C|Z}\rtimes\cC_{Z}$ is bounded and closed, as is
  $\overline{\xtrm{\cC_{C|Z}}}\rtimes\overline{\xtrm{\cC_{Z}}}$, which
  is contained in $\cC_{C|Z}\rtimes\cC_{Z}$.  From the last property,
  $\cCvx(\cC_{C|Z}\rtimes\cC_{Z})\supseteq
  \cCvx\left(\overline{\xtrm{\cC_{C|Z}}}\rtimes\overline{\xtrm{\cC_{Z}}}\right)$.

  By bilinearity of $\rtimes$, we have
  $\cC_{C|Z}\rtimes\cC_{Z}\subseteq
  \cCvx(\overline{\xtrm{\cC_{C|Z}}}\rtimes\overline{\xtrm{\cC_{Z}}})$.
  Since convex closure is idempotent,
  $\cCvx(\cC_{C|Z}\rtimes\cC_{Z})\subseteq
  \cCvx(\overline{\xtrm{\cC_{C|Z}}}\rtimes\overline{\xtrm{\cC_{Z}}})$.
  We conclude that these two convex sets are identical.
  Every bounded closed set contains the extreme points of its convex
  closure (Thm 1.4.5 of
  Ref.~\cite{kadison:qf1997a}).  Accordingly,
  $\xtrm{\cCvx(\cC_{C|Z}\rtimes\cC_{Z})}\subseteq
  \overline{\xtrm{\cC_{C|Z}}}\rtimes\overline{\xtrm{\cC_{Z}}}$.
\end{proof}

We note that the members of $\xtrm{\cC_{C|Z}}\rtimes\xtrm{\cC_{Z}}$
are extremal. This can be seen as follows: Let the member 
$\rho_{C|Z} \rtimes\nu$ of $\xtrm{\cC_{C|Z}}\rtimes\xtrm{\cC_{Z}}$ be a finite
convex combination of members of $\cC_{C|Z}\rtimes\cC_{Z}$, written as $\rho_{C|Z}
\rtimes\nu=\sum_{i}\lambda_{i}\rho_{C|Z;i} \rtimes\nu_{i}$.  By
marginalizing to $Z$ and by extremality of $\nu$, it follows that
$\nu_{i}=\nu$ for all $i$.  Given this, and extremality of
$\rho_{C|Z}$, we can also conclude that $\rho_{C|Z;i}=\rho_{C|Z}$.
Thus the inclusion in Lem.~\ref{lm:cvxrtimes} becomes an equality if
we topologically close the left-hand side. 
In particular, if $\cC_{C|Z}$ and
$\cC_{Z}$ are polytopes, then the right-hand side is finite and
therefore $\cCvx(\cC_{C|Z}\rtimes\cC_{Z})$ is also a polytope with the
expected extreme points, a fact that we exploit for finding PEFs by
our methods.

We fix the convex set $\cC_{C|Z}\subseteq\cS_{C|Z}$ for the remainder
of this section.  For brevity, instead of referring to properties
``for $\cC=\cCvx(\cC_{C|Z}\rtimes\cC_{Z})$'', we just say ``for
$\cC_{Z}$'', provided $\cC_{C|Z}$ is clear from context.

\subsection{Gain Rates for Biased Settings}

For randomness expansion, it is desirable to minimize the entropy of
$\mu[Z]$, since this is the main contributor to the number of random
bits required as input to the protocol. The other contributor is the
seed required for the extractor, but good extractors already use
relatively little seed. We expect that reducing the entropy of
$\mu[Z]$ may reduce the asymptotic gain rate. However, in the
case where we have uniform bounds on the maximum probability as
required for Thm.~\ref{thm:pef_uniformbnd2}, we can show that the
reduction is limited.

\begin{theorem}\label{thm:agrfree}
  Let $\rho_{C|Z}\in\cC_{C|Z}$.  Let $\gamma$ be a positive
  function of $DZ$.  Let $B$ be a function of $CZ$ satisfying
  Eq.~\ref{eq:wmaxprobbnd} for $\cC_{Z}=\{\Unif_{Z}\}$.  Assume
  that $\bar b=\Exp_{\rho_{C|Z}\rtimes\Unif_{Z}}(B)>0$. Let $\nu$ be a
  distribution of $Z$ with $p_{\min}=\min_{z}\nu(z)>0$.  Then the
  asymptotic gain rate at $\rho_{C|Z}\rtimes\nu$ for $\cC_{Z}=\{\nu\}$
  is at least $-\log(1-\bar
  b)+\Exp_{\rho_{C|Z}\rtimes\nu}\log(\gamma(DZ))$.
\end{theorem}

\begin{proof}
  Let $q=1/|\Rng(Z)|$ and $w=-\min(0,\min(B))$.  Define
  $B'(CZ)=B(CZ)\Unif_{Z}(Z)/\nu(Z)=B(CZ)q/\nu(Z)$. For all
  $\sigma_{C|Z}\in\cS_{C|Z}$,
  \begin{align}
    \Exp_{\sigma_{C|Z}\rtimes \nu}(B') 
    &= \sum_{z} \Exp_{\sigma_{C|z}}(B'(Cz))\nu(z)\notag\\
    &= \sum_{z} \Exp_{\sigma_{C|z}}\left(B(Cz)q/\nu(z)\right)\nu(z)\notag\\
    &= \sum_{z} \Exp_{\sigma_{C|z}}(B(Cz))q\notag\\
    &=\Exp_{\sigma_{C|Z}\rtimes\Unif_{Z}}(B).\label{eq:thm:agrfree1}
  \end{align}
  By our assumptions on $B$, for all $\sigma_{C|Z}\in\cC_{C|Z}$,
  $1-\Exp_{\sigma_{C|Z}\rtimes\Unif_{Z}}B(CZ)\geq
  \max_{dz}\gamma(dz)\Prob_{\sigma_{C|Z}\rtimes\Unif_{Z}}(d|z)$, which is
  Eq.~\ref{eq:wmaxprobbnd}.  Since
  $\Prob_{\sigma_{C|Z}\rtimes\nu}(d|z)=\Prob_{\sigma_{C|Z}\rtimes\Unif_{Z}}(d|z)$,
  we can apply Eq.~\ref{eq:thm:agrfree1} to conclude that $B'$
  satisfies Eq.~\ref{eq:wmaxprobbnd} for $\cC_{Z}=\{\nu\}$.  According
  to Thm.~\ref{thm:pef_uniformbnd2},
  \begin{equation}
    F_{\beta}(CZ)=\gamma(DZ)^{\beta}(1-\bar b)^{-\beta}\left(1+\beta\frac{B'(CZ)-\bar b}{1-\bar b}\right)
  \end{equation}
  is a PEF with power $\beta$ for $\cC_{Z}=\{\nu\}$ provided that
  $F_{\beta}\geq 0$.  Note that Eq.~\ref{eq:wmaxprobbnd} implies $\bar
  b<1$.  Since $B'\ge -wq/p_{\min}$, for sufficiently small $\beta$,
  the condition $F_{\beta}\geq 0$ is satisfied. Specifically, we
  require that $\beta<1/a$, where we define $a = (wq/p_{\min}+\bar
  b)/(1-\bar b)$.  The asymptotic gain rate $g$ from the theorem
  statement satisfies
  \begin{equation}
    g \geq \limsup_{\beta\rightarrow 0_{+}} 
    \Exp_{\rho_{C|Z}\rtimes\nu}(\log(F_{\beta}))/\beta.\label{eq:gliminfbeta}
  \end{equation}
  The expression inside the limit is
  \begin{equation}
    \Exp_{\rho_{C|Z}\rtimes\nu}(\log(F_{\beta}))/\beta
    = -\log(1-\bar b) + \Exp_{\rho_{C|Z}\rtimes\nu}\log(\gamma(DZ)) +
    \frac{1}{\beta}\Exp_{\rho_{C|Z}\rtimes\nu}\log\left(1+\beta\frac{ B'(CZ)-\bar b}{1-\bar b}\right). \label{eq:log1mbarb}
  \end{equation}
  In general, given $\beta a<1$ and $x\geq-a$, we can
  approximate $\log(1+\beta x)$ from Taylor expansion around $x=0$ with a
  second-order remainder to get
  \begin{equation}
    \log(1+\beta x)\geq \beta x - \frac{1}{2(1-\beta a)^{2}}(\beta x)^{2},
    \label{eq:log1pbetax}
  \end{equation}
  since $|\frac{d^{2}}{dy^{2}}\log(1+y)|\leq 1/(1-\beta a)^{2}$ for
  $y\geq-\beta a$.  Let $X=(B'(CZ)-\bar b)/(1-\bar b)$. Then
  $\Exp_{\rho_{C|Z}\rtimes\nu}X = 0$. Let
  $v=\Exp_{\rho_{C|Z}\rtimes\nu}(X^{2})$, which is the variance of $X$.
  Substituting $X$ for $x$ in the inequality of
  Eq.~\ref{eq:log1pbetax} and applying expectations to both sides
  gives
  \begin{align}
    \Exp_{\rho_{C|Z}\rtimes\nu}\log\left(1+\beta X\right)
    &\geq -\frac{\beta^{2}v}{2(1-\beta a)^{2}}\notag\\
    &= -O(\beta^{2}),
    \label{eq:obeta2}
  \end{align}
  where the minus sign on the order notation emphasizes that the
  lower bound is negative.  The expression in Eq.~\ref{eq:log1mbarb}
  is therefore bounded below by $-\log(1-\bar b) +
  \Exp_{\rho_{C|Z}\rtimes\nu}\log(\gamma(DZ))-O(\beta)$
  which goes to $-\log(1-\bar b)+
  \Exp_{\rho_{C|Z}\rtimes\nu}\log(\gamma(DZ))$ as
  $\beta\rightarrow 0_{+}$.  In view of Eq.~\ref{eq:gliminfbeta},
  the theorem follows.  
\end{proof}

\subsection{Spot-Checking Settings Distributions}

To approach the asymptotic gain rate in Thm.~\ref{thm:agrfree}
requires small powers which negatively impact the net log-prob rate.
We analyze the effect on the net log-prob rate in the case where $\mu$
is a mixture of a deterministic distribution $\delta_{z_{0}}$ and
$\Unif_{Z}$.  This corresponds to using the same setting for most
trials and randomly choosing \emph{test} trials with uniform settings
distribution. This is referred to as a \emph{spot-checking}
strategy~\cite{miller_c:qc2014b}. Later we consider a protocol where
one random trial out of a block of $2^{k}$ trials has $Z$ uniformly
distributed.

To simplify the analysis, we take advantage of the fact that for
configurations such as those of Bell tests, we can hide the choice of
whether or not to apply a test trial from the devices. This
corresponds to appending a test bit $T$ to $Z$, where $T=1$ indicates
a test trial and $T=0$ indicates a fixed one, with $Z=z_{0}$.  The
set $\cC_{C|ZT}$ is obtained from $\cC_{C|Z}$ by constraining
$\mu[C|zt]$ to be independent of $t$ for each $z$ and
$(\mu[C|zt])_{z}\in\cC_{C|Z}$. For any $\nu_{C|Z}\in\cC_{C|Z}$ there is a
corresponding $\tilde\nu_{C|ZT}\in\cC_{C|ZT}$ defined by
$\tilde\nu_{C|zt}(c)=\nu_{C|z}(c)$ for all $c$, $z$ and $t$.  The map
$\nu_{C|Z}\mapsto \tilde\nu_{C|ZT}$ is a bijection.

Let $q=1/|\Rng(Z)|$ and $\rho_{C|Z}\in\cC_{C|Z}$.  Let $\nu_{r}$ be
the probability distribution of $ZT$ defined by $\nu_{r}(z1) = rq$ and
$\nu_{r}(z0)=(1-r)\delta_{z,z_{0}}$ for some value $z_{0}$ of
$Z$. Since we are analyzing the case where $r$ is small, we assume $r<
1/2$.  The entropy of the distribution $\nu_{r}$ is given by
$S(\nu_{r})=H(r)+r\log(1/q)$, where $H(r)=-r\log(r)-(1-r)\log(1-r)$.
Let $B$ be a function of $CZ$ satisfying Eq.~\ref{eq:maxprobbnd} for
$\cC_{Z}=\{\Unif_{Z}\}$ with $\bar b
=\Exp_{\rho_{C|Z}\rtimes\Unif_{Z}}(B)>0$ and $\bar v =
\Exp_{\rho_{C|Z}\rtimes\Unif_{Z}}((B-\bar b)^{2})$. 
From Eq.~\ref{eq:maxprobbnd}, it follows that
$\bar b<1$.  Define $B_{r}(CZT)$ by
\begin{align}
  B_{r}(CZ0) &= 0,\notag\\
  B_{r}(CZ1) &= B(CZ)\frac{1}{r}.
  \label{eq:def_B_r}
\end{align}
Let
\begin{equation}
  F_{r,\beta}(CZT) = (1-\bar b)^{-\beta}\left(1+\beta\frac{B_{r}(CZT)-\bar b}{1-\bar b}\right).
  \label{eq:pef_r_beta}
\end{equation}
Setting the function $B_{r}$ to zero when $\{T=0\}$ is convenient.
Because $F_{r,\beta}$ is related to the tangent line of
$(1-x)^{-\beta}$ at $x=\bar b$ and $B_{r}=0$ corresponds to $x=0$, the
expected value of $F_{r,\beta}$ is slightly below $1$ for non-test
trials.

\begin{theorem}\label{thm:frbeta_gain} There exist
  constants $d>0$, $d'\geq 0$ independent of $r$ such that for
  $\beta\leq dr$, $F_{r,\beta}$ as defined in Eq.~\ref{eq:pef_r_beta} 
  is a PEF with power $\beta$ for
  $\cC_{C|ZT}\rtimes\nu_{r}$, and its log-prob rate
  $g_{r,\beta}$ at $\tilde\rho_{C|ZT}\rtimes\nu_{r}$ is
  $g_{r,\beta}\geq -\log(1-\bar b)- d'\beta/r$.
\end{theorem}

In practical situations, we anticipate using numerical optimization to
determine the PEFs, which we expect to improve on the bounds
in the theorem.  However, we believe that the constants $d$ and $d'$
obtained in the proof are reasonable. They are given by
\begin{align}
  d &= \frac{1-\bar b}{2w+\bar b},\notag\\
  d' &=2\frac{\bar v + \bar b^{2}}{(1-\bar b)^{2}},
  \label{eq:variance_constants}
\end{align}
with $w=-\min(0,\min(B))$.

\begin{proof}
  The proof is a refinement of that of Thm.~\ref{thm:agrfree}.  For
  any $\sigma_{C|Z}\in\cS_{C|Z}$, consider the expectation of $B_{r}$
  with respect to $\tilde\sigma_{C|ZT}\rtimes\nu_{r}$.
  \begin{align}
    \Exp_{\tilde\sigma_{C|ZT}\rtimes\nu_{r}}(B_{r}) &=
    \sum_{z}\Exp_{\tilde\sigma_{C|z0}}B_{r}(Cz0)\nu_{r}(z0) +
    \sum_{z}\Exp_{\tilde\sigma_{C|z1}}B_{r}(Cz1)\nu_{r}(z1)
    \notag\\
    &= \sum_{z}\Exp_{\sigma_{C|z}}B_{r}(Cz1)qr
    \notag\\
    &= 
    \sum_{z}\Exp_{\sigma_{C|z}}B(Cz)q
    \notag\\
    &= \Exp_{\sigma_{C|Z}\rtimes\Unif_{Z}}B(CZ).\label{eq:thm:frbeta_gain1}
  \end{align}
  Since $\Prob_{\tilde\sigma_{C|ZT}\rtimes\nu_{r}}(d|zt)
  =\Prob_{\sigma_{C|Z}\rtimes\Unif_{Z}}(d|z)$ and $B$ satisfies
  Eq.~\ref{eq:maxprobbnd} for $\cC_{Z}=\{\Unif_{Z}\}$, so does $B_{r}$
  but for $\cC_{ZT}=\{\nu_{r}\}$.  Thus Thm.~\ref{thm:pef_uniformbnd}
  applies, and we have that $F_{r,\beta}$ as defined in
  Eq.~\eqref{eq:pef_r_beta} is a PEF with power $\beta$ for
  $\cC_{C|ZT}\rtimes\nu_{r}$, provided $\beta$ is small enough.  From 
  Eq.~\eqref{eq:def_B_r} and the fact $\min(B)\geq -w$, we
  have that $B_{r}\ge -w/r$.  From the proof of
  Thm.~\ref{thm:agrfree}, we can replace $wq/p_{\min}$ by $w/r$
  in the expression for $a$ there to see that it suffices to satisfy
  $\beta < 1/a = (1-\bar b)/(w/r+\bar b)$ in order to make sure
  $F_{r,\beta}\geq 0$. The upper bound can be
  estimated as
  \begin{equation}
    1/a=\frac{1-\bar b}{w/r+\bar b} = r\frac{1-\bar b}{w+\bar b r} 
    > r\frac{1-\bar b}{w+\bar b/2}
    \label{eq:thm:gainbnd:1acalc}
  \end{equation}
  given our assumption that $r<1/2$ and $\bar b>0$.  With foresight,
  we set $d=(1-\bar b)/(2w+\bar b)<1/(2ra)$.  With $\beta\leq dr$,
  this implies $F_{r,\beta}>0$ and $\beta a <1/2$.

  Next we lower bound
  $\Exp_{\tilde\rho_{C|ZT}\rtimes\nu_{r}}(\log(F_{r,\beta}))/\beta$ by
  the same strategy that we used in the proof of
  Thm.~\ref{thm:agrfree}.  The result is
  \begin{equation}
    g_{r,\beta}=\Exp_{\tilde\rho_{C|ZT}\rtimes \nu_{r}}(\log(F_{r,\beta}))/\beta
    \geq -\log(1-\bar b) - \frac{\beta v_{r}}{2(1-\beta a)^{2}},
    \label{eq:thm:gainbnd1}
  \end{equation}
  where here we define
  $v_{r}=\Exp_{\tilde\rho_{C|ZT}\rtimes\nu_{r}}X_{r}^{2}$ with
  $X_{r}=(B_{r}-\bar b)/(1-\bar b)$. Note that from
  Eq.~\ref{eq:thm:frbeta_gain1} and the definition of $\bar b$,
  $\Exp_{\tilde\rho_{C|ZT}\rtimes\nu_{r}}X_{r}=0$. An explicit expression 
  for $v_r$ is obtained as follows:
  \begin{align}
    v_{r} &= \Exp_{\tilde\rho_{C|ZT}\rtimes\nu_{r}}\left(X_{r}(CZT)^{2}\right) \notag\\
    &= \sum_{zt}\Exp_{\tilde\rho_{C|zt}}\left(X_{r}(Czt)^{2}\right)\nu_{r}(zt) \notag\\
    &= \sum_{zt}\Exp_{\rho_{C|z}}\left(X_{r}(Czt)^{2}\right)\nu_{r}(zt) \notag\\
    &= \sum_{z}\Exp_{\rho_{C|z}}\left(X_{r}(Cz0)^{2}\right)\nu_{r}(z0) 
    +\sum_{z}\Exp_{\rho_{C|z}}\left(X_{r}(Cz1)^{2}\right)\nu_{r}(z1)  \notag\\
    &= \Exp_{\rho_{C|z_{0}}}\left(X_{r}(Cz_{0}0)^{2}\right)(1-r)
    +\sum_{z}\Exp_{\rho_{C|z}}\left(X_{r}(Cz1)^{2}\right)qr \notag\\ 
    &= \frac{\bar b^{2}}{(1-\bar b)^{2}}(1-r)
    +\Exp_{\rho_{C|Z}\rtimes \Unif_{Z}}\left(X_{r}(CZ1)^{2}\right)r.\label{eq:thm:frbeta_gain2}
  \end{align}
  For the second term in the last line, compute
  \begin{align}
    \Exp_{\rho_{C|Z}\rtimes \Unif_{Z}}\left(X_{r}(CZ1)^{2}\right)
    &=\frac{1}{(1-\bar
      b)^{2}}\Exp_{\rho_{C|Z}\rtimes\Unif_{Z}}\left((B_{r}(CZ1)-\bar
      b)^{2}\right)
    \notag\\
    &= \frac{1}{(1-\bar b)^{2}} \Exp_{\rho_{C|Z}\rtimes\Unif_{Z}}
    \left((B(CZ)/r -\bar b)^{2}\right)
    \notag\\
    &= \frac{1}{(1-\bar b)^{2}}\left(
      \Exp_{\rho_{C|Z}\rtimes\Unif_{Z}}\left((B(CZ)/r-\bar b/r)^{2}\right) +
      (\bar b/r-\bar b)^{2}\right)\notag\\
    &= \frac{1}{(1-\bar b)^{2}}\left(\bar v/r^{2}+\bar b^{2}(1/r-1)^{2}\right)\notag\\
    &= \frac{\bar v+\bar b^{2}(1-r)^{2}}{r^{2}(1-\bar b)^{2}},
  \end{align}
  where for the third line we applied the identity
  $\Exp(U^{2})=\Exp((U-\Exp(U))^{2})+\Exp(U)^{2}$ relating the
  expectation of the square to the variance.  Combining this with
  Eq.~\ref{eq:thm:frbeta_gain2} gives
  \begin{equation}
    v_{r}r = \frac{\bar v+((1-r)^{2}+r(1-r))\bar b^{2}}{(1-\bar b)^{2}}
    = \frac{\bar v + (1-r)\bar b^{2}}{(1-\bar b)^{2}} 
    \leq \frac{\bar v +  \bar b^{2}}{(1-\bar b)^{2}}.
  \end{equation}
  Write $c$ for the right-hand side of this inequality, which is
  independent of $r$.  Substituting in Eq.~\ref{eq:thm:gainbnd1}, we
  get
  \begin{align}
    \Exp_{\tilde{\rho}_{C|ZT}\rtimes\nu_{r}}(\log(F_{r,\beta}))/\beta 
    &\geq  -\log(1-\bar b)-\frac{\beta c/r}{2(1-\beta a)^{2}}\notag\\
    &\geq -\log(1-\bar b)-2\beta c/r,
  \end{align}
  because our earlier choice for $d$ implies $\beta a<1/2$.  We now
  set $d'=2c$ to complete the proof of the theorem.  
\end{proof}   

From the proof above, we can extract the variance of $F_{r,\beta}$ at
$\tilde\rho_{C|ZT}\rtimes\nu_{r}$, from which we can obtain a bound on
the variance of $\log(F_{r,\beta})$. Since we refer to it below when
discussing the statistical performance of exponential randomness
expansion, we record the result here.

\begin{lemma}\label{lem:var_frb}
  The variance $v_{r,\beta}$ of $\log(F_{r,\beta})$ with $\beta\leq
  dr$ at $\tilde\rho_{C|ZT}\rtimes\nu_{r}$ satisfies
  \begin{equation}
    v_{r,\beta}\leq 4\log(2)^{2}\frac{\beta^{2}}{r}\frac{\bar v+\bar b^{2}}{(1-\bar b)^{2}}.
  \end{equation}
\end{lemma}

\begin{proof}
  The variance of $F_{r,\beta}$ at $\tilde\rho_{C|ZT}\rtimes\nu_{r}$
  can be determined from $v_{r}$ in the proof of
  Thm.~\ref{thm:frbeta_gain}. It satisfies
  \begin{align}
    \Var(F_{r,\beta})&=(1-\bar b)^{-2\beta}\beta^{2}v_{r}\notag\\
    &\leq \frac{\beta^{2}}{r}\frac{\bar v+\bar b^{2}}{(1-\bar b)^{2(1+\beta)}}.
  \end{align}
  By construction of $F_{r,\beta}$, the mean of $F_{r,\beta}$ is
  $(1-\bar b)^{-\beta}$, and $F_{r,\beta}> (1-\bar b)^{-\beta}/2$.
  The latter follows from $\beta a<1/2$, noted where we defined
  $d$ after Eq.~\ref{eq:thm:gainbnd:1acalc}.  In
  general, for a positive RV $U$ with mean $\bar u$ and $U>\bar u/2$,
  \begin{align}
    \Var(\log(U)) &\leq \Exp\left((\log(U)-\log(\bar u))^{2}\right)\notag\\
    &=\Exp\left((\log(U/\bar u))^{2}\right)\notag\\
    &=\Exp\left((\log(1+(U-\bar u)/\bar u))^{2}\right)\notag\\
    &\leq 4\log(2)^{2}
    \Exp\left( ((U-\bar u)/\bar u)^{2}\right)\notag\\
    &= 4\log(2)^{2}\Var(U)/(\bar u)^{2}.
  \end{align}
  Here, the first step follows from $\Var(Y)\leq
  \Exp((Y-y_{0})^{2})$ for any $y_{0}$, where we substituted
  $\log(U)$ and $\log(\bar u)$ for $Y$ and $y_{0}$. In the
  second-last step, we used $|\log(1+y)|\leq 2\log(2)|y|$ for
  $y\geq -1/2$, substituting $(U-\bar u)/\bar u$ for $y$.
  Setting $U=F_{r,\beta}$ and $\bar u=\Exp(F_{r,\beta})=(1-\bar b)^{-\beta}$, 
  we get the inequality in the lemma.
\end{proof}

\subsection{Exponential Randomness Expansion}

For PEFs satisfying the conditions for Thm.~\ref{thm:frbeta_gain},
it is possible to achieve exponential expansion of the entropy of $Z$.

\begin{theorem}\label{thm:expexp}
  Let $F_{r,\beta}$ be the family of PEFs from
  Thm.~\ref{thm:frbeta_gain}.  Consider $n$ trials with model
  $\cH(\cC_{C|ZT}\rtimes\nu_{r})$ and a fixed error bound $\epse
  $. Consider $\rho_{C|Z}\in\cC_{C|Z}$ and let
  $g_{r,\beta}$ be the log-prob rate of $F_{r,\beta}$ at
  $\tilde\rho_{C|ZT}\rtimes\nu_{r}$, and assume that the
  trials are i.i.d.~for $CZT$ with this distribution.  We can choose
  $r=r_{n}$ and $\beta=\beta_{n}$ as functions of $n$ so that the
  expected net log-prob is $g_{\mathrm{net}}= n
  g_{r_{n},\beta_{n}}-\log(1/\epse )/\beta_{n} \geq n(-\log(1-\bar
    b))/3 = e^{\Omega(nS(\nu_{r_{n}}))}$, where $nS(\nu_{r_{n}})$ is
  the total input entropy of $\Sfnt{Z}\Sfnt{T}$.
\end{theorem}

The performance in the theorem statement may be compared to that
in Ref.~\cite{miller_c:qc2014b}, where the settings entropy for $n$
trials is $\Omega(\log^{3}(n))$ (see the end of Sect. 1.1 of the
reference).  With our terminology, this implies a net log-prob
$e^{O(S^{1/3})}$ where $S$ is the total input entropy, which is
subexponential in input entropy. It is also possible to match the
soundness performance of the ``one-shot'' protocol of
Ref.~\cite{miller_c:qc2014a}, Cor.~1.2, by modifying the proof below
with $r_{n}$ and $\beta_{n}$ proportional to $n^{-\omega}$ with
$0<\omega< 1$ instead of $\omega=1$. However, the referenced
results are valid for quantum side information.

The proof of Thm.~\ref{thm:expexp} given below computes explicit
constants for the orders given in the theorem statement. The
constants depend on the PEF construction in the previous section,
and the performance can always be improved by direct PEF
optimization. As a result we do not expect the constants to be
needed in practice.

The i.i.d.~assumption is a completeness requirement that we expect
to hold approximately if the devices perform as designed.  Explicit
constants expressed in terms of those in Thm.~\ref{thm:frbeta_gain}
are given with Eq.~\ref{eq:thm:expexp_explicit} below.  Notably, the
expected number of test trials is independent of $n$ and proportional
to $\log(1/\epse)$, which is what one would hope for in the absence of
special strategies allowing feed-forward of generated randomness from
independent devices, an example of which is the cross-feeding protocol
of Ref.~\cite{chung:2014}. However, since non-test trials contribute
negatively to the log-prob, this means that the full gain and the
error bound are witnessed by a constant number of test trials.  One
should therefore expect a large instance-to-instance variation in the
actual log-prob obtained, raising the question of whether the success
probability, or the expected number of trials before reaching a
log-prob goal, are sufficiently well behaved. We consider this
question after the proof.

\begin{proof}
  Let $d$ and $d'$ be the constants in Thm.~\ref{thm:frbeta_gain}.
  From this theorem, the expected net log-prob is bounded by
  \begin{align}
    g_{\mathrm{net}} &= n g_{r_{n},\beta_{n}}-\frac{\log(1/\epse)}{\beta_{n}}\notag\\
    &\geq n\left(-\log(1-\bar b)-\frac{d'\beta_{n}}{r_{n}}\right)-\frac{\log(1/\epse)}{\beta_{n}},
    \label{eq:thm:expexp_gtotal}
  \end{align}
  with $\bar b$ as defined for Thm.~\ref{thm:frbeta_gain}.  The input
  entropy per trial $S(\nu_{r})=H(r)+r\log(1/q)$ is bounded above by
  $-2r\log(r)$, provided we take $r\leq q/e$.  For this, note that
  $r\log(1/q)\leq r\log(1/(re))=-r\log(r)-r$ and $-(1-r)\log(1-r)\leq
  r$ since the function $x\mapsto -(1-x)\log(1-x)$ is concave for
  $0\leq x<1$, is equal to $0$ at $x=1$, and has slope $1$ at $x=0$.  If we choose
  $r_{n}=c'/n$ with $c'$ to be determined later, the expected number
  of test trials is $c'$ and the total input entropy satisfies
  $nS(\nu_{r_{n}}) \leq 2c'\log(n/c')$, or equivalently $n\geq
  c'e^{nS(\nu_{r_{n}})/(2c')}$.  For sufficiently large $n$,
  $r\leq q/e$ is satisfied.  If we then choose $\beta_{n} = c r_{n}= c
  c'/n$ with $c \leq d$ such that $g_{\mathrm{net}}$ grows linearly
  with $n$, we have accomplished our goal.  Write $g_{0}=-\log(1-\bar
  b)$.  Substituting the expressions for $\beta_{n}$ and $r_{n}$ in
  Eq.~\ref{eq:thm:expexp_gtotal},
  \begin{align}
    g_{\mathrm{net}}
    &\geq n g_{0}\left( 1-\frac{d' c}{g_{0}} - \frac{\log(1/\epse)}{cc' g_{0}}\right).
    \label{eq:thm:expexp_lowerbnd}
  \end{align}
  We first set $c=\min(d,g_{0}/(3d'))$, which ensures that
  $d'c/g_{0}\leq 1/3$. We then set $c'=3\log(1/\epse)/(cg_{0})$.  This
  gives the inequality
  \begin{equation}
    g_{\mathrm{net}} \geq n g_{0}/3 \geq (c'g_{0}/3)e^{nS(\nu_{r_{n}})/(2c')},
    \label{eq:thm:expexp_explicit}
  \end{equation}
  which implies the theorem.
\end{proof}

According to Lem.~\ref{lem:var_frb}, the variance at $\tilde\rho_{C|ZT}\rtimes\nu_{r}$  of the
quantity $n\log(F_{r_{n},\beta_{n}})/\beta$ whose expectation is the
expected log-prob for the parameters chosen in the proof above is given by
\begin{align}
  v=n\frac{v_{r_{n},\beta_{n}}}{\beta_{n}^{2}}
  &\leq 4\log(2)^{2}\frac{n}{r_{n}}\frac{\bar v+\bar b^{2}}{(1-\bar b)^{2}}\notag\\
  &= 4\log(2)^{2}\frac{n^{2}}{c'}\frac{\bar v+\bar b^{2}}{(1-\bar b)^{2}}\notag\\
  &=2\log(2)^{2}\frac{n^{2}d'}{c'}.\label{eq:log-prob_variance}
\end{align}
Here, to get the last line we used Eq.~\ref{eq:variance_constants}.
In view of Eq.~\ref{eq:thm:expexp_lowerbnd}, the expected log-prob is
at least $2ng_{0}/3$, with up to $1/2$ of that taken away by the
error-bound requirement. Here we are interested in the distribution of
the net log-prob, which depends on the actual trial results.  The
amount that is subtracted from the log-prob to get the net log-prob
for the error bound is independent of the trial results and given by
$f=n\log(1/\epse)/(cc')$.  By increasing $c'$ by a constant factor
while holding $c$ fixed, we can reduce $f$ and still achieve
exponential expansion, albeit with a smaller constant in the exponent.
The first two terms on the right-hand side of
Eq.~\ref{eq:thm:expexp_lowerbnd} are independent of $c'$ and provide a
lower bound on the expected log-prob.  These considerations imply that
if the variance is small enough compared to $(ng_{0}/6)^{2}$ (for
example), with correspondingly high probability the net log-prob
exceeds $ng_{0}/6$.  Small variance is implied by large $c'$, because
according to Eq.~\ref{eq:log-prob_variance}, the variance of the
log-prob is reduced if $c'$ is increased.  To be specific, let
$\kappa\geq\log(1/\epse)$. We can set
\begin{align}
  c' &= \frac{3\kappa}{c g_{0}}\notag\\
  &=\frac{3\kappa}{g_{0}\min(d,g_{0}/(3d'))}\notag\\
  &=\frac{ 9 d'\kappa}{g_{0}^{2}\min(3dd'/g_{0},1)},
\end{align}
from which
\begin{equation}
  v\leq 2\log(2)^{2}\frac{n^{2}g_{0}^{2}\min(3dd'/g_{0},1)}{9\kappa}
  \leq 2\log(2)^{2}\frac{n^{2}g_{0}^{2}}{9\kappa}.\label{eq:vkappa_bound}
\end{equation}
By Chebyshev's inequality, a conservative upper bound on the
probability that
$n\log(F_{r_{n},\beta_{n}})/\beta<2ng_{0}/3-\sqrt{v}/\lambda$ is
$\lambda^{2}$. Since at most $ng_{0}/3$ is subtracted to obtain the
net log-prob, if we require that $\sqrt{v}/\lambda\leq ng_{0}/6$, then
with probability at least $1-\lambda^{2}$ we get a net log-prob of at
least $ng_{0}/6$, which is sufficient for exponential expansion.
From Eq.~\ref{eq:vkappa_bound}, we can achieve this if we
have
\begin{equation}
  2\log(2)^{2}\frac{n^{2}g_{0}^{2}}{9\kappa\lambda^{2}}\leq \frac{n^{2}g_{0}^{2}}{6^{2}},
\end{equation}
that is, if $\kappa\geq 8\log(2)^{2}/\lambda^{2}$.  In terms of $c'$, this
means that $c'$ needs to be set to
$\max(3\log(1/\epse),24\log(2)^{2}/\lambda^{2})/(cg_0)$ for a
probability of at least $1-\lambda^{2}$ of having a net log-prob of $n
g_{0}/6$.

Better bounds on the probabilities above can be obtained by taking
into consideration the i.i.d.~assumptions. For example one can apply
Chernoff-Hoeffding bounds to improve the estimates.  Alternatively, we
can use the option to acquire trials until a desired net log-prob is
obtained with a conservative upper bound on the maximum number of
trials.

\subsection{Block-wise Spot-Checking}

For randomness expansion without having to decompress uniform bits for
generating biased distributions of $Z$, we suggest the following
spot-checking strategy consisting of blocks of trials.  Each block
nominally contains $2^{k}$ trials. One of them is chosen uniformly at
random as a test trial, for which we use the random variable $T_{s}$,
now valued between $0$ and $2^{k}-1$ with a uniform distribution.  The
$l$'th trial in a block uses $Z=z_{0}$ (distribution
$\delta_{z_{0}}$), unless $l=T_{s}+1$ in which case the distribution
of $Z$ is $\Unif_{Z}$. For analyzing these trials, observe that the
distribution of $Z$ at the $l$'th trial conditional on not having
encountered the test trial yet, that is, given that $\{l\leq T_{s}\}$,
is a mixture of $\delta_{z_{0}}$ and $\Unif_{Z}$, with the probability
of the latter being $1/(2^{k}-l+1)$. Thus we can process the trial
using an adaptive strategy, designing the PEF for the $l$'th trial
according to Thm.~\ref{thm:frbeta_gain}, where the distribution of $T$
has probability $r_{l}=1/(2^{k}-l+1)$ of $T=1$.  We call this the
\emph{ block-wise spot-checking} strategy.  Here we must choose a
single power for the PEFs independent of $l$.  Once the test trial is
encountered, the remaining trials in the block do not need to be
performed as we can, equivalently, set the PEFs to be trivial from
this point on. Of course, if there is verifiable randomness even when
the setting is fixed and known by the external entity and devices, we
can continue.  The average number of trials per block is therefore
$(2^k+1)/2$. Note that $r_{2^{k}-l+1}=1/l$ and the probability of
reaching the $2^{k}-l+1$'th trial is $l/2^{k}$. When we choose the
common power of the PEFs, the goal is to optimize the expected
log-prob per block in consideration of the way in which the settings
probabilities change as well as the anticipated trial statistics.

\section{Applications to Bell Tests}
\label{sec:apps}

\subsection{General Considerations}

All our applications are for the $(2,2,2)$ Bell-test configuration.
In this configuration, $Z=XY$ and $D=C=AB$ with $X,Y,A,B\in \{0,1\}$.
The distributions are free for $Z$. The set of conditional
distributions consists of those satisfying the non-signaling
constraints Eq.~\ref{eq:non-signaling_constraints}. We denote this set
by $\cN_{C|Z}$.  This includes distributions not achievable by
two-party quantum states, so additional quantum constraints can be
considered. There is a hierarchy of semidefinite programs to generate
such constraints due to Navascues, Pironio and
Ac\'{\i}n~\cite{navascues:2007} (the NPA hierarchy). Because
semidefinite programs define convex sets that are not polytopes and
have a continuum of extreme points, we cannot make direct use of the
NPA hierarchy.  In the absence of a semidefinite hierarchy for the PEF
constraints, we require the conditional distributions to be a polytope
defined by extreme points. In general, we can use the NPA hierarchy to
generate finitely many additional linear constraints to add to the
non-signaling constraints. The set of extreme points can then be
obtained by available linear programming tools and used for PEF
optimization. To keep the number of extreme points low, it helps to
add only constraints relevant to the actual distribution expected at a
trial. We do not have a systematic way to identify relevant
constraints.  Below, we consider additional quantum constraints based
on Tsirelson's bound~\cite{Tsirelson:1980}.  We note that one can
consider constructing PEFs via Thm.~\ref{thm:pef_uniformbnd} as
discussed after the theorem.  With this method one can take advantage
of the NPA hierarchy directly but is restricted to PEFs satisfying
Eq.~\ref{eq:probest_constr_max}.

The extreme points of $\cN_{C|Z}$ are given by the
deterministic LR distributions and variations of
Popescu-Rohrlich (PR) boxes~\cite{Popescu:1994}. The deterministic
LR distributions are parameterized by two functions $f_{A}:x\mapsto
f_{A}(x)\in\{0,1\}$ and $f_{B}:y\mapsto f_{B}(y)\in\{0,1\}$ defining
each station's settings-dependent outcome, giving distributions
\begin{equation}
  \nu_{f_{A},f_{B}}(ab|xy) = \knuth{f_{A}(x)=a}\knuth{f_{B}(y)=b}.
  \label{eq:lrf}
\end{equation}
This yields the $16$ LR extreme points. The PR boxes have
distributions $\nu_{g}$ defined by functions $g:xy\mapsto\{0,1\}$,
where $|\{xy|g(xy)=1\}|$ is odd, according to
\begin{equation}
  \nu_{g}(ab|xy) = \knuth{a=b\oplus g(xy)}/2,
  \label{eq:prf}
\end{equation}
with addition of values in $\{0,1\}$ modulo $2$. The function $g$
indicates for which settings pairs the stations' outcomes are perfectly
anti-correlated. There are $8$ extreme points in this class.

The simplest quantum constraint to add is Tsirelson's bound~\cite{Tsirelson:1980}, which can
be expressed in the form $\Exp(B(CZ))\leq (\sqrt{2}-1)/4$ for all
distributions compatible with quantum mechanics, with $B(CZ)$ given by
Eq.~\ref{eq:chshvariant} and fixing the settings distribution to be
uniform.  There are $8$ variants of Tsirelson's bound expressed as
\begin{equation}
  \Exp(B_{g}(CZ))\leq (\sqrt{2}-1)/4
  \label{eq:8tsir}
\end{equation}
corresponding to the $8$ PR boxes labelled by $g$, with $B_{g}(CZ)$ as
defined in the next paragraph.  Let $\cQ_{C|Z}$ be the set of
conditional distributions obtained by adding the eight conditional
forms of
Eq.~\ref{eq:8tsir} to the constraints.

\Pc{Justification of the version of Tsirelson's
  bound used here.

  Tsirelson's bound is usually expressed as
  \begin{equation}
    \bar c=\langle A_{0}B_{0}\rangle +\langle A_{0}B_{1}\rangle +\langle A_{1}B_{0}\rangle -\langle A_{1}B_{1}\rangle \leq 2{\sqrt {2}},\notag
  \end{equation}
  where $\langle A_{x}B_{y}\rangle = \Exp(4(A-1/2)(B-1/2)|xy)$ are the
  conventional correlations favored in traditional treatments of the
  CHSH inequality.  From Eq.~\ref{eq:chshvariant}, we have
  $\Exp(B(CZ))=\sum_{xy}\Exp(B(CZ)|xy)=
  -\sum_{xy}(-1)^{x\cdot y}\Exp(|A-B||xy)$ and
  \begin{equation}
    \Exp(|A-B||xy)\mu(xy)=\Exp\left(\vphantom{\big|}(1-4(A-1/2)(B-1/2))/2|xy\right)/4
    = (1-\langle A_{x}B_{y}\rangle)/8.\notag
  \end{equation}
  Summing the terms and identifying the resulting sum of the $\langle
  A_{x}B_{y}\rangle$ as $\bar c/8$, we get $\Exp(B(CZ))
  =(\bar c-2)/8\leq (\sqrt{2}-1)/4$ }

For any PR box $\nu_{g}$, let $p(g)=0$ or $p(g)=1$ according
to whether $|g^{-1}(1)|=1$ or $|g^{-1}(1)|=3$ and define
\begin{equation}
  B_{g}(xyab) =  -(-1)^{g(xy)+p(g)}\knuth{a\not=b\oplus p(g)}.
\end{equation}
This implies that the Bell function $B(CZ)$ in Eq.~\ref{eq:chshvariant}
can be expressed as $B(CZ)=B_{g:xy\mapsto x\times y}(CZ)$,
with $x\times y$ the numeric product of $x$ and $y$.  For any
distribution $\rho\in\cN_{C|Z}\rtimes\{\Unif_{Z}\}$, the expectation
of $\Exp_{\rho}(B_{g}(CZ))$ bounds the contribution of the PR box $\nu_{g}$
to $\rho$ in the sense that if $\rho = (1-p)\rho_{\LR}+p\nu_{g}$, with
$\rho_{\LR}$ being an LR distribution, then
$\Exp_{\rho}(B_{g}(CZ))\leq p/4$.  From Ref.~\cite{bierhorst:2016},
one determines that the bound is tight in the sense that if
$\Exp_{\rho}(B_{g}(CZ))=p/4$, then there exists an LR distribution
$\rho_{\LR}$ such that $\rho=(1-p)\rho_{\LR}+p\nu_{g}$. These LR
distributions are convex combinations of the eight deterministic LR
distributions for which $\Exp_{\nu_{f_{A},f_{B}}}(B_{g}(CZ))=0$.  We
observe that these are extremal, and all non-signaling distributions
can be expressed as convex combinations of extreme points involving at
most one PR box.  This makes it possible to explicitly determine the
extreme points of $\cQ_{C|Z}$.  They are the convex combinations
$(1-q)\nu_{f_{A},f_{B}}+q\nu_{g}$ with $q=\sqrt{2}-1$ that satisfy
Tsirelson's bound with equality on the one-dimensional line connecting
one of the above $8$ deterministic LR distributions to the corresponding
PR box.

We emphasize that probability estimation is not based on the
Bell-inequality framework. However, any PEF $F$ for a Bell-test
configuration has the property that $F-1$ is a Bell function and
therefore associated with a Bell inequality. This follows because the
conditional probabilities for the extremal LR distributions are either
$0$ or $1$. Therefore, regardless of the power $\beta$, the PEF
constraints from these extreme points are equivalent to the
constraints for Bell functions after subtracting $1$.  One can
construct PEFs from Bell functions by reversing this argument and
rescaling the resulting factor to satisfy the constraints induced by
non-LR extreme points. But the relationship between measures of
quality of Bell functions (such as violation signal-to-noise, winning
probability or statistical strength for rejecting LR) and those of
related PEFs is not clear.

\subsection{Applications and Methods}

We illustrate probability estimation with three examples. In the
first, we estimate the conditional probability distribution of the
outcomes observed in the recent loophole-free Bell test with entangled
atoms reported in Ref.~\cite{rosenfeld:qc2016a}. In the second, we
reanalyze the data used for extracting $256$ random bits uniform
within $0.001$ in Ref.~\cite{bierhorst:qc2017a}. Finally, we reanalyze
the data from the ion experiment of Ref.~\cite{pironio:2010}, which
was the first to demonstrate min-entropy estimation from a Bell test.
For the first example, we explore the log-prob rates and net
log-prob rates over a range of choices of parameters for PEF power,
error bound, unknown settings bias and test-trial probability. The
log-prob rates and net log-prob rates are determined directly
from the inferred probabilities of the outcomes, from which the expected net
log-prob can be inferred under the assumption that the trials are
i.i.d.  For the other two examples, we explicitly apply the
probability estimation schema and obtain the total log-prob for
different scenarios. In each scenario, the parameter choices are
made based on training data, following the protocol in
Ref.~\cite{bierhorst:qc2017a}. We show the certifiable min-entropy for
the reported data as a function of the error bound and compare to the
min-entropies reported in Refs.~\cite{bierhorst:qc2017a,pironio:2010}.

For determining the best PEFs, we can optimize the log-prob rate at
the estimated distribution $\nu$ given the power $\beta$. No matter
which of $\cN_{C|Z}$ or $\cQ_{C|Z}$ is considered, the free-for-$Z$
sets $\cC$ obtained accordingly have a finite number of extreme
points. We can therefore take advantage of the results in
Sect.~\ref{sec:first_pef_construction} to solve the optimization
problem effectively. Given the list of extreme points
$(\rho_{i})_{i=1}^{q}$ of the set $\cC$ and the estimated distribution
$\nu$ in $\cC$, the optimization problem can be stated as follows:
\begin{equation}
  \begin{array}[b]{lll}
    \textrm{Maximize:}& \sum_{cz}\log(F(cz))\nu(cz)&\\
    \textrm{Subject to:}& F(cz)\geq 0, & \forall cz \\
    & \sum_{cz}F(cz)\rho_{i}(c|z)^{\beta}\rho_{i}(cz) \leq 1, & \forall i.
  \end{array}\label{eq:opt_pef}
\end{equation}
Since the objective function is concave with respect to $F$ and
the constraints are linear, the problem can be solved by any algorithm
capable of optimizing non-linear functions with linear constraints on
the arguments. In our implementation, we use sequential quadratic
programming. Due to numerical imprecision, it is possible that the
returned numerical solution does not satisfy the second constraint in
Eq.~\ref{eq:opt_pef} and the corresponding PEF is not valid. In this
case, we can multiply the returned numerical solution by a
positive factor smaller than 1, whose value is given by the reciprocal
of the maximal left-hand side of the above second constraint at the
extreme points of $\cC$.  Then, the re-scaled solution is a valid PEF.

All examples involve inferring an experimental probability
distribution, in the first example from the published data, and in the
other two from an initial set of trials referred to as the \emph{training
set}. Due to finite statistics, the data's empirical frequencies are
unlikely to satisfy the constraints used for probability estimation.
It is necessary to ensure that the distribution for which log-prob
rates and corresponding PEFs are calculated satisfies the constraints,
because otherwise, the distribution is not consistent with the
physical model underlying the randomness generation process, and so
unrealistically high predicted log-prob rates may be obtained.  To
ensure self-consistency, before solving the optimization in
Eq.~\ref{eq:opt_pef}, we determine the closest constraints-satisfying
conditional distributions.  For this, we take as input the empirical
frequencies $f(cz)$ and find the most likely distribution $\nu$
of $CZ$ with $(\nu[C|z])_{z}$ in $\cQ_{C|Z}$ and $\nu(z)=f(z)$ for
each $z$. That is, we solve the following optimization problem:
\begin{equation}
  \begin{array}[b]{lll}
    \textrm{Maximize:} & \sum_{cz} f(cz)\log(\nu(c|z)/f(c|z)) & \\
    \textrm{Subject to:} & (\nu[C|z])_{z}\in \cQ_{C|Z}. & 
  \end{array}
  \label{eq:maximum_likelihood}
\end{equation}
The objective function is proportional to the logarithm of the ratio
of the likelihood of the observed results with respect to the
constrained distribution $\nu$ and that with respect to the
empirical frequencies $f$.  \Pc{The
  unadjusted likelihood ratio is given by:
  \begin{equation*}
    \sum_{cz} N_{cz}\log(\nu(cz)/f(cz))
    = \sum_{cz}N f(cz)\log(\nu(cz)/f(cz))
    = \sum_{cz}N f(cz)\log(\nu(c|z)/f(c|z)).
  \end{equation*}
  Here, $N$ is the number of observed trials, and $N_{cz}=f(cz)N$ is
  the number of trials with settings $z$ and outcomes $c$.  }  This
optimization problem involves maximizing a concave function over a
convex domain and can be solved by widely available tools.  Note that
in all cases, we ensure that the solution satisfies Tsirelson's
bounds.  Once $\nu(C|Z)$ has been determined, we use
$(\nu[C|z])_{z}\in\cC_{C|Z}$ for the predictions in the first example,
and for determining the best PEFs for the data to be analyzed in the
other two examples. In each case, we use the model for the
settings probability distribution relevant to the situation, not the
empirical one $f(z)=\nu(z)$.

An issue that comes up when using physical random number sources to
choose the settings is that these sources are almost but not quite
uniform. Thus, allowances for possible trial-by-trial biases in the
settings choices need to be made. We do this by modifying the
distribution constraints. Instead of having a fixed uniform settings
distribution, for any \emph{bias} $b$, we consider the set
\begin{equation}
  \cB_{V,b} = \{\nu : \textrm{$\nu$ is a distribution of $V$ and 
    for all $v$, $|\nu(v)-\Unif_{V}(v)|\leq b/|\Rng(V)|$} \}.
  \label{eq:cbvbdef}
\end{equation}  
For $\Rng(V)=\{0,1\}$, the extreme points of $\cB_{V,b}$ are
distributions $\nu$ of the form $\nu(v)=(1 \pm (-1)^{v}b)/2$.  For the
$(2,2,2)$ configuration, $Z=XY$ and we use
$\cC_{Z,b}=\cCvx(\cB_{X,b}\otimes\cB_{Y,b})$ for the settings
distribution. We then consider the free-for-$Z$ distributions given by
$\cN_{C|Z}\rtimes\cC_{Z,b}$ and $\cQ_{C|Z}\rtimes\cC_{Z,b}$.  Thus, we
allow any settings distribution that is in the convex span of
independent settings choices with limited bias. We do not consider the
polytope of arbitrarily biased joint settings distributions given by
$\cB_{Z,b}$ but note that $\cC_{Z,b}\subseteq \cB_{Z,2b+b^{2}}$ and
$\cB_{Z,b^{2}}\subseteq \cC_{Z,b}$. Another polytope that could be
used is $\cB'_{V,\epsilon}=\{\nu[V]:\nu(v)\geq\epsilon \textrm{ for
  all }v\}$.  We have $\cB_{Z,b}\subseteq\cB'_{Z,(1-b)/4}$.  It is
known that there are Bell inequalities that can be violated for
certain quantum states and settings-associated measurements when the
settings distribution is arbitrary in $\cB'_{Z,\epsilon}$ for all
$1/4\geq \epsilon>0$~\cite{puetz:qc2014b}.

\Pc{This is to show that $\cB_{Z,b^{2}}\subseteq\cC_{Z,b}$.

  The extreme points of $\cB_{Z,c}$ are the six probability
  distributions satisfying $\nu(xy)\in\{1/4\pm c/4\}$. The inversion
  maps $x\mapsto (1-x)$ and $y\mapsto (1-y)$ preserve the polytopes,
  as does the exchange map $xy\mapsto yx$.  So it suffices to consider
  the two extreme points with $\rho_{i}(xy)=1/4+c/4$ for
  $xy\in\{11,00\}$ (case $i=1$), or $xy\in\{11,10\}$ (case $i=2$). We
  show that they are in $\cC_{Z,b}$ for $c=b^{2}$.  The extreme points
  of $\cC_{Z,b}$ are given by
  $\nu_{uv}(xy)=(1+(-1)^{u+x}b)(1+(-1)^{v+y}b)/4$ with
  $u,v\in\{0,1\}$.  This has values $(1+2b+b^{2})/4$ for $xy=uv$,
  $(1-2b+b^{2})/4$ for $xy=(u-1)(v-1)$ and $(1-b^{2})/4$ otherwise.
  We can get $\rho_{1}$ at $c=b^{2}$ with $(\nu_{00}+\nu_{11})/2$.
  The extreme point $\rho_{2}$ at $c=b$ is obtained from
  $(\nu_{11}+\nu_{10})/2$.  }

\subsection{Exploration of Parameters for a Representative 
  Distribution of the Outcomes}

We switch to logarithms base $2$ for all log-probs and entropies in
the explorations to follow. Thus all explicitly given numerical values
for such quantities can be directly related to bits without
conversion. For clarity, we indicate the base as a subscript throughout.

For the first example, we use outcome table XII from the
supplementary material of Ref.~\cite{rosenfeld:qc2016a} to determine a
representative distribution $\rho_{\textrm{atoms}}$ for
state-of-the-art loophole-free Bell tests involving entangled atoms.
This experiment involved two neutral atoms in separate traps at a
distance of $398\,m$.  They were prepared in a heralded entangled
state before choosing uniformly random measurement settings and
determining the outcomes.  The experiment was loophole-free and showed
violation of LR at a significance level of $1.05\times
10^{-9}$. The outcome table shows outcome counts for two different
initial states labeled $\Psi^{+}$ and $\Psi^{-}$. We used the
counts for the state $\Psi^{-}$, which has better violation. We
determined the optimal conditional distribution satisfying $\cQ_{C|Z}$
as explained in the previous section, which gives the
settings-conditional distribution $\rho_{\textrm{atoms}}$
shown in Table~\ref{tab:rosenfeld_dist}.

\begin{table}
  \caption{Inferred settings-conditional distribution $\rho_{\textrm{atoms}}(ab|xy)$.
    Shown are the settings-conditional probabilities of the outcomes
    estimated from all trials with the prepared state $\Psi^{-}$ 
    reported in Ref.~\cite{rosenfeld:qc2016a}. This is 
    the $\cQ_{C|Z}$-satisfying distribution
    based on the outcomes from $27683$ heralded trials. 
    Each column is for one 
    combination of outcomes $ab$, and the rows correspond to the settings combination
    $xy$. Because the probabilities are settings-conditional,
    each row's probabilities add to $1$, except for rounding
    error. We give the numbers at numerical precision for traceability. They are
    used for determining PEFs, not to make a statement about the actual
    distribution of the outcomes in the experiment. Statistical error does not affect
    the validity of the PEFs obtained 
    (see the paragraph after Eq.~\ref{eq:opt_pef}).
  }
  \label{tab:rosenfeld_dist}
  \begin{equation}
    \begin{array}{|ll||l|l|l|l|}
      \hline
      &ab&00&10&01&11\\
      xy&&&&&\\
      \hline\hline
      00&&
      0.114583230563265&   0.408949785618886&   0.369310344143205&   0.107156639674644
      \\
      10&&
      0.399140705802719&   0.124392310379432&   0.111262723543957&  0.365204260273892
      \\
      01&&
      0.102208313465938&   0.403210760398438&   0.381685261240533&   0.112895664895092
      \\
      11&&
      0.127756153189431&   0.377662920674945&   0.382647276157245&  0.111933649978380
      \\
      \hline
    \end{array}\notag
  \end{equation}
\end{table}

We begin by exploring the relationship between power,
$\textrm{log}_2$-prob rate and net $\textrm{log}_2$-prob
rate. Fig.~\ref{fig:rosenfeld_lpg(beta)} shows the optimal
$\textrm{log}_2$-prob rates and net $\textrm{log}_2$-prob rates as a
function of power $\beta$ for a few choices of error bound $\epse$. As
expected, the $\textrm{log}_2$-prob rates decrease as the power $\beta$ increases. The
asymptotic gain rates can be inferred from the values approached as
$\beta$ goes to zero. The inferred asymptotic gain rates are
approximately $0.088$ and $0.191$ for $\cN_{C|Z}$ and $\cQ_{C|Z}$,
respectively.  For small $\beta$, constraints $\cQ_{C|Z}$ improve the
$\textrm{log}_2$-prob rates by around a factor of $2$ over
$\cN_{C|Z}$. The net $\textrm{log}_2$-prob rates show the trade-off
between error bound and $\textrm{log}_2$-prob rate, with clear maxima
at $\beta>0$.

\begin{figure}
  \begin{center}
    \includegraphics[scale=0.8,viewport=5.5cm 8cm 17cm 20.5cm]{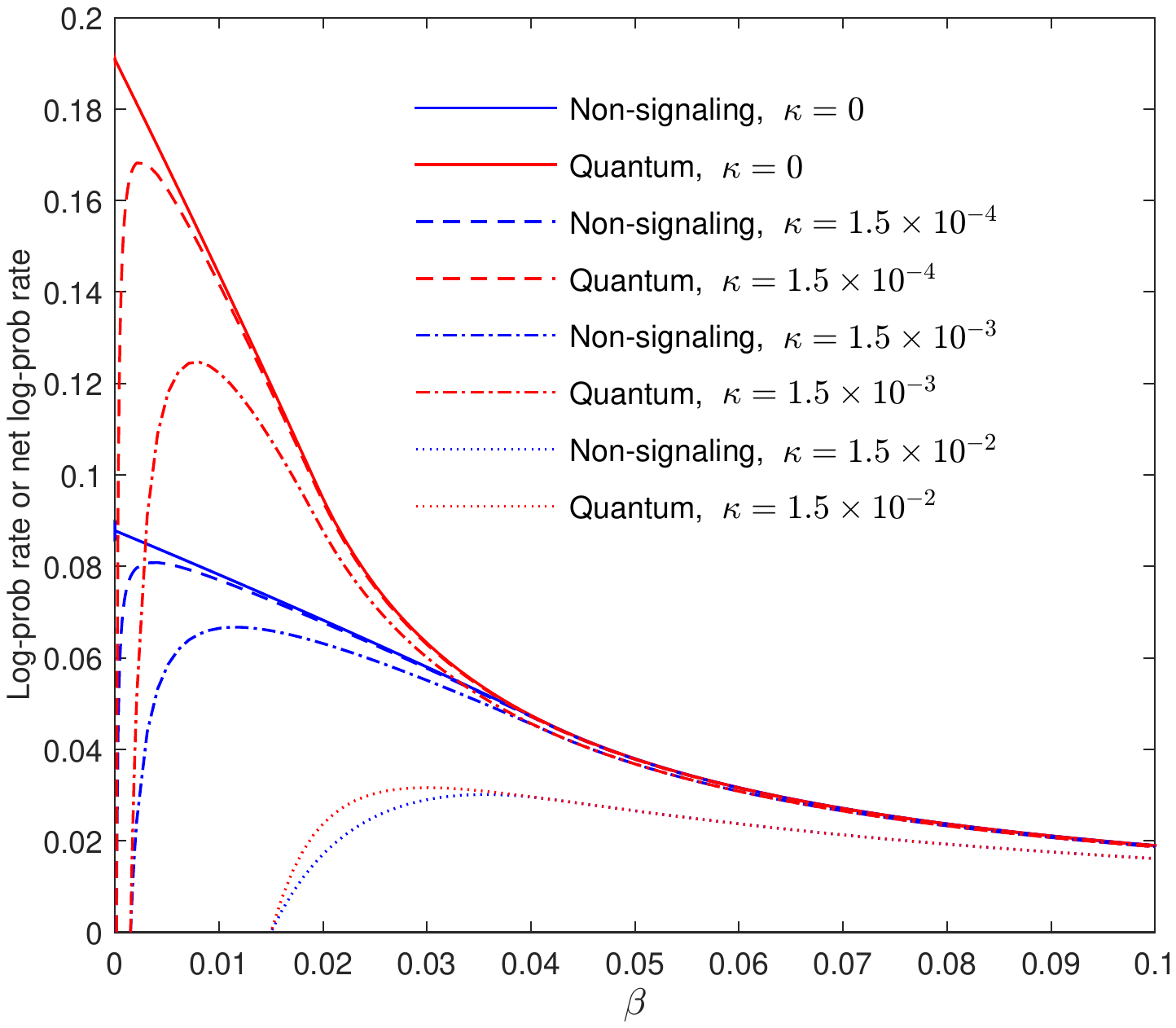}
  \end{center}
  \caption{$\textrm{Log}_2$-prob rates (labelled with $\kappa=0$) at 
    $\rho_{\textrm{atoms}}$ for different
    powers $\beta$. The rates are shown for non-signaling $\cN_{C|Z}$
    and quantum $\cQ_{C|Z}$ (Tsirelson's bounds only) constraints.
    Also shown are the net $\textrm{log}_2$-prob rates according
    to Eq.~\ref{eq:netlogprobdef} for error bounds
    $\epse =2^{-\kappa\sigma n}$, where $\sigma$ is the $\textrm{log}_2$-prob
    rate and for three positive values of $\kappa$. 
    Note that the constant $\kappa$ can
    be interpreted as a $\textrm{log}_2$-error-bound rate with respect to the 
    nominally available entropy for randomness generation.}
  \label{fig:rosenfeld_lpg(beta)}
\end{figure}

Another view of the trade-off is obtained from the expected net
$\textrm{log}_2$-prob obtained after $27683$ trials at different final
error bounds, where we optimize the expected $\textrm{log}_2$-prob
with respect to power $\beta$.  This is shown in
Fig.~\ref{fig:rosenfeld_lp(eps)}. The expected net $\textrm{log}_{2}$-prob is
essentially the amount of min-entropy that we expect to certify by
probability estimation in measurement outcomes from the $27683$ trials
in the experiment of Ref.~\cite{rosenfeld:qc2016a}. The number of
near-uniform bits that could have been extracted from this
experiment's data is substantially larger than that shown in
Fig.~\ref{fig:Maryland_logprob} for another Bell test using two ions.

\begin{figure}
  \begin{center}
    \includegraphics[scale=0.8,viewport=5.5cm 8cm 17cm 20.5cm]{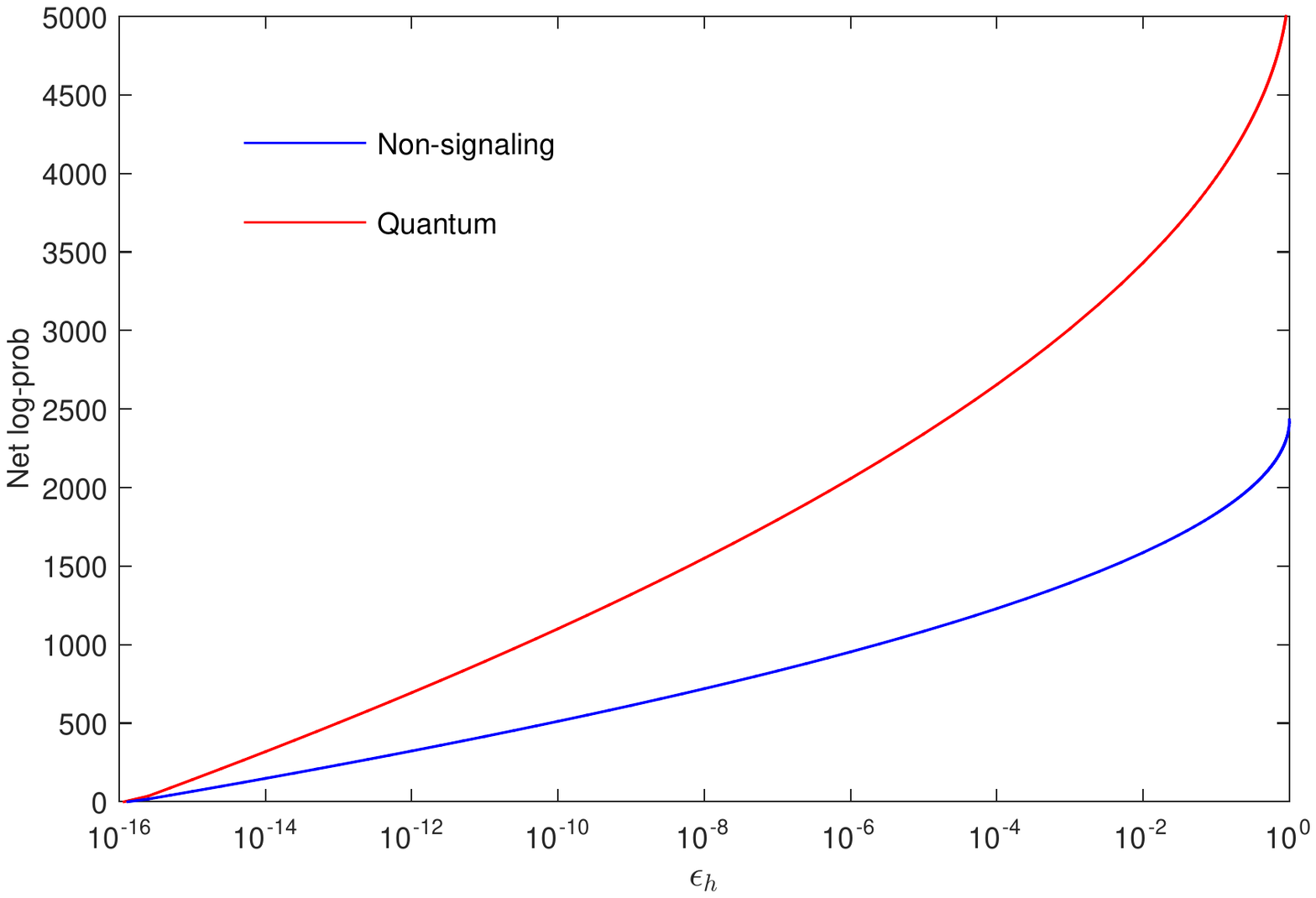}
  \end{center}
  \caption{Expected net $\textrm{log}_2$-probs for $27683$
    i.i.d.~trials at $\rho_{\textrm{atoms}}$ as a function of error
    bound. The number of trials is as reported in
    Ref.~\cite{rosenfeld:qc2016a} for the state
    $|\Psi^{-}\rangle$. }
  \label{fig:rosenfeld_lp(eps)}
\end{figure}

We determined the robustness of the asymptotic gain rates and the
$\textrm{log}_2$-prob rates against biases in the settings random
variables $X$ and $Y$, with the settings distribution constrained by
$\cC_{Z,b}$ as explained in the previous section.  The bias is bounded
but may vary within these bounds from trial to trial.
Fig.~\ref{fig:rosenfeld_agr(bias)} shows how the asymptotic gain rate
depends on bias. While the maximum bias tolerated in this case is
less than $0.05$, based on the results of Ref.~\cite{puetz:qc2014b},
we expect that there are quantum states and measurement settings for
which all biases strictly below $1$ can be tolerated in the ideal
case. See also Ref.~\cite{kessler:qc2017a}.
Fig.~\ref{fig:rosenfeld_lp(beta,bias)} shows the dependence of the
$\textrm{log}_2$-prob rate on power $\beta$ at three representative
biases.

\begin{figure}
  \begin{center}
    \includegraphics[scale=0.8,viewport=5.5cm 9cm 17cm 20.5cm]{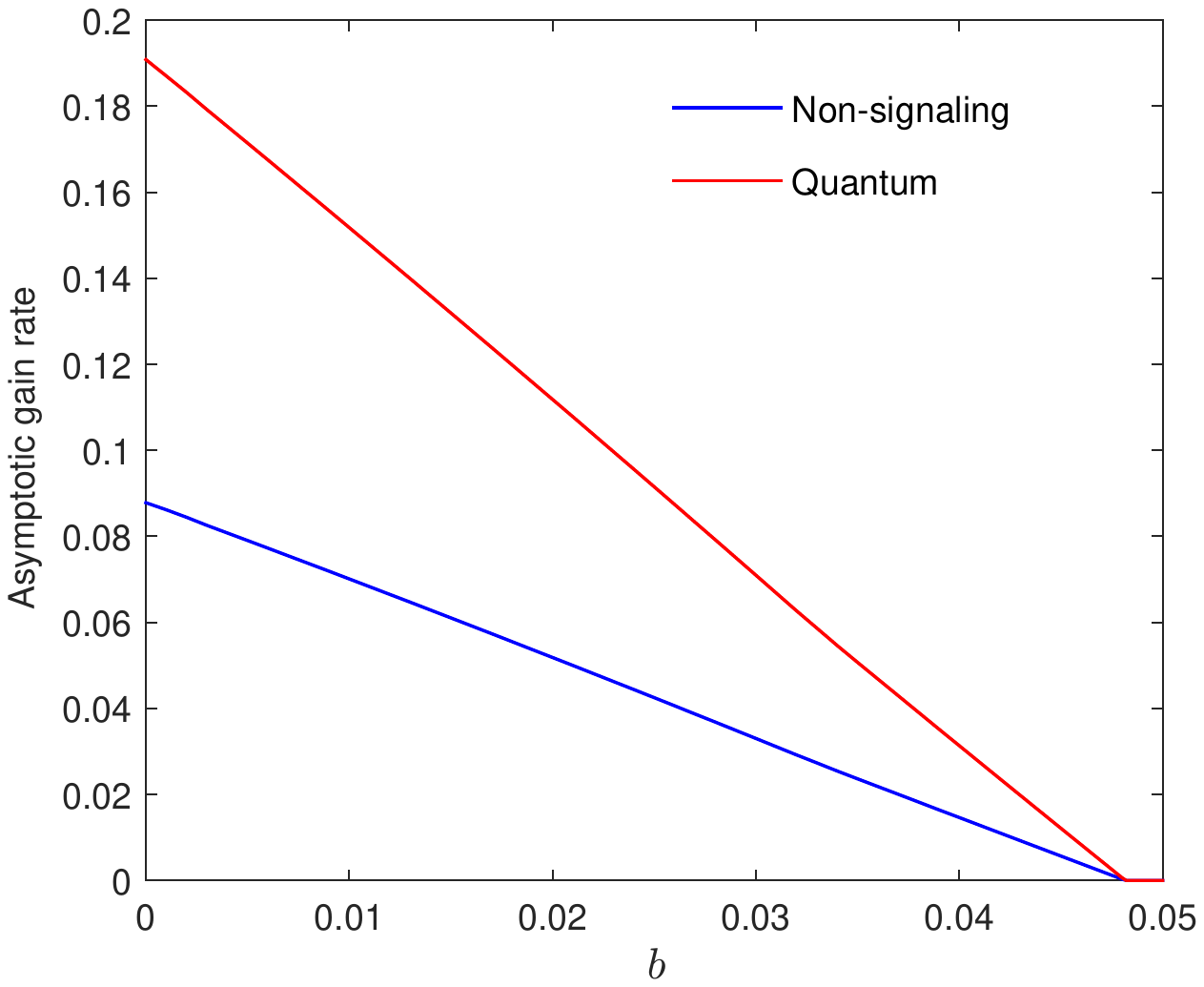}
  \end{center}
  \caption{Asymptotic gain rate (base $2$) as a function of bias at
    $\rho_{\textrm{atoms}}$.  This is the $\textrm{log}_2$-prob rate
    achieved for power $\beta\rightarrow 0_{+}$.  The probability
    distribution of the settings $Z=XY$ at each trial is in the convex
    closure of the distributions where $X$ and $Y$ are independent and
    their probabilities lie in the interval $[(1-b)/2,(1+b)/2]$. Note
    that except for these convex constraints, the probabilities can
    vary arbitrarily from trial to trial.}   
  \label{fig:rosenfeld_agr(bias)}
\end{figure}

\begin{figure}
  \begin{center}
    \includegraphics[scale=0.8,viewport=5.5cm 8cm 17cm 20.5cm]{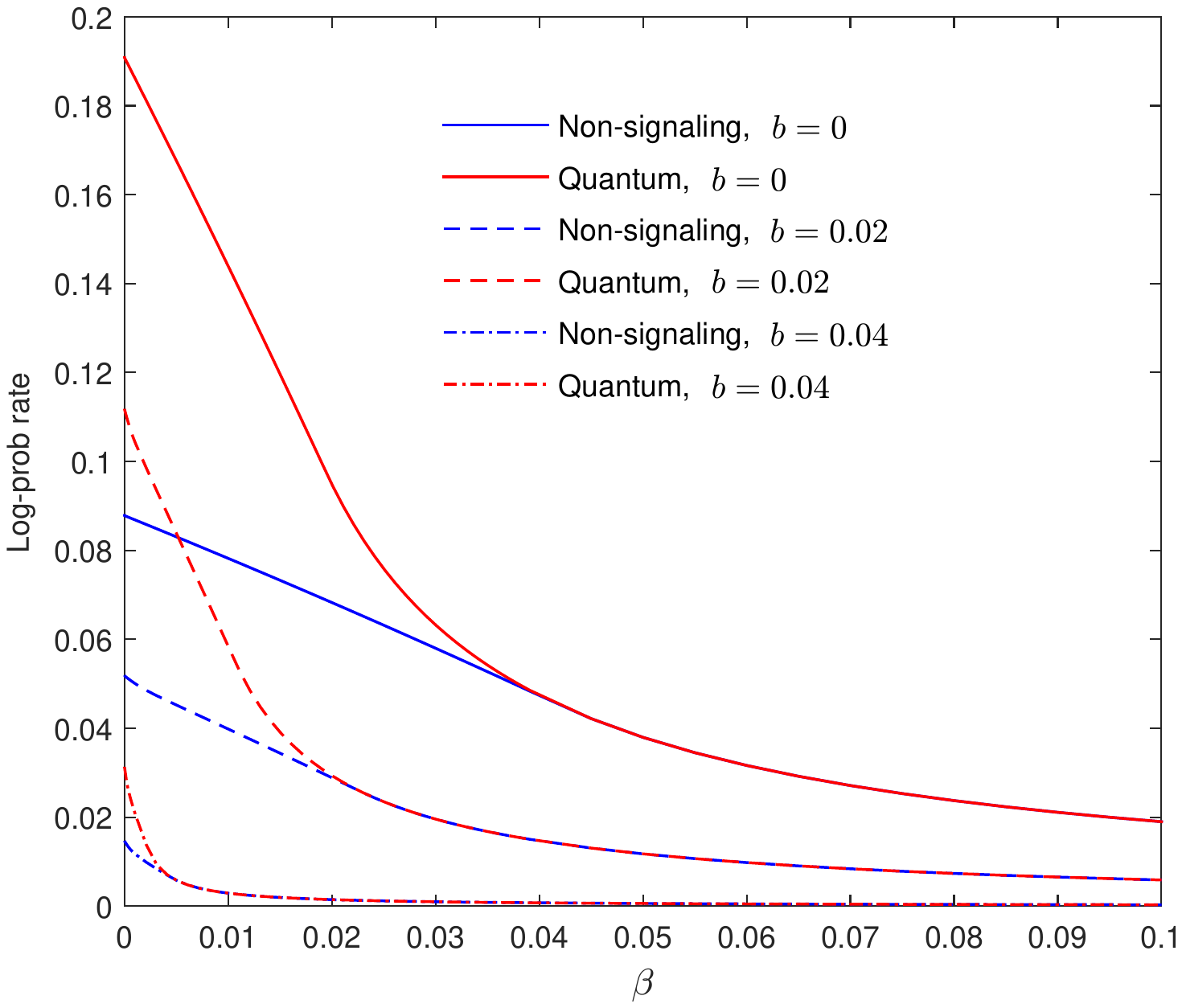}
  \end{center}
  \caption{$\textrm{Log}_2$-prob rates for representative biases at
    $\rho_{\textrm{atoms}}$. Here we show the $\textrm{log}_2$-prob rate
    as a function of power $\beta$ at three representative biases.
    See also Fig.~\ref{fig:rosenfeld_agr(bias)}.}
  \label{fig:rosenfeld_lp(beta,bias)}
\end{figure}

Finally, we consider the challenge of producing more random bits than
are consumed. This requires using a strategy to minimize the entropy
used for settings.  We assume i.i.d.~trials, with settings distributed
according to
\begin{equation}
  \nu_{r}(xy)=(1-r)\knuth{xy=11}+r/4,
  \label{eq:spot_check_prob}
\end{equation}
which uses uniform settings distribution with probability $r$ and 
a default setting $11$ with probability $1-r$.  Let
$S(r)=-(3r/4)\log_{2}(r/4)-(1-3r/4)\log_{2}(1-3r/4)$ be the base-$2$
entropy of $\nu_{r}$, $\sigma(\beta)$ the $\textrm{log}_2$-prob rate for a
given PEF with power $\beta$ at $\rho_{\textrm{atoms}}$, and
$\epse =2^{-\kappa}$ the desired error bound parameterized in
terms of the $\kappa$ and independent of $n$.  If we optimistically
assume that all the available min-entropy is extractable, then the
expected net number of random bits after $n$ trials is given by
\begin{equation}
  \sigma_{\textrm{net}}(n,\kappa,\beta,r)=n\sigma(\beta)-\kappa/\beta-n S(r).
\end{equation}
We refer to $\sigma_{\textrm{net}}$ as the expected \emph{net entropy}.  At given
values of $n$ and $\kappa$, we can maximize the expected net entropy over
$\beta>0$ and $r>0$.  Randomness expansion requires that
$\sigma_{\textrm{net}}\geq 0$.  Here we do not account for the seed
requirements or the extractor constraints.
Fig.~\ref{fig:rosenfeld_sigmanet} shows the expected net entropy as a function
of $n$ at $\rho_{\textrm{atoms}}$.  Of interest here
are the \emph{break-even} points, which are where the expected net entropy becomes
positive.  For this, consider the break-even equation
$n\sigma(\beta)-\kappa/\beta-n S(r)=0$. Equivalently,
$\kappa=\beta(\sigma(\beta)-S(r))n$, and to minimize the number of trials $n$
required for break-even, it suffices to maximize $\beta(\sigma(\beta)-S(r))$,
independently of $\kappa$.  (This independence motivates our choice of
parametrization of $\epse$.)  Given the maximizing values of $\beta$
and $r$, the critical value for $n$ at the break-even point is
expressed as
\begin{equation}
  n_c=\frac{\kappa}{\beta(\sigma(\beta)-S(r))}.
\end{equation}
One can see that $n_c$ scales linearly with $\kappa$.  Following this
strategy, we obtained the values of the parameters at the
break-even points.  They are shown in Table~\ref{tab:rosenfeld_sigmanet}.

\begin{figure}
  \begin{center}
    \includegraphics[scale=0.8,viewport=5.5cm 8cm 17cm 20.5cm]{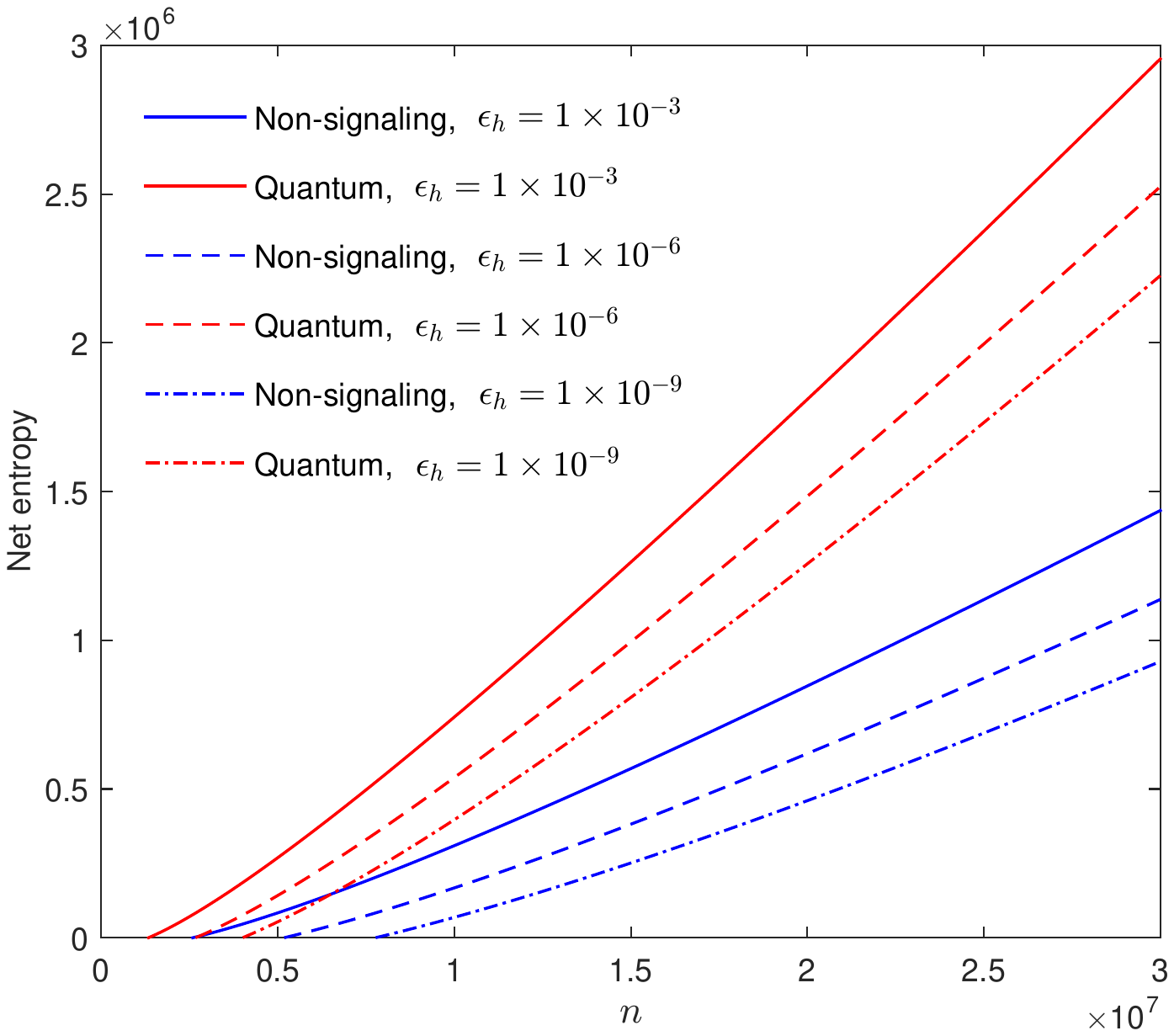}
  \end{center}
  \caption{Expected net entropy (in bits) at $\rho_{\textrm{atoms}}$. We
    optimized the power $\beta$ and the probability $r$ according to
    Eq.~\eqref{eq:spot_check_prob} given the number of trials $n$ and
    the error bound $\epse $ as explained in the text.  The break-even
    points are where the curves cross zero.  The parameters for those
    points are in Table~\ref{tab:rosenfeld_sigmanet}.}
  \label{fig:rosenfeld_sigmanet}
\end{figure}

\begin{table}
  \caption{Parameters for the break-even points for randomness expansion
    with $\rho_{\textrm{atom}}$.
    These are the parameters for the points where the curves in
    Fig.~\ref{fig:rosenfeld_sigmanet} cross zero.
  } 
  \label{tab:rosenfeld_sigmanet}
  \begin{equation}
    \begin{array}{|l|l|l|l|l|l|l|}
      \hline
      &\epse&n_c&\beta&r&\sigma\\
      \hline\hline
      \cN_{C|Z}&
      1\times 10^{-3}&2588821&&&\\
      &1\times 10^{-6}&5177642&1.3624\times 10^{-4}&3.7451\times 10^{-3}&0.060561\\
      &1\times 10^{-9}&7766462&&&\\
      \hline
      \cQ_{C|Z}&
      1\times 10^{-3}&1333781&&&\\
      &1\times 10^{-6}&2667562&1.4945\times 10^{-4}&7.3516\times 10^{-3}&0.108035\\
      &1\times 10^{-9}&4001343&&&\\

      \hline
    \end{array}
    \notag
  \end{equation}  
\end{table}

\subsection{Randomness from an Optical Loophole-Free Bell Test}

Ref.~\cite{bierhorst:qc2017a} describes the generation of $256$ bits
of randomness within $0.001$ of uniform from the optical loophole-free
Bell test of Ref.~\cite{shalm:2015}, assuming only non-signaling
constraints for certification. The data set used is labelled XOR 3
(pulses 3-9) and exhibits a $p$-value against LR of $2.03\times
10^{-7}$ according to the analysis in Ref.~\cite{shalm:2015}. It
consists of about $1.82\times 10^{8}$ trials with a low violation of
LR per trial. The protocol of Ref.~\cite{bierhorst:qc2017a} was
developed specifically to obtain randomness from Bell-test data with
little violation per trial. It implicitly involves probability
estimation. Table~\ref{tab:krister_dist} shows the closest 
distribution in $\cQ_{C|Z}$ inferred from the training set consisting
of the first $50\times 10^{6}$ trials of XOR 3.

\begin{table}
  \caption{Settings-conditional distribution inferred for the XOR 3 training data. The training set consists of the first $50\times 10^{6}$ trials.
    The table is organized as described in the caption of Table~\ref{tab:rosenfeld_dist}.}
  \label{tab:krister_dist}
  \begin{equation}
    \begin{array}{|ll||l|l|l|l|}
      \hline
      &ab&00&10&01&11\\
      xy&&&&&\\
      \hline\hline
      00&&0.999596756631154 &  0.000106695746779 &  0.000100495174505 &  0.000196052447562
      \\
      10&& 0.999039892488787 &  0.000663559889146 &  0.000086780398739 &  0.000209767223328
      \\
      01&&0.998967962694884  &  0.000089505202208 &  0.000729289110776 &  0.000213242992132
      \\
      11&& 0.998187653168081 &  0.000869814729010 &  0.000939019719445 &  0.000003512383464
      \\
      \hline
    \end{array}\notag
  \end{equation}
\end{table}

The analysis in Ref.~\cite{bierhorst:qc2017a} left open the question
of how much min-entropy one can certify in XOR 3 given bias in the
settings choices. As reported in Ref.~\cite{shalm:2015}, small biases
were present in the experiment, due to fluctuations in temperature
that affected the physical random sources used and issues with the
high-speed synchronization electronics. A plot of the net
$\textrm{log}_2$-prob achieved as a function of error bound at a few
representative biases is shown in Fig.~\ref{fig:peters_logprob}.  We
conclude that randomness certification from the XOR 3 data set is
robust against biases substantially larger than the average ones
inferred in Ref.~\cite{shalm:2015}. The robustness is similar to that
reported in Ref.~\cite{shalm:2015} for the $p$-values against LR.  To
get the results in Fig.~\ref{fig:peters_logprob}, we determined the
optimal PEFs and powers according to the conditional distribution in
Table~\ref{tab:krister_dist}. Then we applied the PEFs obtained to the
remaining trials (the \emph{analysis} trials).  Here, we did not adapt
the PEFs while processing the analysis trials.

We also determined the asymptotic gain rates from the conditional
distribution in Table~\ref{tab:krister_dist}. The dependence of the
asymptotic gain rate on the bias parameter $b$ is shown in
Fig.~\ref{fig:gain-bias-xor3}.  Notably, the robustness against bias is
similar to that for $\rho_{\textrm{atoms}}$ shown in
Fig.~\ref{fig:rosenfeld_agr(bias)}. This may be due to both
experiments being designed for the same family of Bell inequalities.
For higher Bell-test bias tolerance,
different inequalities and measurement settings need to be
used~\cite{puetz:qc2014b}.

\begin{figure}
  \begin{center}
    \includegraphics[scale=0.8,viewport=5.5cm 8cm 17cm 20.5cm]{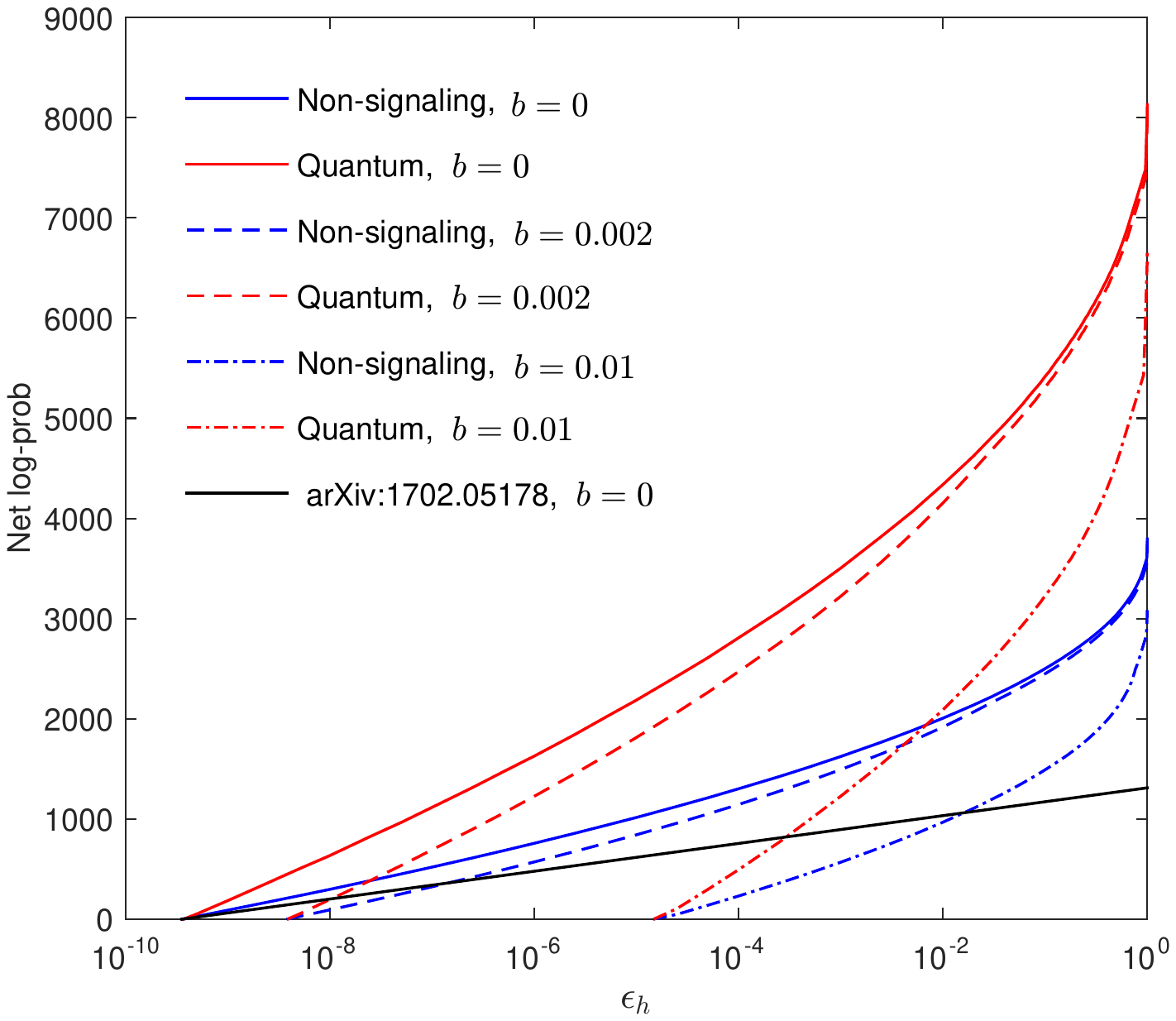}    
  \end{center}
  \caption{Net $\textrm{log}_2$-probs achieved in the XOR 3 data set
    from Ref.~\cite{shalm:2015}. This is the net $\textrm{log}_2$-prob
    achieved by probability estimation applied to the analysis set
    after determining the best PEF and power according to the
    conditional distribution inferred from the training set.  The
    error bound is an input parameter. The curves show the achieved
    net $\textrm{log}_2$-probs for non-signaling ($\cN_{C|Z}$) and quantum
    ($\cQ_{C|Z}$) constraints and three representative biases.
    Here, we performed probability estimation according to
    Eq.~\ref{eq:upe_max}. The curve that is lowest on the right is
    the net $\textrm{log}_2$-prob reported in
    Ref.~\cite{bierhorst:qc2017a}, where only non-signaling
    constraints were exploited. }
  \label{fig:peters_logprob}
\end{figure}

\begin{figure}
  \begin{center}
    \includegraphics[scale=0.8,viewport=5.5cm 8cm 17cm 20.5cm]{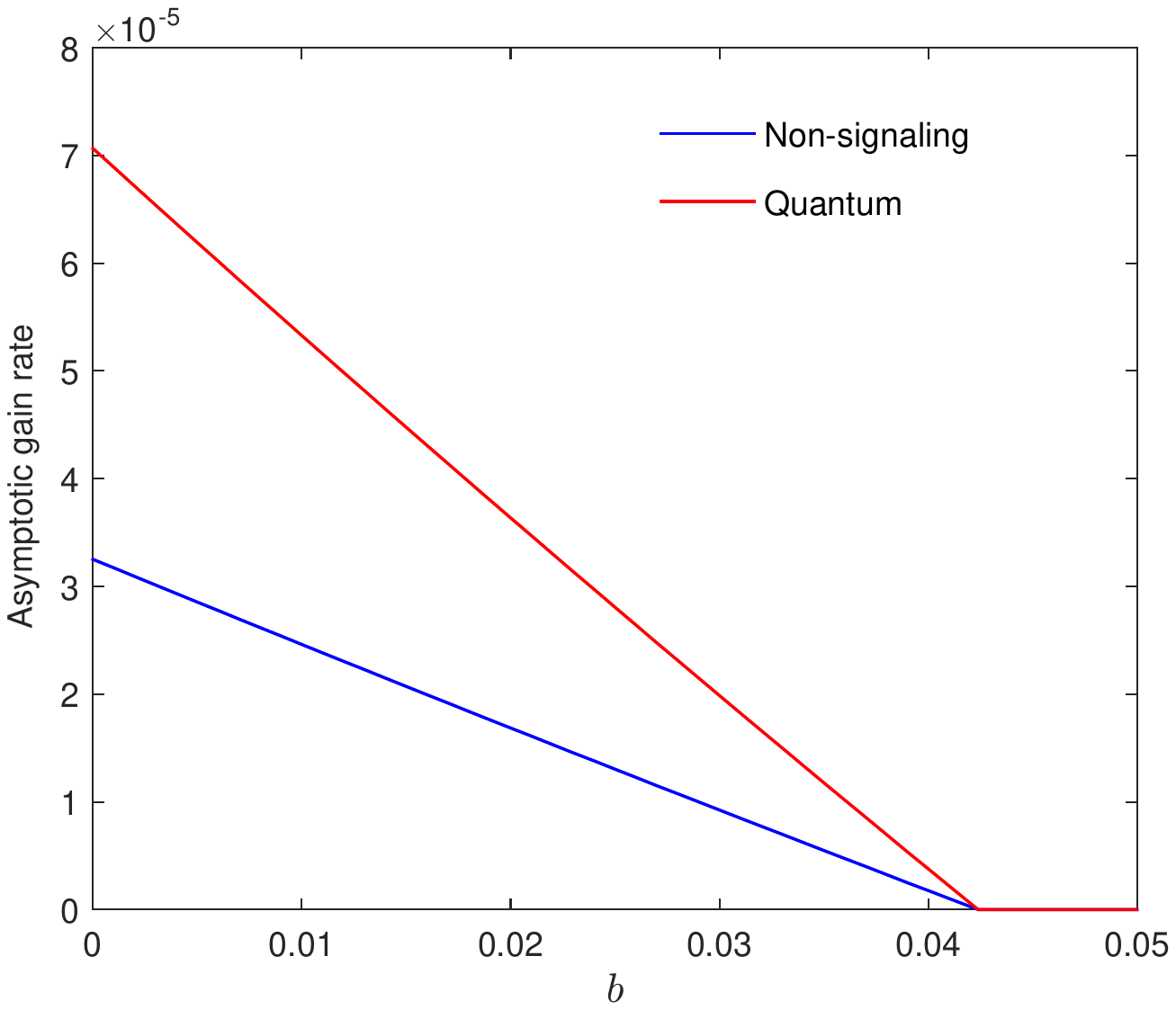}    
  \end{center}
  \caption{Asymptotic gain rate (base $2$) as a function of bias estimated from
    the training data in the XOR 3 data set in Ref.~\cite{shalm:2015}
    for non-signaling and quantum constraints. }
  \label{fig:gain-bias-xor3}
\end{figure}

\subsection{Reanalysis of ``Random Numbers Certified by Bell's Theorem~\cite{pironio:2010}''}

Finally, we applied the procedure that we used in reanalyzing the XOR
3 data set in the previous section to the data from the experiment
reported in Ref.~\cite{pironio:2010}. This experiment involved two
ions separated by about $1$ meter in two different ion traps. As a
Bell test, the experiment closed the detection loophole, 
but was subject to the locality loophole. The authors certified a
min-entropy of $42$ bits with respect to classical side information at an error
bound of $0.01$ with $3016$ trials. They assumed quantum constraints 
(implemented by the NPA hierarchy) for this certification. For the 
reanalysis, we set aside a training
set consisting of the first $1000$ trials to estimate
a constraints-satisfying conditional distribution. The distribution
obtained is given in Table~\ref{tab:pironio_dist}.  We optimized PEFs
and their powers as described in the previous section and applied them
to the remaining 2016 trials, the analysis set.  The result is shown in
Fig.~\ref{fig:Maryland_logprob}.  We expect that if it had been
possible to optimize the PEFs based on the calibration experiments
preceding the $3016$ trials, then the $\textrm{log}_2$-probs obtained would have
been approximately $50\,\%$ larger, assuming that the trial statistics
were sufficiently stable.

\begin{table}
  \caption{Settings-conditional distribution inferred for
    the experiment reported in Ref.~\cite{pironio:2010}. 
    The distribution is determined from the training set consisting
    of the first  $1000$ trials.
    The table is organized as described in the caption 
    of Table~\ref{tab:rosenfeld_dist}.}
  \label{tab:pironio_dist}
  \begin{equation}
    \begin{array}{|ll||l|l|l|l|}
      \hline
      &ab&00&10&01&11\\
      xy&&&&&\\
      \hline\hline
      00&&  0.395306466091468 & 0.117610486235535   &  0.093816274721118  &  0.393266772951878
      \\
      10&&  0.385009648861242 &  0.101263041214040  &  0.104113091951344  &  0.409614217973373
      \\
      01&&  0.411397337408393 &  0.101519614918611  &  0.097960367844259  &  0.389122679828738
      \\
      11&&  0.077378395334667 &  0.408894294740616  &  0.431979309917985  &  0.081748000006733
      \\
      \hline
    \end{array}\notag
  \end{equation}
\end{table}

\begin{figure}
  \begin{center}
    \includegraphics[scale=0.8,viewport=5.5cm 8cm 17cm 20.5cm]{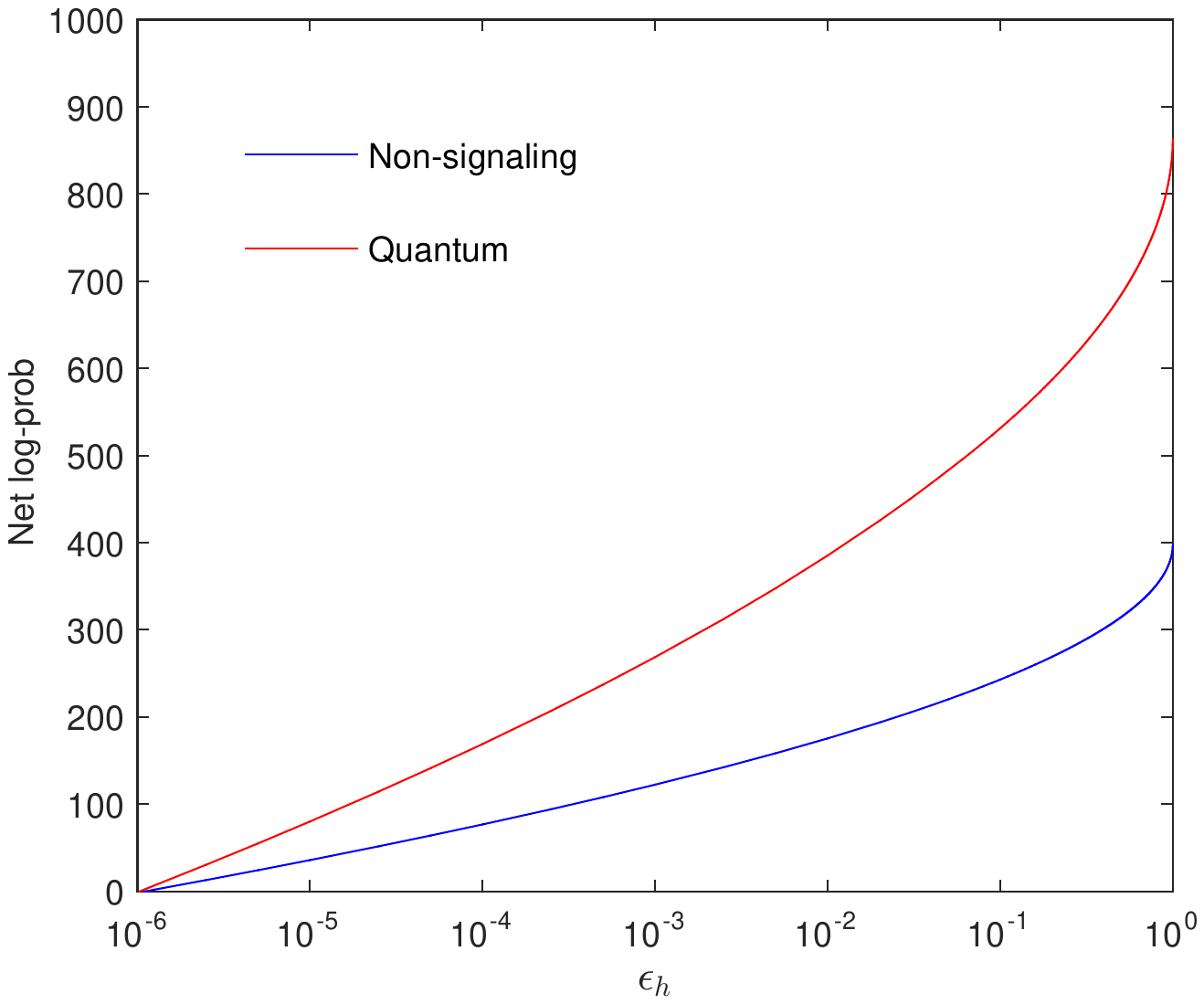}    
  \end{center}
  \caption{Net $\textrm{log}_2$-probs achieved in the data from
    Ref.~\cite{pironio:2010}. The analysis set consisted of the last
    $2016$ trials. Results for non-signaling and quantum constraints
    are shown. For these curves, we performed probability
    estimation according to Eq.~\ref{eq:upe_max}. The reference
    reported $42$ bits at $\epse =10^{-2}$ for quantum constraints.}
  \label{fig:Maryland_logprob}
\end{figure}

\begin{acknowledgments}
  We thank D. N. Matsukevich for providing the experimental data for
  Ref.~\cite{pironio:2010}, Yi-Kai Liu and Carl Miller for help with
  reviewing this paper, and Scott Glancy for discussions. This work
  includes contributions of the National Institute of Standards and
  Technology, which are not subject to U.S. copyright.
\end{acknowledgments}

\bibliography{xbytm}

{\small
  \begin{description}[\compact]
  \item[]\textbf{arXiv revision notes:}
  \item[]
    \begin{description}[\compact]
    \item[V1.] Original submission.
    \item[V2.] First revision.
    \item[]
      \begin{description}[\compact]
      \item[1.] Added missing lower bound on $p\delta$ in the
        statement of Thm.~\ref{thm:prot_chainlemmas}.  Corrected
        missing factor of $\delta$ in the proof of this theorem.
      \item[2.] Added comments in
        Sect.~\ref{sec:minentropy_extraction} to clarify that we do
        not need the extractor to be classical-proof or quantum-proof,
        but the TMPS extractor has these additional properties. We
        could have, but did not roll back the changes to the
        extractor constraints required for these properties.
      \item[3.]  Corrected the inclusion of sets of
        settings-distribution after Eq.~\ref{eq:cbvbdef} to
        $\cC_{Z,b}\subseteq \cB_{Z,2b+b^{2}}$.
      \item[4.] Changed the statement and proof of
        Thm.~\ref{thm:gainrate_optimality}, which implies optimality
        of the asymptotic gain rates. The proof's bounding hyperplane
        argument was corrected.
      \item[5.] Added comparison of the performance of exponential expansion
        to the literature after the statement of Thm.~\ref{thm:expexp}.
      \item[6.] Added a reference to Tsirelson's bound in
        Sect.~\ref{sec:apps}
      \item[7.] Changed the power $\beta$ to be positive.
      \item[8.] Added a summary section~\ref{s:somr} to highlight the main results.
      \item[9.] Editorial corrections throughout.
      \end{description}
    \end{description}
  \end{description}
}

\end{document}

%% file: xbytm_defs.tex
\providecommand{\ignore}[1]{}

\newif\ifcmnt
\cmnttrue
\ifdefined\cmntsoff\cmntfalse\fi

\ifcmnt
    \providecommand{\aucmnt}[1]{#1}

    \providecommand{\Pcolor}{\color{gray}}
\else
    \providecommand{\aucmnt}[1]{}

    \providecommand{\Pcolor}{}
\fi

\newcommand{\Pc}[1]{\aucmnt{{\Pcolor \textbf{[}Permanent comment: {\color{gray}#1}\textbf{]}}}}

\newcommand{\Prob}{\mathbb{P}}
\newcommand{\Exp}{\mathbb{E}}
\newcommand{\Var}{\mathrm{Var}}
\newcommand{\Probv}{\mathbb{P}}

\newcommand{\LR}{\mathrm{LR}}
\newcommand{\TV}{\mathrm{TV}}
\newcommand{\TMPS}{\mathrm{TMPS}}
\newcommand{\UPE}{\mathrm{UPE}}
\newcommand{\Unif}{\mathrm{Unif}}
\newcommand{\knuth}[1]{\left\llbracket #1 \right\rrbracket}

\newcommand{\epse}{\epsilon_{h}}
\newcommand{\epsx}{\epsilon_{x}}
\newcommand{\sige}{\sigma_{h}}

\newcommand{\nats}{\mathbb{N}}

\newcommand{\Rng}{\mathrm{Rng}}
\newcommand{\Supp}{\mathrm{Supp}}

\newcommand{\cCvx}{\overline{\mathrm{Cvx}}}
\newcommand{\xtrm}[1]{\mathrm{Extr}(#1)}

\newcommand{\Sfnt}[1]{\mathbf{#1}} 
\newcommand{\Pfnt}[1]{\textsf{#1}} 

\newcommand{\qvbar}{{|}}
\newcommand{\qrangle}{\rangle}
\newcommand{\qlangle}{\langle}
\newcommand{\ket}[1]{\qvbar{#1}\qrangle}
\newcommand{\bra}[1]{\qlangle{#1}\qvbar}

\newcommand{\ketbra}[2]{\ket{#1}\bra{#2}}

\newcommand{\cB}{\mathcal{B}}
\newcommand{\cC}{\mathcal{C}}
\newcommand{\cD}{\mathcal{D}}
\newcommand{\cE}{\mathcal{E}}
\newcommand{\cale}{\mathit{x}}
\newcommand{\cH}{\mathcal{H}}

\newcommand{\cM}{\mathcal{M}}
\newcommand{\cN}{\mathcal{N}}
\newcommand{\cO}{\mathcal{O}}
\newcommand{\cP}{\mathcal{P}}
\newcommand{\cQ}{\mathcal{Q}}
\newcommand{\cS}{\mathcal{S}}
\newcommand{\cU}{\mathcal{U}}
\newcommand{\cV}{\mathcal{V}}
\newcommand{\cX}{\mathcal{X}}

\newtheorem{theorem}{Theorem}
\newtheorem*{theorem*}{Theorem}
\newtheorem{lemma}[theorem]{Lemma}
\newtheorem*{lemma*}{Lemma}
\newtheorem{corollary}[theorem]{Corollary}
\newtheorem{definition}[theorem]{Definition}
\newtheorem*{definition*}{Definition}
